\documentclass[runningheads]{llncs}
\usepackage{amsmath}
\usepackage{amssymb}

\usepackage{graphicx}
\usepackage{xcolor}
\usepackage{listings}
\usepackage{colortbl}
\usepackage{booktabs}
\usepackage{dcolumn}
\newcolumntype{d}[1]{D{.}{.}{#1} }
\newcolumntype{s}[1]{D{/}{/}{#1} }

\usepackage[bookmarksdepth=1]{hyperref}
\usepackage{xspace}
\newcommand\peqcheck{{\scshape{PEQ\-check}}\xspace}
\newcommand\cpachecker{{\scshape{CPA\-checker}}\xspace}

\spnewtheorem*{theorem1}{Theorem~1}{\normalfont\bfseries}{\itshape}
\spnewtheorem*{theorem2}{Theorem~2}{\normalfont\bfseries}{\itshape}
\spnewtheorem*{theorem3}{Theorem~3}{\normalfont\bfseries}{\itshape}
\spnewtheorem*{theorem3b}{Theorem~3}{\normalfont\bfseries}{\itshape}

\begin{document}
%
% Originally
%\title{\peqcheck: Checking Functional Equivalence\\ for Parallelized Programs\thanks{This work was funded by the Hessian LOEWE initiative within the Software-Factory 4.0 project.}}
\title{\peqcheck: Localized and Context-aware Checking of Functional Equivalence ~~~~~~(Technical Report)}%\thanks{This work was funded by the Hessian LOEWE initiative within the Software-Factory 4.0 project.}}

\titlerunning{Localized and Context-aware Checking of Functional Equivalence}
% If the paper title is too long for the running head, you can set
% an abbreviated paper title here
%
\author{Marie-Christine Jakobs%\inst{1}\orcidID{0000-1111-2222-3333} \and
}
\authorrunning{M.-C. Jakobs}

\institute{Technical University of Darmstadt, Department of Computer Science,\\ Darmstadt, Germany}
\maketitle              % typeset the header of the contribution

\begin{abstract}
Refactorings must not alter the program's functionality.
However, not all refactorings fulfill this requirement.
Hence, one must explicitly check that a refactoring does not alter the functionality.
Since one rarely has a formal specification of the program's behavior, we utilize the original program as functional specification.
Then, we check whether the original and refactored program are \emph{functionally equivalent}.
To this end, we apply a common idea and reduce equivalence checking to program verification.
To increase efficiency, our equivalence checker \peqcheck constructs one verification task per refactored code segment instead of one per function as typically done by prior work.
In addition, \peqcheck considers the context of the code segments.
For instance, only variables that are modified and live are required to be equivalent and read-only variables may be shared between original and refactored code segments.
We show that \peqcheck is sound.
Moreover, our evaluation testifies that the localized and context-aware checking performed by \peqcheck can indeed be beneficial.

\end{abstract}

\section{Introduction}
Developers perform refactoring~\cite{RefactoringSurvey,RefactoringFowler} to improve the quality of their software, e.g., the software's performance.
To improve the software's performance, one may parallelize execution hot spots, e.g., using OPenMP~\cite{OpenMP}.
Indeed, code parallelization with OpenMP is the motivation for our work. 
While parallelization aims at improving the software's performance and in general a refactoring aims at improving the software's quality, the refactoring must ensure that the software's functionality is not altered.
To prevent that a refactoring inadvertently changes the software's functionality, a verification of the refactored software should check that the software's functionality is preserved.

Various approaches exist that aim to guarantee that a refactored program preserves the functionality.
One approach that only works for (semi-)automatic refactorings is to prove the correctness of the applied refactoring rules~\cite{SCAMVerifyRefactoring,DBLP:conf/pepm/SultanaT08,ParameterizedEQCheck}.
In industry, regression testing~\cite{Test} is used, but testing typically does not examine all program paths and, thus, may miss regressions.
An alternative to testing is formal software verification~\cite{SWVerificationSurvey}.
Incremental and regression verification techniques~\cite{RegresionMC,ExtremeMC,EvolCheck,PrecisionReuse,DiSE,RV-DAC} propose solutions to efficiently re-verify modified programs.
However, many of those techniques rely on a specification of the functional behavior, which is rarely available.
In contrast, regression verification techniques that check the functional equivalence of the original and refactored software do not require a specification.

Different approaches exist to check functional equivalence of two programs (or functions).
For example, one can apply relational program verification~\cite{DBLP:conf/popl/Benton04,DBLP:journals/tcs/Yang07,ProductPrograms}, establish a (bi)simulation relation~\cite{CoVaC,DBLP:conf/oopsla/0001SCA13,DBLP:conf/aplas/DahiyaB17,DBLP:conf/pldi/ChurchillP0A19}, translate the programs into models and show model equivalence~\cite{PromelaFEQ,TASS,PRESGen,DBLP:conf/cav/VerdoolaegeJB09,DBLP:conf/date/ShashidharBCJ05}, compute symbolic summaries and check if the summaries are equivalent~\cite{DiffSymExe,ARDiff}, translate the equivalence problem into a Horn constraint problem~\cite{Reve}, or combine program generation with verification~\cite{RV-DAC,RV-Journal,UCKLEE,SymDiff,RIE,RV-thread,CIVL,DBLP:journals/ijpp/AbadiKPV19}. 
The last solution translates the equivalence problem into a program verification task  (a program with assertions) and uses a standard verifier to prove the verification task.
Since this solution is independent of the proof technique, it directly profits from existing verification technologies and their improvements.
This makes it particularly appealing and is one of the reasons why we want to use this approach to check functional equivalence of a sequential program and its OpenMP parallelization.

Unfortunately, most of the existing approaches~\cite{RV-DAC,RV-Journal,UCKLEE,SymDiff,RIE} that reduce functional equivalence checking to program verification focus on sequential programs and are unsound for parallel programs.
For example, they assume that a function returns the same result whenever it is called with the same inputs (including global variables).
Based on this assumption, the approaches replace function calls by uninterpreted functions. 
However, this assumption is no longer guaranteed when another thread interferes with the function execution.
While CIVL~\cite{CIVL} and RVT~\cite{RV-thread} support parallel programs, they perform equivalence checking on program or function level.
To reduce the complexity of equivalence checking, e.g., to reduce the state space that needs to be considered during verification, we aim at equivalence checking on the level of (parallelized) code segments.
Currently, only the approach of Abadi et al.~\cite{DBLP:journals/ijpp/AbadiKPV19} reduces equivalence checking to program verification and supports (parallelized) code segments.
Their approach only works if there exist a bijection between inputs of the code segments and a bijection between the outputs.
However, this assumption is unnecessarily strict, e.g., a bijection between inputs may prohibit that one of the code segments applies the strategy pattern.

To overcome this problem, we propose \peqcheck, a sound approach that generates verification tasks to check equivalence of code segments.
While motivated by OpenMP parallelization, \peqcheck cannot only check equivalence of sequential and parallelized code segment, but also supports equivalence checking of sequential code segments and equivalence checking of parallel code segments.
To determine the context of the code segments, \peqcheck utilizes dataflow analyses to find out how variables are used in and after the code segments.
Based on the context information, \peqcheck employs a fine-grained differentiation of variables and, thus, reduces the complexity of the generated verification task. 
For instance, variables that are not modified are shared, inputs are only equalized when they are used before they are written in the code segment, and equivalence checking is restricted to modified variables that are used after the code segments.
While \peqcheck is the first regression verification approach that uses such a fine-grained differentiation of variables, existing approaches use some of these optimizations..
For example, SymDiff~\cite{SymDiff} only checks equivalence of modified variables and
RVT~\cite{RV-DAC,RV-Journal} only initializes global variables that are written to by at least one of the programs.

We show soundness of our \peqcheck approach, implemented it in a prototype tool, and evaluated it on several examples.
Our evaluation testifies that \peqcheck detects non-equivalence and that \peqcheck's localized and context-aware checking can be beneficial.

This technical report is an extension of our conference paper~\cite{PEQcheck} and enhances our conference paper with the soundness proofs.
To be self-contained, the technical report presents all contributions of the conference paper~\cite{PEQcheck}.
\subsection{Illustration}\label{ssec:example}
We use the sequential and parallelized program shown on the left-hand side of Fig.~\ref{fig:example} to explain the idea of our \peqcheck approach.
Both programs, \texttt{sum2\_seq} and \texttt{sum2\_par}, iteratively sum up the first $N$ numbers and then add~2.
To check that \texttt{sum2\_seq} and \texttt{sum2\_par} are functionally equivalent, we inspect the equivalence of the two highlighted code segments. 
The verification task generated to inspect equivalence is shown on the right-hand side of Fig.~\ref{fig:example}.
In the following, we explain how to generate this task.

\begin{figure}[t]%
{
\begin{minipage}{0.48\textwidth}
\begin{lstlisting}[caption=Sequential program]
int sum2_seq(unsigned char N)
{
 int j, sum;
\end{lstlisting}
\vspace{-\baselineskip}
\begin{lstlisting}[backgroundcolor=\color{cyan!20}]
 sum = N;
\end{lstlisting}
\vspace{-1.2\baselineskip}
\begin{lstlisting}[backgroundcolor=\color{cyan!20}]
 for(j=N-1; j>=0; j--)
 {
  sum += j;
 }
\end{lstlisting}
\vspace{-\baselineskip}
\begin{lstlisting}
 return sum + 2;
}
\end{lstlisting}

%\vspace{0.25\baselineskip}
\begin{lstlisting}[caption=Parallelized program]
int sum2_par(unsigned char N)
{
 int i, sum;
\end{lstlisting}
\vspace{-0.5\baselineskip}
\begin{lstlisting}[backgroundcolor=\color{yellow!40}]
 sum = 0;
\end{lstlisting}
\vspace{-1.2\baselineskip}
\begin{lstlisting}[backgroundcolor=\color{yellow!40},basicstyle=\tiny]
#pragma omp parallel for reduction(+:sum)
\end{lstlisting}
\vspace{-1.2\baselineskip}
\begin{lstlisting}[backgroundcolor=\color{yellow!40}]
 for(i=1; i <= N; i++)
 {
  sum += i;
 }
\end{lstlisting}
\vspace{-\baselineskip}
\begin{lstlisting}
 return sum + 2;
}
\end{lstlisting}
\end{minipage}
\hfill
\begin{minipage}{0.48\textwidth}
\begin{lstlisting}[caption=Verification task]
int main()
{
 unsigned char N;
 int i, j, sum_s, sum;

 N = random_uchar();
\end{lstlisting}
\begin{lstlisting}[backgroundcolor=\color{cyan!20}]
 sum_s = N;
\end{lstlisting}
\vspace{-1.2\baselineskip}
\begin{lstlisting}[backgroundcolor=\color{cyan!20}]
 for(j=N-1; j>=0; j--)
 {
  sum_s += j;
 }
\end{lstlisting}
\begin{lstlisting}[backgroundcolor=\color{yellow!40}]
 sum = 0;
\end{lstlisting}
\vspace{-1.2\baselineskip}
\begin{lstlisting}[backgroundcolor=\color{yellow!40},basicstyle=\tiny]
#pragma omp parallel for reduction(+:sum)
\end{lstlisting}
\vspace{-1.2\baselineskip}
\begin{lstlisting}[backgroundcolor=\color{yellow!40}]
 for(i=1; i <= N; i++)
 {
  sum += i;
 }
\end{lstlisting}
\vspace{-0.25\baselineskip}
\begin{lstlisting}[language=C]	
 assert(sum_s == sum);
	
 return 0;
}
\end{lstlisting}
\end{minipage}
}
%\vspace{-2mm}
\caption{Example sequential program, its parallelization, and the generated verification task for equivalence checking (taken from \cite{PEQcheck})}%
\label{fig:example}%
%\vspace{-4mm}
\end{figure}

First, we determine the context of the code segments.
Therefore, we find out which variables are used in the code segments and how.
More concretely, we collect the variables $\mathcal{V}$ used by the code segments, determine which variables are modified~($\mathcal{M}$) in the code segment, which variables are used in the code segment before they are defined~($\mathcal{UB}$), and which variables are live after the code segment~($\mathcal{L}$).\footnote{Note that it is safe to ignore variables that are only defined in the scope of the code segment because the neither need to be declared nor initialized and cannot be live after the code segment because they cannot be accessed after the code segment.}
Four our example, we get $\mathcal{V}_\mathrm{seq}=\{j, N, sum\}$, $\mathcal{V}_\mathrm{par}=\{i, N, sum\}$, $\mathcal{M}_\mathrm{seq}=\{j, sum\}$, $\mathcal{M}_\mathrm{par}=\{i, sum\}$, $\mathcal{UB}_\mathrm{seq}=\mathcal{UB}_\mathrm{par}=\{N\}$, and $\mathcal{L}_\mathrm{seq}=\mathcal{L}_\mathrm{par}=\{sum\}$.
Based on this information, we then determine which variables (1)~may be shared, (2)~need to be declared, (3)~need to be equivalent, and (4)~whether and how to initialize the variables.

To decide this, we also need to relate the variables of the two segments.
So far, we relate variables by their name.
Thus, our approach fails if there exist variables with the same name, but different types.\footnote{One can overcome this limitation by providing the relation of the variables.} % in the two code segments.
However not that we allow both code segments to use additional (input) variables, e.g., the parallelized code segment uses additional variable $i$.

Given the sets $\mathcal{V}_\mathrm{seq}, \mathcal{V}_\mathrm{par}, \mathcal{M}_\mathrm{seq}$, and $ \mathcal{M}_\mathrm{par}$, we identify the shared variables.
This is important for code generation because variables that occur in both code segments (i.e., $\mathcal{V}_\mathrm{seq}\cap\mathcal{V}_\mathrm{par}$) and are not modified can be shared safely and our approach shares them.
In contrast, modified variables that occur in both programs must be duplicated.
We decided that the sequential code segment will use the duplicated variables.
Hence, our example shares common variable $N$ and duplicates variable $sum$.
We use $sum\_s$ for the duplicated variable.

%Now, we have everything at hand to build the verification task.
At last, we construct the verification task.
% declarations
At the beginning, a verification task declares the variables $\mathcal{V}=\mathcal{V}_\mathrm{seq}\cup\mathcal{V}_\mathrm{par}$ and the duplicated variables.
In our example, we declare variables $\{i, j, N, sum, sum\_s\}$.
% initialization
Thereafter, the verification task initializes variables non-deterministically when they are used before they are defined by the code segment.
More concretely, the task must initialize all variables \(\mathcal{UB}^*_\mathrm{seq}\cup\mathcal{UB}_\mathrm{par}\), where \(\mathcal{UB}^*_\mathrm{seq}\) is obtained from \(\mathcal{UB}_\mathrm{seq}\) by replacing duplicated variables by their duplicate. %we need to take into account that some variables will be renamed.
In case a variable and its duplicate must be initialized, the task will initialize the duplicate  with the same value as the original variable.
In our example verification task, we call function \texttt{random\_ushort} to non-deterministically initialize variable $N$.
After preparing the inputs, the verification task executes the sequential and parallelized code segment.
Thereby, the sequential segment uses the duplicated variables wherever necessary.
Finally, the task uses one assert statement per relevant output variable, which checks the equivalence of this variable and its duplicate.
The relevant output variables are all variables that (1)~are shared, (2)~are modified, and (3)~may be live afterward, i.e., the variables in the set %$(\mathcal{V}_\mathrm{seq}\cap\mathcal{V}_\mathrm{par})\cap(\mathcal{M}_\mathrm{seq}\cup\mathcal{M}_\mathrm{par})\cap(\mathcal{L}_\mathrm{seq}\cup\mathcal{L}_\mathrm{par})$.
$(\mathcal{M}_\mathrm{seq}\cup\mathcal{M}_\mathrm{par})\cap(\mathcal{L}_\mathrm{seq}\cup\mathcal{L}_\mathrm{par})$.
Our example contains one assert statement that inspects the equivalence of variables $sum$ and $sum\_s$.
\section{Programs} 
We present our approach on a simple imperative language on integer variables that excludes synchronization primitives because we do not study synchronization issues.
The following grammar describes our programs\footnote{Our implementation supports C programs with OpenMP pragmas for parallelization.}.
\[ \begin{array}{l l l} S &:=  & E ~|~ v:=_\ell aexpr; ~|~ \mathbf{assert}_\ell~bexpr; ~|~\\ && \mathbf{if}_\ell~bexpr~\mathbf{then}~ S_1~\mathbf{else}~S_2 ~|~ \mathbf{while}_\ell~bexpr~\mathbf{do}~ S ~|~ S_1;S_2 ~|~ [S_1 \| \dots \| S_n]\end{array}\]
Program~\(E\) denotes the empty program.
Arithmetic expressions~\(aexpr\) and boolean expressions~\(bexpr\) are assumed to be constructed by applying standard operators on integers.
Furthermore, subprograms~\(S_i\) can be composed to build more complex programs~\(S\).
Note that we annotate each basic statement with a label~\(\ell\), which is assumed to be unique in the complete program.
Thus, subprograms of a program can be identified unambiguously.

We use \(\mathcal{V}\) to denote the set of all program variables and subset~\(\mathcal{V}(S)\subseteq\mathcal{V}\) describes the variables of (sub)program \(S\), i.e., all variables that either occur in an arithmetic or boolean expression of \(S\) or occur on the left-hand side of an assignment in \(S\).
Similarly, subset~\(\mathcal{V}(expr)\subseteq\mathcal{V}\) denotes the variables that are used in expression~\(expr\).

To generate accurate verification tasks, \peqcheck renames certain program variables in the code segments of the sequential program.
For the sake of renaming, \peqcheck relies on a bijective, renaming function~\(\rho:\mathcal{V}\mapsto\mathcal{V}\) and replaces all occurrences of any variable~\(v\) by \(\rho(v)\).
The result of the replacement is the renamed program~\(\mathcal{R}(S, \rho)\).
Similarly, \(\mathcal{R}(expr, \rho)\) represents the renaming of \(expr\).
For the example in Fig.~\ref{fig:example}, we use the renaming function \(\rho_\mathrm{sum2}\), where \(\rho_\mathrm{sum2}(sum)=sum\_s\), \(\rho_\mathrm{sum2}(sum\_s)=sum\), and \(\rho_\mathrm{sum2}(v)=v\) otherwise.
%Thus, \(\mathcal{R}(sum+j, \rho_\mathrm{sum2})=sum\_s+j\)

\begin{figure}[t]
	\centering
\[\begin{array}{l}
\frac{}{(v:=_\ell aexpr;, \sigma)\xrightarrow{v:=aexpr;}(E,\sigma[v:=\sigma(aexpr)])} \quad \frac{\sigma(bexpr)=\mathrm{true}}{(\mathbf{assert}_\ell~bexpr;, \sigma)\xrightarrow{bexpr}(E,\sigma)}\quad\frac{}{(E;S,\sigma)\xrightarrow{\textbf{nop}}(S, \sigma)} \\% \quad\frac{E;S}{S}\\
\\
\frac{\sigma(bexpr)=\mathrm{true}}{( \mathbf{if}_\ell~bexpr~\mathbf{then}~ S_1~\mathbf{else}~S_2,\sigma)\xrightarrow{bexpr}(S_1,\sigma)} \quad
\frac{\sigma(bexpr)=\mathrm{false}}{( \mathbf{if}_\ell~bexpr~\mathbf{then}~ S_1~\mathbf{else}~S_2,\sigma)\xrightarrow{\neg bexpr}(S_2,\sigma)}\\
\\
\frac{\sigma(bexpr)=\mathrm{true}}{(\mathbf{while}_\ell~bexpr~\mathbf{do}~ S, \sigma)\xrightarrow{bexpr}(S;\mathbf{while}_\ell~bexpr~\mathbf{do}~ S, \sigma)} \quad
\frac{\sigma(bexpr)=\mathrm{false}}{(\mathbf{while}_\ell~bexpr~\mathbf{do}~ S, \sigma)\xrightarrow{\neg bexpr}(E,\sigma)}\\
\\
%\frac{(S_1,\sigma)\xrightarrow{op}(S'_1,\sigma'),~ S'_1\neq E}{(S_1;S_2,\sigma)\xrightarrow{op}(S'_1;S_2\sigma')} \hfill~
%\frac{(S_1,\sigma)\xrightarrow{op}(S'_1,\sigma'),~S'_1=E}{(S_1;S_2,\sigma)\xrightarrow{op}(S_2\sigma')}\hfill~
%\frac{(S_1,\sigma)\xrightarrow{op}(S'_1,\sigma')}{(E;S_1,\sigma)\xrightarrow{op}(S'_1\sigma')}\hfill~
%\frac{(S_1,\sigma)\xrightarrow{op}(S'_1,\sigma')}{([E \| \dots \| E];S_1,\sigma)\xrightarrow{op}(S'_1\sigma')}\\
%\frac{(S_i,\sigma)\xrightarrow{op}(S'_i,\sigma'),~S'_i\neq E\vee\exists j\in[1,n]:j\neq i\wedge S_j\neq E}{([S_1 \| \dots \| S_i\|\dots \| S_n],\sigma)\xrightarrow{op}([S_1 \| \dots \| S'_i\|\dots \| S_n],\sigma')}\hfill
%\frac{(S_i,\sigma)\xrightarrow{op}(S'_i,\sigma'), S'_i=E\wedge(\forall j\in[1,n]:j\neq i\vee S_j=E)}{([S_1 \| \dots \| S_i\|\dots \| S_n],\sigma)\xrightarrow{op}(E,\sigma')}\\
%\frac{[S_1\|\dots\| S_n], \forall 1\leq i\leq n: S_1\equiv E}{E}\\
\frac{(S_1,\sigma)\xrightarrow{op}(S'_1,\sigma')}{(S_1;S_2,\sigma)\xrightarrow{op}(S'_1;S_2, \sigma')} \hfill~
\hfill
\frac{(S_i,\sigma)\xrightarrow{op}(S'_i,\sigma')}{([S_1 \| \dots \| S_i\|\dots \| S_n],\sigma)\xrightarrow{op}([S_1 \| \dots \| S'_i\|\dots \| S_n],\sigma')}\hfill
\frac{}{([E \| \dots \| E],\sigma)\xrightarrow{\textbf{nop}}(E,\sigma)}\\

\end{array}
\]
	\caption{Rules for operational semantics}
	\label{fig:semRules}
\end{figure}

For the program semantics, we consider an operational semantics that defines a program's executions.
The semantics describes executions as transitions between execution states.
An execution state is a pair of a program plus a data state.
A data state \(\sigma: \mathcal{V}\mapsto \mathbb{Z}\) assigns an integer value to each variable.
As usual, we denote the set of all data states by \(\Sigma\) and write \(\sigma(expr)\) to denote the evaluation of \texttt{expr} in data state \(\sigma\in\Sigma\).
Furthermore, we define \(\rho(\sigma)\in\Sigma\) such that for all \(v\in \mathcal{V}: \rho(\sigma)(v)=\sigma(\rho^{-1}(v))\) and introduce \(\sigma=_{|_{V}}\sigma'\) to describe that the variables of subset \(V\subseteq \mathcal{V}\) are identical in states \(\sigma\in\Sigma\) and \(\sigma'\in\Sigma\), i.e., \(\forall v\in V: \sigma(v)=\sigma'(v)\).
Furthermore, for any \(\sigma, \sigma'\in \Sigma\) and any subset \(V\subseteq \mathcal{V}\), we write \(\sigma=_{|_{V}}\sigma'\) if for all \(v\in V: \sigma(v)=\sigma'(v)\).

The 10~rules shown in Fig.~\ref{fig:semRules} define the execution steps.
The state update~\(\sigma[v:=\sigma(aexpr)]\), which is used in the rule for the assignment, returns a new data state~\(\sigma'\) with \(\sigma'(w)=\sigma(w)\) for all \(w\in\mathcal{V}\) with \(w\neq v\) and \(\sigma'(v)=\sigma(aexpr)\).
Since we have not fixed the expression syntax, we also do not specify the expression evaluation.
However, our approach requires that expression evaluation (a)~is deterministic, (b)~only depends on the variables used in the expression, i.e., \mbox{\(\forall \sigma, \sigma'\in\Sigma:\)} \mbox{\(\sigma=_{|_{\mathcal{V}(expr)}}\sigma'\implies \sigma(expr)=\sigma'(expr)\)}, and (c)~is consistent with renaming, i.e., \(\sigma(expr)=\rho(\sigma)(\mathcal{R}(expr, \rho))\).
In addition, we assume that in all states~\(\sigma\in\Sigma\) an expression~$v$ that references variable~$v$ evaluates to the variable's value in state~\(\sigma\)  and that the equivalence of two variables \(v\) and \(v'\) (encoded as boolean expression \(v==v'\)) checks that their values are identical, i.e., \(\forall v, v'\in\mathcal{V}: \sigma(v == v')\implies\sigma(v)=\sigma(v')\).

Next, we inductively define the executions \(ex(S)\) of a program \(S\).
{\footnotesize\[\frac{\sigma\in\Sigma}{(S,\sigma)\in ex(S)}\quad
\frac{(S_0,\sigma_0)\xrightarrow{op_1}\dots\xrightarrow{op_n}(S_n,\sigma_n)\in ex(S), \quad (S_n,\sigma_n)\xrightarrow{op_{n+1}}(S_{n+1},\sigma_{n+1}) }{(S_0,\sigma_0)\xrightarrow{op_1}\dots\xrightarrow{op_n}(S_n,\sigma_n)\xrightarrow{op_{n+1}}(S_{n+1},\sigma_{n+1})\in ex(S)}\]}
We write \((S,\sigma)\rightarrow^*(S',\sigma')\) if we are not interested in the intermediate steps of the execution.
Furthermore, execution \((S,\sigma)\rightarrow^*(S',\sigma')\) (i)~terminates normally if \(S'=E\) and (ii)~violates an assertion if \(S'\) violates an assertion in~\(\sigma'\).
A program~\(S\) violates an assertion in state~\(\sigma\) if (a)~there exists an assert statement \(S_{a}=\textbf{assert}_\ell~bexpr\) whose assertion is violated in state \(\sigma\) (i.e., \(\sigma(bexpr)=false\)), and \(S=S_a\) or \(S=S_a;S'\) or (b)~\(S=[S_1 \| \dots \| S_i\|\dots \| S_n]\) or \(S=[S_1 \| \dots \| S_i\|\dots \| S_n];S'\) and there exists an \(S_i\) that violates an assertion in \(\sigma\).

Analogous to executions, we define syntactic paths~\(syn_P(S)\) of a program~\(S\).
However, syntactic paths ignore the data state.
{\footnotesize\[\frac{}{S\in syn_P(S)}\quad
\frac{S_0\xrightarrow{op_1}\dots\xrightarrow{op_n}S_n\in syn_P(S), \quad \exists\sigma,\sigma'\in\Sigma: (S_n,\sigma)\xrightarrow{op_{n+1}}(S_{n+1},\sigma') }{S_0\xrightarrow{op_1}\dots\xrightarrow{op_n}S_n\xrightarrow{op_{n+1}}S_{n+1}\in syn_P(S)}\]}
Again, we write \(S\rightarrow^*S'\) if we are not interested in the intermediate steps.

Next, we use the introduced semantics to define when two (sub)programs are equivalent.
We focus on \emph{partial equivalence}, i.e., we limit equivalence to executions that terminate normally.
In addition, we are only interested in equivalence of output variables, i.e., variables that contain the computation results, and ignore the values of intermediate variable.
Given the set of output variables, two programs are equivalent if all executions of both programs that start in the same data state \(\sigma\) and terminate normally agree on the values of the output variables.
\begin{definition}%[Partial equivalence]
Let \(S_1\) and \(S_2\) be two (sub)programs and \(V\subseteq\mathcal{V}\) be the output variables.
\(S_1\) and \(S_2\) are\emph{ partially equivalent} w.r.t.\ \(V\) (denoted by \(S_1\equiv_V S_2\)) if 
\[\forall \sigma,\sigma',\sigma''\in\Sigma, v\in V: ((S_1,\sigma)\rightarrow^*(E,\sigma')\wedge(S_2,\sigma)\rightarrow^*(E,\sigma''))\Rightarrow\sigma'(v)=\sigma''(v).\]
\end{definition}
Our goal is to translate partial equivalence into verification tasks, each of the tasks encoding partial equivalence of subprograms.
As explained in Sec.~\ref{ssec:example}, the encoding relies on information about how variables are used in a subprogram.
Based on the above semantics, we formally define the required usage sets. %, i.e., modified variables~(\(\mathcal{M}\)), variables used before definition (\(\mathcal{UB}\)), and variables live after the subprogram (\(\mathcal{L}\)).
The set of modified variables contains all variables whose value changes.
\begin{definition}
Let \(S\) be a (sub)program. The variables \emph{modified} by \(S\) are:
\[\mathcal{M}(S):=\{v\in\mathcal{V}\mid \exists \sigma, \sigma'\in\Sigma: (S,\sigma)\rightarrow^*(\cdot,\sigma') \wedge \sigma(v)\neq\sigma'(v) \}.\]
\end{definition}
In side-effect free programs only assignments modify variables.
For those programs, the set~\(\mathcal{M}(S)\) of modified variables can be overapproximated by the set of variables that occur in \(S\) on the left-hand side of an assignment.  

The set \(\mathcal{UB}\) contains all variables~$v$ that may be used before they are defined.
For programs, these are the variables that are used uninitialized on some program execution.
Formally, there exists a path such that variable~$v$ occurs in an expression of an operation \(op_i\) on the path and $v$ does not occur on the left-hand side of an assignment on the first i-1 steps of the path.
\begin{definition}
Let \(S\) be a (sub)program. Its variables \emph{used before definition} in execution \(p=(S_0,\sigma_0)\stackrel{op_1}{\rightarrow}\dots\stackrel{op_n}{\rightarrow}(S_n,\sigma_n)\in ex(S)\) are:
\[\begin{array}{l l}\mathcal{UB}_{ex}(p):=&\{v\in\mathcal{V}\mid \exists 1\leq i\leq n:\forall 1\leq j<i: op_j\not\equiv v:=aexpr;  \\
& \wedge (op_i\equiv v':=aexpr; \wedge v\in \mathcal{V}(aexpr) \vee op_i\equiv bexpr\wedge v\in\mathcal{V}(bexpr)) \}\end{array}\]
The variables \emph{used before definition} in \(S\) are \(\mathcal{UB}(S):=\bigcup_{p\in ex(S)}\mathcal{UB}_{ex}(p)\).
%\[\begin{array}{l  l}\mathcal{UB}(S):= & \{v\in\mathcal{V}\mid \exists\sigma\in\Sigma: (S,\sigma)\stackrel{op_1}{\rightarrow}\dots\stackrel{op_n}{\rightarrow}(\cdot,\cdot)\wedge\forall 1\leq k<n: op_k\not\equiv v:=aexpr;\\
 %&  \wedge (op_n\equiv v:=aexpr; \wedge v\in \mathcal{V}(aexpr) \vee op_n\equiv bexpr\wedge v\in\mathcal{V}(bexpr)) \}\end{array}\]
\end{definition}
In practice, one can approximate \(\mathcal{UB}\) using an uninitialized variable analysis on subprogram \(S_i\) or perform a reaching definition analysis~\cite{ProgramAnalysis} on the program~\(S\).

Finally, we define the set~\(\mathcal{L}\) that includes all variables that are live after a subprogram~\(S_1\) of program~\(S\), i.e., all variables that live at at least one program \(S''\) that can be reached from \(S\) after executing \(S_1\).
%To this end, we consider all \dots reachable from \
%The definition is based on variables live at a (sub)program.
%Basically, a variable is live after a subprogram~\(S_1\) if it is live at a program \(S''\) that can be reached from \(S\) after the execution of \(S_1\).
Variables are live at a (sub)program if they may be used in the (sub)program before they are redefined.
To correctly consider output variables, we assume that they are used after the program terminated normally.
Hence, the set~\(\mathcal{L}\) depends on the subprogram, the program, and the output variables. 
Note that the following definition of \(\mathcal{L}\) is only adequate for subprograms that do not occur in parallel statements, which applies to all subprograms that we may compare in our equivalence checks. 

\begin{definition}
Let \(S\) be program, \(S_1\) a subprogram of \(S\) and \(V\subseteq\mathcal{V}\) the output variables. 
We define the \emph{variables live at} \(S_1\) by 
\[\begin{array}{l l}\mathcal{L}(S_1, V):= &\{v\in\mathcal{V}\mid \exists S_1^0\stackrel{op_1}{\rightarrow}\dots, \stackrel{op_n}{\rightarrow}S_1^n\in syn_P(S_1): \forall i\in[1,k]: op_k\not\equiv v:=expr;\\  & \wedge ((k=n-1\wedge (op_n\equiv v':=aexpr; \wedge v\in \mathcal{V}(aexpr)\\ & \vee op_n\equiv bexpr\wedge v\in\mathcal{V}(bexpr)))\vee k=n\wedge S_1^n=E \wedge v\in V ) \}.\end{array} \]
The variables \emph{live} in \(S\) after \(S_1\) are:
\[\begin{array}{l l}\mathcal{L}(S_1, S, V):=&\{v\in\mathcal{V}\mid \exists\sigma,\sigma'\in\Sigma: (S,\sigma)\rightarrow^*(S',\sigma')\in ex(S)\wedge\\ 
&(S'=S_1\wedge v\in \mathcal{L}(E, V)\vee S'=S_1;S''\wedge v\in \mathcal{L}(S'', V))\}\end{array}\]
\end{definition}
In practice, one may use a live variable analysis~\cite{ProgramAnalysis} to compute the set~\(\mathcal{L}\).

\section{Encoding Partial Equivalence of Subprograms}\label{sec:encoding}
In this section, we describe how our \peqcheck approach encodes the partial equivalence of two subprograms into a veri\-fi\-cation task.
Then, we prove that \peqcheck is sound. 
At the end, we discuss limitations of \peqcheck.

To encode partial equivalence of two subprograms, we need to make sure that both subprograms get the same inputs.
Two solution are proposed in the literature.
The first solution saves the data state before executing the first subprogram, save the result (data state) of the execution of the first subprogram, and loads the state~\cite{SymDiff,UCKLEE} before executing the second subprogram.
The second solution assigns equal values to the inputs of the two subprograms~\cite{RV-DAC,RIE}.
To save and load the state, one can either use dedicated methods~\cite{UCKLEE} to write and read the states from (persistent) memory or copy the variable values to and from additional variables that do not occur in the program.
The first option requires the verifier to understand the dedicated methods, which arbitrary verifiers likely will not.
Therefor, we exclude this option.
Saving and loading with additional variables as well as the second solution need to duplicate (shared, modified) variables.
Although when saving and loading  one does not need to rename variables, we think that assigning equal values (as done by the second solution) allows the verifier to more easily learn about the relation of the variables in the two subprograms.
Thus, our encoding will take up the second solution. % for our encoding.
% need additional variables to copy output
% erklären, warum für diesen Ansatz entschieden haben

As demonstrated in Section~\ref{ssec:example}, the \peqcheck encoding consists of three parts: (1)~the (equal) initialization of variables, (2)~the execution of the two subprograms, and (3)~checking equivalence of output variables.
We begin with a description of the general construction of this three parts and later discuss proper inputs required for a sound task generation.
Furthermore, note that our description forgoes to label program statements.

The initialization part is responsible for providing equal inputs to common variables in both code segments.\footnote{In practice, the initialization part also declares variables and due to default initialization, initializes variables non-deterministically . This is not required for our programs.  Furthermore, note that the initialization part is not required to guarantee soundness, but it is important to reliable detect equivalences.}
Our initialization part aims at equalizing duplicated input variables and assumes that non-duplicated (input) variables will not be modified by any of the two subprograms.
\peqcheck will guarantee this assump\-tion.
Before we can describe the initialization in detail, we have to decide whether to assign (1)~the duplicated variable the value of the variable or (2)~the variable the value of the duplicated variable.
Basically, it does not matter which option we choose because both variables are contained in \(\mathcal{V}\), the initialization is the first part of the encoding, and at the beginning of a program all variables are unconstrained.\footnote{In practice, the difference between the two options is whether the variable or its duplication are initialized non-deterministically before this initialization part.}
We chose option~(2) simply because then we initialize the variables of the parallelized code segment with the values of the counterparts in the sequential code segment.
Following option~(2), our initialization part adds one assignment per variable that should be equalized such that the assignment assigns to the variable the valude of duplicated variable.
To know which variables to equalize the initialization encoding is provided with a sequence~$V$ of these variables.
In addition, the initialization encoding requires the renaming function to identify the duplicated variable.
%\todo{initialization only for completeness reason}
%
\[init(\rho, V):=\left\{\begin{array}{l l}E & \textrm{if~} V=\left<\right>\\ v:=\rho(v); & \textrm{if~} V=\left<v\right>\\ v:=\rho(v);init(\rho, V') & \textrm{if~} V=\left<v\right>\circ V'\\\end{array}\right.\]
Next, we describe how the equalization part checks that output variables have identical values, i.e., are equal.
Similar to the initialization part, we only check output variables that are duplicated and use assert statements for checking.
To find out which output variables may violate the partial equivalence property, we generate one assert statement per output variable such that the boolean expression in the assert statement compares the value of the original and duplicated variable.
Again, we require a renaming function~\(\rho\) to identify the duplicated variables and a sequence~$V$ of variables which should be checked for equivalence.
Given this information, the following definition summarizes our idea for the equalization part.
\[equal(\rho, V):=\left\{\begin{array}{l l}E & \textrm{if~} V=\left<\right>\\ \textbf{assert}~\rho(v)==v; & \textrm{if~} V=\left<v\right>\\ \textbf{assert}~\rho(v)==v;equal(\rho, V') & \textrm{if~} V=\left<v\right>\circ V'\\\end{array}\right.\]
After defining the initialization and equalization, we have everything at hand to define the verification task for equivalence checking.
From initialization and equalization, we know that we need a renaming function and two sets of variables.
Set~\(I\) describes the variables that should be equally initialized and set~\(C\) denotes the variables that should be checked for equivalence.
Given this information, the verification task becomes a  sequential composition of the initialization, the renamed subprogram \(S_1\), the subprogram \(S_2\), and the equalization.
To make a set of variables available for definitions \(init\) and \(equal\), we use a function \(\mathrm{toSeq}\) that transforms a set of variables into a sequence.
For example, one implementation of \(\mathrm{toSeq}\) could use the lexical ordering of the variables.
\[eq\_task(S_1, S_2, \rho, I,C):=init(\rho, \mathrm{toSeq}(I));\mathcal{R}(S_1,\rho);S_2;equal(\rho, \mathrm{toSeq}(C))\]
So far, we only presented how to encode an equivalence task, but left out the constraints on the inputs.
To be sound, inputs \(\rho\) and \(I\) must be chosen carefully. % and equalization.
%%% belegung der Variablen (\rho, I, C)

First, let us discuss the constraints on the renaming function~\(\rho\).
To guarantee that the initialization part equalizes \(v\) and its duplicate~\(\rho(v)\) for all variables~\(v\) in the set~\(I\)\footnote{Although the initialization part is not required for soundness, it must work properly if we include it in the \peqcheck approach.}, we require that (a) renaming does not mess up the initialization, i.e., \(\forall v\in I: \rho(v)=v\vee \rho(v)\notin I\). %\todo{require lemma here?}
To ensure that the executions of subprograms \(\mathcal{R}(S_1,\rho)\) and \(S_2\) do not interfere with each other, the renaming function must %(b)~be bijective and 
(b)~prohibit interfering, i.e., 
\(\forall v\in\mathcal{V}(S_1)\cup\mathcal{M}(S_2):\rho(v)\notin\mathcal{M}(S_2)\) and \(\forall v\in\mathcal{M}(S_1): \rho(v)\notin\mathcal{V}(S_2)\cup\mathcal{M}(S_1)\).
%\(\forall v\in\mathcal{V}(S_1):\rho(v)\notin\mathcal{M}(S_2)\) and \(\forall v\in\mathcal{M}(S_1): \rho(v)\notin\mathcal{V}(S_2)\).
%\todo{require lemma here?}
We call renaming functions fulfilling the latter constraints \emph{appropriate for renaming}.
The renaming function~\(\rho_\mathrm{sum2}\) introduced in the previous section is appropriate for renaming.

To ensure that subprograms~\(S_1\) and \(S_2\) get the same input, the initialization must consider all duplicated variables that \(S_1\) or \(S_2\) use before definition.
%\todo{check this condition}

In practice, we consider overapproximations \(\mathcal{UB}(S_1)\subseteq U_1\subseteq \mathcal{V}(S_1)\) and \(\mathcal{UB}(S_2)\subseteq U_2\subseteq \mathcal{V}(S_2)\) of the variables used before definition and  overapproximations \(\mathcal{M}(S_1)\subseteq M_1\subseteq \mathcal{V}(S_1)\) and \(\mathcal{M}(S_2)\subseteq M_2\subseteq \mathcal{V}(S_2)\) of the modified variables.
We use the overapproximations of the modified variables limit the equivalence check to a  subset \(C\subseteq M_1\cup M_1\) of the possibly modified variables.
Variable liveness and, thus, the subset of output variables, will further determine the set~\(C\).
To generate a bijective renaming function~\(\rho_\mathrm{switch}\), we rely on an injective function \(\mathrm{switch}: M_1\cup M_2\rightarrow \mathcal{V}\setminus(\mathcal{V}(S_1)\cup\mathcal{V}(S_2))\) that defines the duplicate variables.
Based on such an injective function \(\mathrm{switch}\), the renaming function~\(\rho_\mathrm{switch}\) switches all modified variables (\(M_1\cup M_2\)) with a non-program variable and keeps all other variables, i.e.,
for all variables \(v\in\mathcal{V}\) the renamed variable is \(\rho_\mathrm{switch}(v)=\mathrm{switch}(v)\) if \(v\in M_1\cup M_2\), \(\rho_\mathrm{switch}(v)=v_m\) if there exists \(v_m\in M_1\cup M_2\) and \(\mathrm{switch}(v_m)=v\), and \(\rho_\mathrm{switch}(v)=v\) in all other cases.\footnote{We rename the variables in \(img(\mathrm{switch})\) to guarantee bijectivity.} % of \(\rho_\mathrm{switch}\).}
Note that \(\rho_\mathrm{switch}\) is a renaming function appropriate for renaming and fulfills condition (a) for \(I=(U_1\cap U_2)\cap(M_1\cup M_2)\)\footnote{Proved by Lemma~\ref{lem:corrRenaming} in the appendix.}, which we use to generate our tasks.
In our example, we use \(\mathrm{switch}_\mathrm{sum2}:sum\mapsto sum\_s\) to generate \(\rho_\mathrm{sum2}\). 

Next, we discuss soundness of our encoding.
To be sound, our encoding must ensure that if the verification of the encoded equivalence task succeeds, i.e., none of its executions violates an assertion, then the two subprogram \(S_1\) and \(S_2\) will be partially equivalent with respect to the unmodified variables and the variables~\(C\), which are checked for equivalence.
The following theorem ensures this property for equivalence tasks created with the inputs discussed above.
\begin{theorem}\label{theo:EquivTotal}
Let \(S_1\) and \(S_2\) be two (sub)programs, \(\mathcal{UB}(S_1)\subseteq U_1\subseteq \mathcal{V}(S_1)\), \(\mathcal{UB}(S_2)\subseteq U_2\subseteq \mathcal{V}(S_2)\), \(\mathcal{M}(S_1)\subseteq M_1\subseteq \mathcal{V}(S_1)\), \(\mathcal{M}(S_2)\subseteq M_2\subseteq \mathcal{V}(S_2)\), \(\rho_\mathrm{switch}\) a renaming function, and \(C\subseteq M_1\cup M_2\). Define the equivalence task to be \(S=eq\_task(S_1, S_2, \rho_\mathrm{switch}, (U_1\cap U_2)\cap(M_1\cup M_2) ,C)\).

If all execution \((S,\sigma)\rightarrow^*(S',\sigma')\in ex(S)\) do not violate an assertion, then \(S_1\equiv_{\mathcal{V}\setminus((\mathcal{M}(S_1)\cup\mathcal{M}(S_2)\setminus C)} S_2\).
\end{theorem}
\begin{proof} See~appendix~\ref{ssec:proof:theo:EquivTotal}.\end{proof}

So far, we learnt how to soundly apply \peqcheck to the complete program.
However, our goal is to split equivalence checking of two programs into equivalence checking of pairs of subprograms.

To split equivalence checking of programs~\(S\) and \(S'\), we assume that there exists a partial, injective \emph{replacement function} \(\gamma\) such that \(S'\) can be derived from \(S\) by replacing all subprograms~\(S_1\) of \(S\) with \(S_1\in dom(\gamma)\) by \(\gamma(S_1)\). 
We write \(\Gamma(S,\gamma)\) to denote the result of this replacement and make the following assumptions about the replacement: %make the following assumption:
Programs \(E, [E \| \dots \| E]\notin dom(\gamma)\).
The domain \(dom(\gamma)\) only contains subprograms of \(S\) and all subprograms in the domain \(dom(\gamma)\) do not occur in a parallel statement of \(S\).
Similarly, we assume that all subprograms in the image \(im(\gamma)\) of \(\gamma\) do not occur in a parallel statement of \(S'\).
Thus, we e.g.\, ensure that thread interference cannot invalidate the result of \peqcheck's equivalence checking.
Note that such a replacement function always exists. One can always use \(\gamma=\{(S,S')\}\).

Given a replacement function~\(\gamma\) and the set~\(V\) of output variables, \peqcheck generates one equivalence task per pair \((S_1, S_2)\in\gamma\). % for which \(S_1\) is a subprogram of \(S\).
Thereby, it utilizes overapproximations \(\mathcal{UB}(S_1)\subseteq U_1\subseteq \mathcal{V}(S_1)\) and \(\mathcal{UB}(S_2)\subseteq U_2\subseteq \mathcal{V}(S_2)\) of the variables used before definition,  overapproximations \(\mathcal{M}(S_1)\subseteq M_1\subseteq \mathcal{V}(S_1)\) and \(\mathcal{M}(S_2)\subseteq M_2\subseteq \mathcal{V}(S_2)\) of the modified variables, and overapproximations \(\mathcal{L}(S_1, S, V)\subseteq L_1\subseteq \mathcal{V}\) and \(\mathcal{L}(S_2, S', V)\subseteq L_2\subseteq \mathcal{V}\)  of the variables live after \(S_1\) and \(S_2\).
Based on these sets, \peqcheck builds the equivalence task 
{\footnotesize\(eq\_task(S_1, S_2, \rho_\mathrm{switch}, (U_1\cap U_2)\cap(M_1\cup M_2), (M_1\cup M_2)\cap (L_1\cup L_2))\).}

After the generation of the verification tasks, \peqcheck analyzes each verification task and returns that \(S\) and \(S'\) are equivalent if none of the tasks violates an assertion.
The subsequent theorem shows that this behavior of \peqcheck is sound when the variables used before definition are computed precisely.

\begin{theorem}\label{theo:EquivNoApprox}
Let \(S\) and \(S'\) be two programs, \(\gamma\) be a replacement function such that \(S'=\Gamma(S,\gamma)\), and  \(V\subseteq\mathcal{V}\) be a set of outputs.
If for all \((S_1, S_2)\in\gamma\) there exists \(\mathcal{M}(S_1)\subseteq M_1\subseteq \mathcal{V}(S_1)\), \(\mathcal{M}(S_2)\subseteq M_2\subseteq \mathcal{V}(S_2)\), \(\mathcal{L}(S_1, S, V)\subseteq L_1\subseteq \mathcal{V}\), \(\mathcal{L}(S_2, S', V)\subseteq L_2\subseteq \mathcal{V}\), and renaming function \(\rho_\mathrm{switch}\) such that the equivalence task \(eq\_task(S_1, S_2, \rho_\mathrm{switch}, (\mathcal{UB}(S_1)\cap\mathcal{UB}(S_2))\cap(M_1\cup M_2), (M_1\cup M_2)\cap (L_1\cup L_2))\) does not violate an assertion, then \(S\equiv_V S'\).
\end{theorem}
\begin{proof}See appendix~\ref{ssec:proof:theo:EquivNoApprox}.\end{proof}

Computing the precise set of variables used before definition might costly and or even impossible in practice.
Therefore, one typically computes overapprxomiations of these sets.
However, we learnt from our proof attempts that not all overapproximations are appropriate because modifications are defined semantically while live variables are defined syntactically.
The precise problem is that during initialization the verification task could equalizes a variable~$v$ whose value is not identical before, $v$ is assigned in code segment~\(S_2\), but the value of $v$ does not change (i.e., it is not modified in \(S_2\)), and $v$ becomes live in \(S_2\) afterwards.
In this particular case, the comparison in the equalization will consider the wrong value for the variable~$v$ in \(S_2\).
One can avoid this problem if the overapproximation of the modified variables for \(S_2\) always considers all assignments in \(S_2\).

\begin{theorem}\label{theo:EquivApprox}
Let \(S\) and \(S'\) be two programs, \(\gamma\) be a replacement function such that \(S'=\Gamma(S,\gamma)\), and  \(V\subseteq\mathcal{V}\) be a set of outputs.
If for all \((S_1, S_2)\in\gamma\) there exists overapproximations \(\mathcal{UB}(S_1)\subseteq U_1\subseteq \mathcal{V}(S_1)\), \(\mathcal{UB}(S_2)\subseteq U_2\subseteq \mathcal{V}(S_2)\),
 \(\mathcal{M}(S_1)\subseteq M_1\subseteq \mathcal{V}(S_1)\), %\(\mathcal{M}(S_1)\cup\{v\in\mathcal{V}\mid\exists S_1\rightarrow^*S_k\stackrel{v:=aexpr}{\rightarrow}S_r\in syn_P(S_1)\}\subseteq M_1\subseteq \mathcal{V}(S_1)\), 
\(\mathcal{M}(S_2)\cup\{v\in\mathcal{V}\mid\exists S_2\rightarrow^*S_k\stackrel{v:=aexpr}{\rightarrow}S_r\in syn_P(S_2)\}\subseteq M_2\subseteq \mathcal{V}(S_2)\), 
 \(\mathcal{L}(S_1, S, V)\subseteq L_1\subseteq \mathcal{V}\), \(\mathcal{L}(S_2, S', V)\subseteq L_2\subseteq \mathcal{V}\), and renaming \(\rho_\mathrm{switch}\) s.t.\ \mbox{\(eq\_task(S_1, S_2, \rho_\mathrm{switch}, (U_1\cap U_2)\cap(M_1\cup M_2), (M_1\cup M_2)\cap (L_1\cup L_2))\)} does not violate an assertion, then \(S\equiv_V S'\).

\end{theorem}
\begin{proof}See~appendix~\ref{ssec:proof:theo:EquivApprox}.\end{proof}

\begin{figure}[t]%
{\scriptsize
\begin{minipage}{0.48\textwidth}
\begin{lstlisting}[ caption=Original program]
int foo_orig(int x, int y)
{
\end{lstlisting}
\vspace{-\baselineskip}
\begin{lstlisting}[backgroundcolor=\color{cyan!20}]
  if(x<1)
    y=0;
\end{lstlisting}
\vspace{-\baselineskip}
\begin{lstlisting}
  return y;
}
\end{lstlisting}
\vspace{-0.5\baselineskip}
\begin{lstlisting}[caption=Modified program]
int foo_mod(int x, int y)
{
\end{lstlisting}
\vspace{-\baselineskip}
\begin{lstlisting}[backgroundcolor=\color{yellow!40}]
  if(!(x>0))
    y=0;
\end{lstlisting}
\vspace{-\baselineskip}
\begin{lstlisting}
  return y;
}
\end{lstlisting}
\end{minipage}
\hfill
\begin{minipage}{0.48\textwidth}
\begin{lstlisting}[caption=Verification task]
int main()
{
 int x, y_s, y;
 x = random_int();
\end{lstlisting}
\vspace{-0.5\baselineskip}
\begin{lstlisting}[backgroundcolor=\color{cyan!20}]	
  if(x<0)
    y_s=1;
\end{lstlisting}
\vspace{-1\baselineskip}
\begin{lstlisting}[backgroundcolor=\color{yellow!40}]
  if(!(x>0))
    y=1;
\end{lstlisting}
\vspace{-0.5\baselineskip}
\begin{lstlisting}[language=C]	
 assert(y_s == y);
 return 0;
}
\end{lstlisting}
\end{minipage}
}
%\vspace{-2mm}
\caption{Behaviorally equivalent original and modified subprograms for which our \peqcheck approach fails to show equivalence}%
\label{fig:probEnc}%
%\vspace{-2mm}
\end{figure}
\subsection{Discussion} % Discussion
As has been shown above, our \peqcheck approach is sound, i.e., it never approves two inequivalent programs. 
However, it cannot be complete because functional equivalence of two programs is undecidable~\cite{RV-DAC}.
%Next, let us discuss some of the specific limitations of the \peqcheck encoding.
%First, 
Thus, our \peqcheck approach may not testify all equivalent programs and there exist equivalent code segments for which the generated verification task violates an assertion.
An example is shown in Fig.~\ref{fig:probEnc}.
For this example, the generated verification task violates the assertion because variable~\(y\) is live, is not used before definition, but is not defined on all program paths.
Hence, variable~\(y\) is duplicated, but not equally initialized and when following the else branch the assertion could be violated.
One could avoid this issue by also equally initializing all variables that are only modified on some paths.
Another completeness issue is a code segment that may violate an assertion.
To deal with assertions \(\mathbf{assert}~bexpr;\) in code segments, the task encoding can replace them by \(\mathbf{while}~bexpr~\mathbf{do}~E\).
Another reason why the equivalence is not detected is that an equivalence task considers more input values to a code segment than can be reached by all program executions.
To improve on this problem, one could aim at computing (an overapproximation of) the input ranges for code segments and restrict the initialization with to the computed input ranges.
However, it is unlikely that one succeeds to always compute the precise range of input values.

\begin{figure}[t]%
\begin{minipage}{0.48\textwidth}
\begin{lstlisting}[ caption=Sequential program]
int sum_seq(int N)
{
  int sum = 0, a[N];
\end{lstlisting}
\vspace{-\baselineskip}
\begin{lstlisting}[backgroundcolor=\color{orange!20}]
  for(int i = 0; i < N; i++)
    a[i] = i;
\end{lstlisting}
\vspace{-\baselineskip}		
\begin{lstlisting}[backgroundcolor=\color{green!20}]		
  for(int i = 0; i < N; i++)
    sum += a[i];
\end{lstlisting}
\vspace{-\baselineskip}
\begin{lstlisting}
  return sum;
}
\end{lstlisting}
\end{minipage}
\hfill
\begin{minipage}{0.48\textwidth}
\begin{lstlisting}[caption=Parallelized program]
int sum_par(int N)
{
  int sum = 0, a[N];
\end{lstlisting}
\vspace{-\baselineskip}
\begin{lstlisting}[backgroundcolor=\color{orange!20}, basicstyle=\tiny]
#pragma omp parallel for
\end{lstlisting}
\vspace{-1.2\baselineskip}
\begin{lstlisting}[backgroundcolor=\color{orange!20}]
  for(int j = N-1; j>=0; j--)
    a[j] = j;
\end{lstlisting}
\vspace{-\baselineskip}	
\begin{lstlisting}[backgroundcolor=\color{green!20}, basicstyle=\tiny]
#pragma omp parallel for reduction(+:sum)
\end{lstlisting}
\vspace{-1.2\baselineskip}
\begin{lstlisting}[backgroundcolor=\color{green!20}]		
  for(int j = N-1; j>=0; j--)
    sum += a[j];
\end{lstlisting}
\vspace{-\baselineskip}
\begin{lstlisting}
  return sum;
}
\end{lstlisting}
\end{minipage}
%\vspace{-2mm}
\caption{Behaviorally equivalent sequential and parallelized program whose first code segments (highlighted in orange) are not identical}%
\label{fig:probSegments}%
%\vspace{-2mm}
\end{figure}
A further aspect is the choice of code segments.
Structurally, code segments must be subprograms and they must not occur in a parallel statement, which limits the granularity of code segments, but not the applicability of the approach.
% Furthermore, we cannot use uninterpreted function like e.g.~\cite{RV-DAC,SymDiff} to abstract from function calls.
However when choosing the wrong code segments, one may miss equivalent programs.
For example, consider the sequential and the parallelized program shown in Fig.~\ref{fig:probSegments} whose for loops are not identical.
If we use two code segments, one per for loop, then equivalence checking fails.
In contrast, it succeeds if we choose the code segment to contain both for loops.
\section{\peqcheck Implementation}
To check functional equivalence with \peqcheck, one must (1)~identify the code segments, (2)~generate the verification tasks for the code segments, and (3)~verify the tasks.
Currently, we perform steps~(2) and (3) automatically and execute step~(1) manually, i.e., we manually insert pragma statements \texttt{\#pragma scope\_i} and \texttt{\#pragma epocs\_i} to specify the start and end of code segment \(i\).

\textbf{Step~1: Identifying code segments.}
When checking the equivalence of a sequential program and its OpenMP parallelization, manually identifying the code segments is simple. 
Using the code blocks associates with the outermost OpenMP directives often works well and we applied this strategy to determine the code segments for our parallelized examples.
When checking two sequential versions, selecting adequate code segments is more challenging. 
Naively using each statement that differs in both versions as a single code segment likely results in many inequivalent code segments, even if the versions are equivalent.
Also, using functions as code segments might be a bad choice as we will see in our experiments, especially if the functions contain multiple independent changes.
Generally, a developer should have deeper insights in which code parts to select.
Therefore, it might be a good idea to combine all changes of a function that belong to the same commit.

\textbf{Step~2: Generation of verification tasks.}
We developed a proto\-type tool that implements the approach from Section~\ref{sec:encoding}.
Our prototype tool is part of the framework for equivalence checking of parallelized code (FECheck)\footnote{\url{https://git.rwth-aachen.de/svpsys-sw/FECheck}} and in our experiments we use tag version \texttt{PEQcheck-Formalise2021}.
The tool builds on the ROSE compiler framework~\cite{ROSE} (v0.9.13.0) and it uses ROSE's live analysis to identify the variables live afterwards.
In addition, it executes ROSE's reaching definition analysis to determine which variables are modified and which are used before definition.
Both analyses are intra\-procedural.
Thus, we overapproximate the behavior of global variables and parameters passed.
For example, we assume that global variables and non-scalar parameters are always live and that a called function always use all global variables and non-scalar parameters before they are defined and also modifies them.
   
\textbf{Step~3: Verification.}
To verify the generated tasks, we utilize the verifiers CIVL~\cite{CIVL} (version 1.20\_5259 with theorem prover Z3~\cite{Z3} (version~4.8.10)) and \cpachecker\footnote{\url{https://cpachecker.sosy-lab.org/download.php}}~\cite{CPACHECKER} (version~2.0).
CIVL is developed to verify parallelized programs like OpenMP programs while \cpachecker is a successful verifier for sequential programs.
To verify a task~\texttt{task.c} with OpenMP constructs, we execute CIVL with the following command.
\begin{quote}
\noindent\texttt{civl verify -input\_omp\_thread\_max=2 -checkDivisionByZero=false -checkMemoryLeak=false -timeout=300 nondet\_funs.c task.c}
\end{quote}
The command limits CIVL's verification to 5\,min and two threads. 
File~\texttt{nondet\_funs.c} implements the random input functions, which return elements from [-5;5].
To verify sequential tasks~\texttt{task.c}, we rely on \cpachecker's default analysis and execute the following command line.
\begin{quote}
\texttt{scripts/cpa.sh -default -noout -timelimit 300s -preprocess -spec config/specification/Assertion.spc task.c}
\end{quote}
In our experiments, we utilize additional scripts to automatically perform steps~(2) and (3) on our examples.
\section{Experiments}
In our experiments, we plan to demonstrate \peqcheck's generality and to examine whether localized equivalence checking is beneficial.
To demonstrate \peqcheck's generality, we apply it to different benchmark sets: one for parallelized programs and one for different versions of sequential programs.
Unfortunately, we could not compare \peqcheck with existing approaches that use a similar encoding idea because these approaches are not available~\cite{UCKLEE,DBLP:journals/ijpp/AbadiKPV19} (for C~programs~\cite{SymDiff,RIE}) or the tool compilation failed~\cite{RV-DAC,RV-Journal}.

\textbf{Environmental Set Up.}
Our experiments are executed on a machine with an Intel i7-8565U CPU (frequency of 1.8\,GHz) and 32\,GB RAM, which runs an Ubuntu~18.04.
Furthermore, we count the lines of codes with the tool~\texttt{cloc}~v1.74\footnote{\url{https://github.com/AlDanial/cloc}}.

\textbf{Benchmark.}
We consider two sets of benchmarks.
Our first benchmark set aims at checking equivalence of sequential and parallelized programs.
It contains four own examples and our parallelizations of the \texttt{*\_spec.c} files  from the functional equivalence suite (FEVS)~\cite{FEVS}.
Note that we did not parallelize programs \texttt{diffusion1d-gd}, \texttt{diffusion2d-gd}, and \texttt{nbody} because their header files are missing.
Furthermore, we failed to parallelize \texttt{fib}.
In addition, we parallelized the iterative instead of the recursive factorial implementation.
To deal with I/O inputs, we replaced them by calls to random functions and we also replaced the assert statements.
As described earlier, the local code segments are the parallel code segmenets. 
Our second benchmark set focuses on checking equivalence of two sequential program versions.
It contains the non-recursive programs considered by R\^{e}ve~\cite{Reve} (except for \texttt{loop4} and \texttt{loop5}, which were not available).
The local code segments are the smallest subprogram that is influenced by a change.
To examine whether localized equivalence checking is bene\-ficial, we use another set of code segments, named~$\mathrm{all}$, which contains one code segment per program that covers the complete program.

\subsection{\peqcheck on Parallelized Programs}
Table~\ref{tab:OpenMPStudy} shows the results of our \peqcheck evaluation on the first benchmark set, the benchmark tasks with the parallelized programs. 
The first four tasks are our own examples (\texttt{ex} is the example from Fig.~\ref{fig:example}) and the remaining tasks represent the FEVS examples. 
Benchmark tasks that end on \texttt{-e} and are highlighted in light red represent incorrect parallelizations.
For each benchmark task, Tab.~\ref{tab:OpenMPStudy} shows the number of local code segments\footnote{By construction, the set~$\mathrm{all}$ contains one segment per task.}, the lines of code of the sequential program, the parallelized program  and the verification tasks (for both configurations of code segments).
If the local configuration~seg contains more than one verification task, the table reports the maximal number of lines of code among all tasks.
In addition, the table shows the total time spent on generating the verification tasks plus the total time spent on verification and the verification results.
Again, the times are provided for  both configurations of code segments.

\setlength{\tabcolsep}{5pt}
\begin{table}[t]
  \caption{Evaluation results of \peqcheck on a sequential program and its parallelization (taken from \cite{PEQcheck})}
	\label{tab:OpenMPStudy}
	\centering
	\scalebox{0.75}{
		\begin{tabular}{l d{0} d{0} d{0} d{0} d{0} d{0} d{0} d{0} d{0} c c}
		\toprule
		 &  & \multicolumn{4}{c}{LOC}  & \multicolumn{2}{c}{time$_\mathrm{enc}$ (s)} & \multicolumn{2}{c}{time$_\mathrm{CIVL}$ (s)}& \multicolumn{2}{c}{status}\\
    \cmidrule(l{5pt}r{5pt}){3-6}\cmidrule(l{5pt}r{5pt}){7-8}\cmidrule(l{5pt}r{5pt}){9-10}\cmidrule(l{5pt}r{5pt}){11-12}
		Benchmark  & \multicolumn{1}{c}{\#seg-} & \multicolumn{1}{c}{$P_\mathrm{seq}$} & \multicolumn{1}{c}{$P_\mathrm{par}$} & \multicolumn{1}{c}{$P_\mathrm{seg}$} & \multicolumn{1}{c}{$P_\mathrm{all}$} & \multicolumn{1}{c}{$t_\mathrm{seg}^\mathrm{E}$}  & \multicolumn{1}{c}{$t_\mathrm{all}^\mathrm{E}$} & \multicolumn{1}{c}{$t_\mathrm{seg}^V$} & \multicolumn{1}{c}{$t_\mathrm{all}^V$} & \multicolumn{1}{c}{$s_\mathrm{seg}$} & \multicolumn{1}{c}{$s_\mathrm{all}$}\\
	tasks  & \multicolumn{1}{c}{ments}  &   &   & (max) &  &   & &  &  & & \\

		\midrule
adder-s2			 & 1 & 10 & 11 & 21 & 44 & 5 & 6 & 5 & 4 & \checkmark & NA\\ %null
\rowcolor{red!10}adder-s-e			 & 1 & 10 & 11 & 21 & 44 & 4 & 6 & 4 & 4 & $\times$ & NA\\ % null 
adder-s				 & 1 & 10 & 11 & 21 & 44 & 4 & 5 & 6 & 4 & \checkmark & NA\\ % null
ex             & 1 & 12 & 13 & 25 & 25 & 5 & 5 & 5 & 4 & \checkmark & \checkmark\\ % ETAPS example, % identical (seg, all)
\midrule
adder2-nd			 & 1 & 15 & 16 & 30 & 54 & 6 & 6 & 78 & 4 & \checkmark & NA\\ % null
adder2				 & 2 & 17 & 20 & 28 & 28 & 8 & 6 & 607 & 304 & TO & TO\\% TO (305), TO (304)
\rowcolor{red!10}adder-e				 & 2 & 17 & 19 & 26 & 26 & 8 & 5 & 608 & 305 & TO & TO\\ % TO (304), TO
\rowcolor{red!10}adder-nd-e		 & 1 & 15 & 16 & 30 & 54 & 6 & 5 & 18 & 5 & $\times$ & NA\\ % null
adder-nd			 & 1 & 15 & 16 & 31 & 54 & 6 & 6 & 73 & 5 & \checkmark & NA\\%null
adder					 & 2 & 17 & 19 & 26 & 26 & 8 & 5 & 608 & 304 & TO & TO \\ % TO (304), TO (304)
diffusion1d-nd & 3 & 44 & 51 & 61 & 197 & 12 & 6 & 49 & 4 & NA/$\times_\mathrm{mem}$/$\times$& EX\\ %Err(null) (5), MEM/free (41), assert(6) % exception
diffusion1d		 & 2 & 43 & 45 & 42 & 153 &  8 & 6 & 11 & 42 & $\times_\mathrm{mem}$ & $\times$\\ % OUT OF BOUNDS (5, 6) assertion (CIVL eingeführt)
diffusion2d-nd & 2 & 50 & 59 & 91 & 237 & 8 & 6 & 11 & 5 & NA &  NA \\ % (ERR (5), ERR (6) ERR
diffusion2d		 & 2 & 63 & 67 & 106 & 242 & 8 & 6 & 18 & 305 & $\times_\mathrm{mem}$ & TO\\ %DEREFENCE (7) OUT OF BOUNDS (12) 
factorial2		 & 1 & 11 & 15 & 28 & 26 & 6 & 6 & 7 & 6 & \checkmark & \checkmark\\ % nearly identical (seg, all), seg allgemeiner
\rowcolor{red!10}factorial-e		 & 1 & 11 & 14 & 27 & 25 & 6 & 6 & 5 & 6 & $\times$ & $\times$\\ % nearly identical (seg, all), seg allgemeiner
factorial			 & 1 & 11 & 12 & 25 & 23 & 6 & 6 & 7 & 6 & \checkmark & \checkmark\\ % nearly identical (seg, all), seg allgemeiner
\rowcolor{red!10}gausselim-e		 & 2 & 98 & 108 & 64 & 228 & 8 & 6 & 19 & 6 & $\times_\mathrm{mem}$ & NA\\% DEREFERENCE (7, 16), null
gausselim			 & 3 & 100 & 111 & 64 & 229 & 12 & 6 & 25 & 6 & $\times_\mathrm{mem}$ & NA\\% DEREFERENCE(7), OUT OF BOUNDS (7), DEREFENCE(14), null
integrate			 & 1 & 59 & 60 & 150 & 163 & 7 & 7 & 1 & 1 & EX & EX\\ % EXCEPTION, exception
laplace				 & 3 & 48 & 55 & 71 & 148 & 12 & 6 & 11 & 4 & EX/$\times_\mathrm{mem}$& EX \\ % EXCEPTION (3), OUT OF BOUNDS (5), EXCEPTION (3), exception
matmat				 & 1 & 33 & 37 & 78 & 133 & 6 & 6 & 304 & 5 & TO & NA\\ % null
\rowcolor{red!10}mean-e				 & 1 & 17 & 18 & 31 & 54 & 5 & 6 & 8 & 4 & $\times$ & NA\\ %null
mean					 & 1 & 17 & 18 & 31 & 54 & 6 &  6 & 79 & 5 & \checkmark & NA\\ %null
wave1d-nd			 & 2 & 99 & 101 & 61 & 315 & 9 & 7 & 13 & 7 & $\times_\mathrm{mem}$/NA & NA\\ % OUT OF BOUNDS (7), ERR (6), null
wave1d				 & 2 & 87 & 89 & 136 & 295 & 9 & 6 & 17 & 306 & $\times_\mathrm{mem}$ & TO\\ % OUT OF BOUNDS (10, 7)
\bottomrule		
		\end{tabular}
		}
		%\vspace{-4mm}
\end{table}

First, we study \peqcheck's results for the local code segments (configuration~seg).
Looking at the lines of code~(LOC), we observe that an encoded verification task is often larger than the sequential program and than the parallelized program.
One explanation is that the tasks contain the code of the sequential and the parallelized code segment and in our examples the code segments often dominate the program code.
In addition, the verification tasks contain code that initializes input variables and code that checks equivalence of output variables, which is not present in the sequential and parallelized program. %, more specifically modified live variables. 
Inspecting  the time for generating all verification tasks of a benchmark task (column~$t_\mathrm{seg}^\mathrm{E}$), we recognize that task generation only tasks a few seconds.
Thus, it is rather fast.
However, the generation time may slow down if the input programs get significantly larger.
Now, let us look at the verification of the generated equivalence tasks (columns~$t_\mathrm{seg}^V$ and $s_\mathrm{seg}$). 
For one third of the benchmark tasks, CIVL fails with an exception (EX), a time out (TO), or no available result (NA)\footnote{A result is not available if CIVL returns result null.}.
In addition, the verification of 8 of 26~benchmark task results in status~$\times_\mathrm{mem}$, which means that CIVL detects a memory violation, either an out of bounds access or an invalid dereference.
These memory violations exist because the size of pointer-based arrays assumed by the generated tasks and the program mismatch.  
Note that this is not a general problem of the \peqcheck approach because program executions with memory violations do not terminate normally.
Thus these executions are not considered for partial equivalence.
Furthermore, the problems is a C specific issue of our \peqcheck implementation.
To fix the issue, we must integrate an additional program analysis that aims to find out which variable stores the size of a particular pointer-based array. 
Then, our encoding  must guarantee  that the variable in which the encoding stores the size of the pointer-based array and the program variable storing the size of that pointer-based array contain the same value.
Next, we look at the 12~benchmark tasks with status~\checkmark and $\times$. 
We observe that (a)~equivalence~(\checkmark) is only reported for equivalent tasks and (b)~inequivalence~($\times$) is detected for 4 of the 6~inequivalent tasks. %
Thus, localized equivalence checking with \peqcheck can correctly detect (in)equivalence.

%
%Looking at the verification times, which we only report when neither an error or timeout occurs, we conclude that when CIVL reports a result, it reports it within seconds.
%However all in all, CIVL is not a good choice for our tasks.

Finally, we compare \peqcheck with localized equivalence checking (configuration~seg) against all at once checking (configuration~all).
First, we observe that the times for the generation of the verification tasks is similar.
Nevertheless, the tasks for configuration~all are typically larger than the local tasks (73\% of the tasks are larger and only 3 are smaller).
Thus, the size of the generated tasks is a first indication that localized checking reduces complexity.
Another indicator is the verification itself.
For configuration~all, the verification fails in 80\% of the benchmark tasks, while for configuration~seg the verification fails for about one third of the cases and reports a memory violation in another 20\% of the benchmark tasks.
One reason for more failures are the encoding of (random) input functions. 
The encoding uses static local variables to ensure that the sequential and parallelized code get the same identical value for their ith call to a random function and
CIVL seems to have problems with static local variables.
%\begin{itemize}
	%\item Unterschied factorial Programme, nicht deterministische Initialisierung vs. Programm Initialisierung
	%\item typical larger programs
	%\item ERR (problem kann nicht mit static lokaler Variable umgehen, Behandlung nicht deterministischer Eingaben (random funktion, verifier methode)
	%\item TO time out in all wird auch erreicht mit doppeltem Timeout (Gesamtzeit für localized verification)
%\end{itemize}\todo[inline]{adapt evaluation sections}
Furthermore, configuration~seg also performs better in terms of correct results.
Configuration~seg determines the correct result for 12 of the 26~tasks, while configuration~all reports the correct result in 4 cases and these cases are also correctly handled by configuration~seg.
In addition, we notice little difference for tasks for which both configurations reported either \checkmark or $\times$. % between the two configurations.
In summary, localized equivalence checking seems to be beneficial.

\begin{table}[t]
  \caption{Evaluation results of \peqcheck on pairs of sequential programs \\(taken from \cite{PEQcheck})}
	\label{tab:SeqStudy}
	\centering
	\scalebox{0.75}{
			\begin{tabular}{l d{0} d{0} d{0} d{0} d{0} d{0} d{0} d{0} d{0} c c}
		\toprule
		 &  & \multicolumn{4}{c}{LOC}  & \multicolumn{2}{c}{time$_\mathrm{enc}$ (s)} & \multicolumn{2}{c}{time$_\mathrm{CPA\checkmark}$ (s)}& \multicolumn{2}{c}{status}\\
    \cmidrule(l{5pt}r{5pt}){3-6}\cmidrule(l{5pt}r{5pt}){7-8}\cmidrule(l{5pt}r{5pt}){9-10}\cmidrule(l{5pt}r{5pt}){11-12}
		Benchmark  & \multicolumn{1}{c}{\#seg-} & \multicolumn{1}{c}{$P_\mathrm{orig}$} & \multicolumn{1}{c}{$P_\mathrm{mod}$} & \multicolumn{1}{c}{$P_\mathrm{seg}$} & \multicolumn{1}{c}{$P_\mathrm{all}$} & \multicolumn{1}{c}{$t_\mathrm{seg}^\mathrm{E}$}  & \multicolumn{1}{c}{$t_\mathrm{all}^\mathrm{E}$} & \multicolumn{1}{c}{$t_\mathrm{seg}^V$} & \multicolumn{1}{c}{$t_\mathrm{all}^V$} & \multicolumn{1}{c}{$s_\mathrm{seg}$} & \multicolumn{1}{c}{$s_\mathrm{all}$}\\
	 tasks &   \multicolumn{1}{c}{ments} &   &   & (max) &  &   & &  &  & & \\
\midrule
barthe & 1 & 16 & 16 & 39 & 35 & 5 & 5 & 5 & 303 & $\times$ & TO\\
\rowcolor{red!10}barthe-e & 1 & 19 & 22 & 44 & 44 & 5 & 5 & 5 & 302 & $\times$ & TO \\
barthe2 & 1 & 14 & 14 & 31 & 29 & 5 & 5 & 305 & 305 & TO & TO\\
barthe2-big & 1 & 19 & 19 & 33 & 39 & 5 & 5 & 303 & 304 & TO & TO\\
barthe2-big2 & 1 & 24 & 24 & 33 & 49 & 6 & 5 & 306 & 305 & TO & TO \\
bug15 & 2 & 13 & 13 & 24 & 25 & 8 & 5 & 55 & 5 & \checkmark & \checkmark\\
digits10 & 1 & 29 & 32 & 75 & 75 & 5 & 5 & 32 & 30 & \checkmark & \checkmark \\ 
\rowcolor{red!10}digits10-e & 1 & 26 & 29 & 62 & 62 & 6 & 6 & 4 & 4 & $\times$ & $\times$\\
loop & 1 & 11 & 11 & 24 & 24 & 5 & 5 & 305 & 305 & TO & TO\\
loop2 & 1 & 11 & 11 & 26 & 24 & 5 & 5 & 306 & 304 & TO &  TO\\
loop3 & 1 & 14 & 14 & 26 & 32 & 5 & 5 & 4 & 305 & $\times$ & TO\\
\rowcolor{red!10}loop5-e & 1 & 14 & 14 & 29 & 29 & 5 & 5 & 6 & 5 & $\times$ & $\times$\\
nested-while & 2 & 21 & 19 & 26 & 43 & 8 & 5 & 7 & 306 & \checkmark & TO\\
\rowcolor{red!10}nested-while-e & 2 & 20 & 18 & 26 & 41 & 8 & 5 & 8 & 7 & $\times$ & $\times$\\
simple-loop & 1 & 9 & 9 & 17 & 17 & 5 & 5 & 5 & 5 & \checkmark & \checkmark\\
\rowcolor{red!10}simple-loop-e & 1 & 10 & 10 & 18 & 18 & 5 & 5 & 4 & 4 & $\times$ & $\times$ \\
while-if  & 1 & 20 & 20 & 55 & 53 & 5 & 6 & 4 & 4 & \checkmark & \checkmark \\
\rowcolor{red!10}while-if-e & 1 & 17 & 17 & 42 & 40 & 5 & 5 & 5 & 5 & $\times$ & $\times$ \\
\bottomrule		
		\end{tabular}
		}
				%\vspace{-4mm}
\end{table}

\subsection{\peqcheck on Sequential Versions}
Our second experiment, which uses the second benchmark set, demonstrates that the \peqcheck approach is not restricted to parallelization.
Table~\ref{tab:SeqStudy} shows our evaluation results for the pairs of sequential programs from our second benchmark set. 
The structure of Tab.~\ref{tab:SeqStudy} is similar to Tab.~\ref{tab:OpenMPStudy}.

Again, we first look at the results for \peqcheck with local code segments (configuration seg).
Studying Tab.~\ref{tab:SeqStudy}, we recognize that the sequential verification tasks are always larger than the two input programs and the generation of the verification task takes a few seconds, i.e., it is fast.
Looking at the verification (columns~$t_\mathrm{seg}^V$ and $s_\mathrm{seg}$), we observe that the verifier \cpachecker times out (status~TO) for 5 of 18~tasks.
In addition, we notice that the verifier~\cpachecker reports an incorrect status for the tasks \texttt{barthe} and \texttt{loop3}.
Our inspection reveals that \cpachecker correctly detects the inequivalence of the generated tasks, but the two programs execute the two segments with a restricted set of inputs and are therefore equivalent.
Although \peqcheck incorrectly detects equivalence in two cases, it correctly reports equivalence (status \checkmark) for 5 of the 12~equivalent tasks.
Moreover, \peqcheck correctly reports status~$\times$, i.e., inequivalence, for all pairs of programs that are inequivalent (i.e., benchmark task with suffix \texttt{-e}, which are highlighted in light red).
Thus, localized equivalence checking with \peqcheck can also correctly detect (in)equivalence of sequential programs.

Next, we compare \peqcheck with localized equivalence checking (configuration~seg) against all at once checking (configuration~all).
In the second set of benchmarks, which we currently study, configuration~all is identical to checking equivalence of functions, which is also typically done by the approach in the related work.
First, let us look at the generated verification tasks.
Again, the generation times are similar.
Furthermore, their sizes do not differ significantly because also the localized segments contain most of the functions' code and often only leave out the declaration and initialization of variables. 
Looking at the status columns, we observe that either the status is the same or configuration~all times out while configuration~seg returns a result (either status $\times$ or \checkmark).
If both configurations do not time out, the verification times are similar.
The only exception is \texttt{bug15}, for which the verification in configuration~all profits from restricted input values.
We conclude that also for our sequential examples, localized equivalence checking is beneficial.
\section{Related Work} 
Functional equivalence checking is a particular instance of relational program verification~\cite{DBLP:conf/popl/Benton04,DBLP:journals/tcs/Yang07}. 
To verify relational properties between two programs, Barthe et al.\ proposes to construct and verify product programs~\cite{ProductPrograms}. 
A product program merges the two input programs such that synchronous steps are executed in lockstep.
Thus, product programs integrate the two programs tighter than a sequential composition. 

Nevertheless, many approaches are tailored to check function equivalence. 
There exist model-based approaches~\cite{PromelaFEQ,TASS,PRESGen,DBLP:conf/cav/VerdoolaegeJB09,DBLP:conf/date/ShashidharBCJ05} that translate the two programs, which should be proven equivalent, into models and inspect model equivalence.
Simulation-based approaches, e.g., \cite{CoVaC,DBLP:conf/oopsla/0001SCA13,DBLP:conf/aplas/DahiyaB17,DBLP:conf/pldi/ChurchillP0A19}, try to establish a (bi)\-simu\-lation relation between the two programs. 
Pathg~\cite{OpenMPConsistency} is an eqivalence checker for OpenMP program, which checks that the OpenMP program is equivalent with its sequential version (the program without the OpenMP directives). 
Pathg assumes that only race conditions may cause inequivalence and, hence, applies symbolic simulation on segments with race conditions to inspect whether the races affect the output.
Fractal symbolic analysis~\cite{FractalSymboliAna} transforms the two programs, which should be shown to be equivalent, into two simpler programs.
The transformation ensures that equivalence of the simplified programs implies the equivalence of the original programs.
After the transformation, the guarded symbolic expressions (descriptions of the effect of a program on a variable) of all variables that are modified and live are compared.
R\^{e}ve~\cite{Reve} translates the equivalence of two deterministic functions into Horn constraints with uninterpreted symbols. 
DSE~\cite{DiffSymExe} and ARDiff~\cite{ARDiff} employ symbolic execution to compute function summaries.
To determine functional equivalence between two functions, they check the logical equivalence of their summaries. 
While DSE and ARDiff abstract certain common code regions by uninterpreted functions, we analyze the equivalence of code regions that differ.

A widely-used idea, which we apply as well, is to encode the equivalence check as a program~\cite{RV-DAC,RV-Journal,UCKLEE,SymDiff,RIE,RV-thread,CIVL,DBLP:journals/ijpp/AbadiKPV19}.
Often, the encoded program initializes the same inputs (e.g., global variables, parameters) with equal, but non-deterministic values, then sequentially executes the two functions, and finally inspects whether the same output variables (e.g., global variables, return values) have identical values.
Regression verification~\cite{RV-DAC,RV-Journal}, which is realized in the tool~RVT, Sym\-Diff~\cite{SymDiff}, and RIE~\cite{RIE} employ this idea for each matched pair of sequential functions. 
RVT and SymDiff utilize uninterpreted functions for function calls\footnote{RVT only replaces calls that are recursive or already proven equivalent.} 
and RIE considers function summaries.
Similar to \peqcheck, RVT~\cite{RV-DAC,RV-Journal} and RIE~\cite{RIE} rename variables and initialize matching variables with equal input values.
In contrast, Sym\-Diff~\cite{SymDiff} and UC-Klee~\cite{UCKLEE} save and load the initial state and store the state after each function execution.
Moreover, RIE~\cite{RIE} uses heap equivalence instead of equivalence of output variables.
While the RVT, SymDiff, RIE, and UC-Klee focus on equivalence between two sequential programs, there also exist approaches checking parallel programs.
Chaki et al.~\cite{RV-thread} suggest an approach that encodes equivalence of multi-threaded programs into one sequential verification task per function pair. 
CIVL~\cite{CIVL} supports functional equivalence checking for concurrent programs using pthreads, OpenMP, MPI, etc.
Given the input and output variables, CIVL encodes the functional equivalence check in a single composite program, which equalizes the inputs.
Abadi et al.~\cite{DBLP:journals/ijpp/AbadiKPV19} propose an approach that encodes the functional equivalence of a sequential and a parallelized code segment.
Their approach enfolds the two code segments into two separate functions.
The two functions share the input variables, but use separate output variables.
Input and output variable are determined by dataflow analyses, which are not further specified
However, the approach requires that the input and output variables of the two segments are identical.

Some approaches~\cite{DBLP:conf/sas/PartushY13,DBLP:conf/oopsla/PartushY14,DiffSymExe,conditional-equivalence} go further than equivalence checking and determine when the two programs are equivalent.
\section{Conclusion}
Program refactorings are performed regularly in software development and ensuring that a refactoring is performed correctly, i.e., ensuring that the refactoring is behavior preserving, is crucial.
To deal with the problem that often no formal specification is available, one typically uses the original program as behavior specification and checks whether the original program and the refactored program are functionally equivalent.

We presented \peqcheck, an approach to check functional equivalence of original and refactored program, and proved its soundness.
\peqcheck is motivated by rather local OpenMP parallelizations.
Therefore, it reduces equivalence checking to generating and verifying one verification task per (parallelized) code segment. 
During the generation task, it furthermore considers the context of the code segments, i.e., how variables are used in and after the code segments.
Both, localized checking and context-awareness reduce the complexity of the verification task, which is confirmed by our experiments.
While designed for equivalence checking of sequential programs and their parallelization, \peqcheck is not limited to these checks.
As demonstrated, it can also be applied to pairs of sequential programs.
Although we have seen in our experiments that \peqcheck is incomplete (overapproximation of the input space may lead to a missed equivalence), a problem that common in for modular verification approaches, our experiments show that \peqcheck is feasible and that it can be beneficial. 

\subsubsection*{Acknowledgements}
This work was funded by the Hessian LOEWE initiative within the Software-Factory 4.0 project.
%
% ---- Bibliography ----
%
% BibTeX users should specify bibliography style 'splncs04'.
% References will then be sorted and formatted in the correct style.
%
 \bibliographystyle{splncs04}
 \bibliography{literature}

\begin{thebibliography}{10}
\providecommand{\url}[1]{\texttt{#1}}
\providecommand{\urlprefix}{URL }
\providecommand{\doi}[1]{https://doi.org/#1}

\bibitem{DBLP:journals/ijpp/AbadiKPV19}
Abadi, M., Keidar{-}Barner, S., Pidan, D., Veksler, T.: Verifying parallel code
  after refactoring using equivalence checking. Int. J. Parallel Program.
  \textbf{47}(1),  59--73 (2019),
  \url{https://doi.org/10.1007/s10766-017-0548-4}

\bibitem{ARDiff}
Badihi, S., Akinotcho, F., Li, Y., Rubin, J.: {ARDiff}: scaling program
  equivalence checking via iterative abstraction and refinement of common code.
  In: Proc. FSE. pp. 13--24. {ACM} (2020),
  \url{https://doi.org/10.1145/3368089.3409757}

\bibitem{PRESGen}
Bandyopadhyay, S., Banerjee, K.: {PRESGen}: {A} fully automatic equivalence
  checker for validating optimizing and parallelizing transformations. In:
  Proc.\ SEM4HPC@HPDC. pp. 13--20. {ACM} (2017),
  \url{https://doi.org/10.1145/3085158.3086158}

\bibitem{ProductPrograms}
Barthe, G., Crespo, J.M., Kunz, C.: Relational verification using product
  programs. In: Proc.\ {FM}. pp. 200--214. LNCS~6664, Springer (2011),
  \url{https://doi.org/10.1007/978-3-642-21437-0\_17}

\bibitem{DBLP:conf/popl/Benton04}
Benton, N.: Simple relational correctness proofs for static analyses and
  program transformations. In: Proc.\ {POPL}. pp. 14--25. {ACM} (2004),
  \url{https://doi.org/10.1145/964001.964003}

\bibitem{CPACHECKER}
Beyer, D., Keremoglu, M.E.: \textsc{CPAchecker}: A tool for configurable
  software verification. In: Proc.\ CAV. pp. 184--190. LNCS~6806, Springer
  (2011), \url{https://doi.org/10.1007/978-3-642-22110-1_16}

\bibitem{PrecisionReuse}
Beyer, D., L{\"{o}}we, S., Novikov, E., Stahlbauer, A., Wendler, P.: Precision
  reuse for efficient regression verification. In: Proc. FSE. pp. 389--399.
  {ACM} (2013), \url{https://doi.org/10.1145/2491411.2491429}

\bibitem{RV-thread}
Chaki, S., Gurfinkel, A., Strichman, O.: Regression verification for
  multi-threaded programs. In: Proc. {VMCAI}. pp. 119--135. LNCS~7148, Springer
  (2012), \url{https://doi.org/10.1007/978-3-642-27940-9\_9}

\bibitem{DBLP:conf/pldi/ChurchillP0A19}
Churchill, B.R., Padon, O., Sharma, R., Aiken, A.: Semantic program alignment
  for equivalence checking. In: Proc.\ {PLDI}. pp. 1027--1040. {ACM} (2019),
  \url{https://doi.org/10.1145/3314221.3314596}

\bibitem{DBLP:conf/aplas/DahiyaB17}
Dahiya, M., Bansal, S.: Black-box equivalence checking across compiler
  optimizations. In: Proc.\ {APLAS}. pp. 127--147. LNCS~10695, Springer (2017),
  \url{https://doi.org/10.1007/978-3-319-71237-6\_7}

\bibitem{SWVerificationSurvey}
D'Silva, V., Kroening, D., Weissenbacher, G.: A survey of automated techniques
  for formal software verification. TCAD  \textbf{27}(7),  1165--1178 (2008),
  \url{https://doi.org/10.1109/TCAD.2008.923410}

\bibitem{Reve}
Felsing, D., Grebing, S., Klebanov, V., R{\"{u}}mmer, P., Ulbrich, M.:
  Automating regression verification. In: Proc.\ ASE. pp. 349--360. {ACM}
  (2014), \url{https://doi.org/10.1145/2642937.2642987}

\bibitem{RefactoringFowler}
Fowler, M.: Refactoring - Improving the Design of Existing Code. Addison-Wesley
  (1999), \url{http://martinfowler.com/books/refactoring.html}

\bibitem{SCAMVerifyRefactoring}
Garrido, A., Meseguer, J.: Formal specification and verification of {Java}
  refactorings. In: Proc.\ SCAM. pp. 165--174. IEEE (2006),
  \url{https://doi.ieeecomputersociety.org/10.1109/SCAM.2006.16}

\bibitem{RV-DAC}
Godlin, B., Strichman, O.: Regression verification. In: Proc. DAC. pp.
  466--471. {ACM} (2009), \url{https://doi.org/10.1145/1629911.1630034}

\bibitem{RV-Journal}
Godlin, B., Strichman, O.: Regression verification: proving the equivalence of
  similar programs. STVR  \textbf{23}(3),  241--258 (2013),
  \url{https://doi.org/10.1002/stvr.1472}

\bibitem{ExtremeMC}
Henzinger, T.A., Jhala, R., Majumdar, R., Sanvido, M.A.A.: Extreme model
  checking. In: Verification: Theory and Practice, Essays Dedicated to Zohar
  Manna on the Occasion of His 64th Birthday. pp. 332--358. LNCS~2772, Springer
  (2003), \url{https://doi.org/10.1007/978-3-540-39910-0\_16}

\bibitem{PEQcheck}
Jakobs, M.C.: {PEQcheck}: Localized and context-aware checking of functional
  equivalence. In: Proc. FormaliSE. {IEEE} (2021)

\bibitem{conditional-equivalence}
Kawaguchi, M., Lahiri, S., Rebelo, H.: Conditional equivalence. Tech. rep.
  (2010),
  \url{https://www.microsoft.com/en-us/research/publication/conditional-equivalence/}

\bibitem{ParameterizedEQCheck}
Kundu, S., Tatlock, Z., Lerner, S.: Proving optimizations correct using
  parameterized program equivalence. In: Proc.\ {PLDI}. pp. 327--337. {ACM}
  (2009), \url{https://doi.org/10.1145/1542476.1542513}

\bibitem{SymDiff}
Lahiri, S.K., Hawblitzel, C., Kawaguchi, M., Reb{\^{e}}lo, H.: {SYMDIFF:} {A}
  language-agnostic semantic diff tool for imperative programs. In: Proc.\ CAV.
  pp. 712--717. LNCS~7358, Springer (2012),
  \url{https://doi.org/10.1007/978-3-642-31424-7\_54}

\bibitem{FractalSymboliAna}
Mateev, N., Menon, V., Pingali, K.: Fractal symbolic analysis. In: Proc.\
  {ICS}. pp. 38--49. {ACM} (2001), \url{https://doi.org/10.1145/377792.377804}

\bibitem{RefactoringSurvey}
Mens, T., Tourw{\'{e}}, T.: A survey of software refactoring. TSE
  \textbf{30}(2),  126--139 (2004),
  \url{https://doi.org/10.1109/TSE.2004.1265817}

\bibitem{Z3}
de~Moura, L.M., Bj{\o}rner, N.: {Z3:} an efficient {SMT} solver. In: Proc.\
  TACAS. pp. 337--340. LNCS~4963, Springer (2008),
  \url{https://doi.org/10.1007/978-3-540-78800-3\_24}

\bibitem{ProgramAnalysis}
Nielson, F., Nielson, H.R., Hankin, C.: Principles of program analysis.
  Springer (1999), \url{https://doi.org/10.1007/978-3-662-03811-6}

\bibitem{OpenMP}
OpenMP: {OpenMP} application programming interface (version 5.0). Tech. rep.,
  OpenMP Architecture Review Board (2018),
  \url{https://www.openmp.org/specifications/}

\bibitem{DBLP:conf/sas/PartushY13}
Partush, N., Yahav, E.: Abstract semantic differencing for numerical programs.
  In: Proc.\ SAS. pp. 238--258. LNCS~7935, Springer (2013),
  \url{https://doi.org/10.1007/978-3-642-38856-9\_14}

\bibitem{DBLP:conf/oopsla/PartushY14}
Partush, N., Yahav, E.: Abstract semantic differencing via speculative
  correlation. In: Proc.\ {OOPSLA}. pp. 811--828. {ACM} (2014),
  \url{https://doi.org/10.1145/2660193.2660245}

\bibitem{DiffSymExe}
Person, S., Dwyer, M.B., Elbaum, S.G., Pasareanu, C.S.: Differential symbolic
  execution. In: Proc.\ FSE. pp. 226--237. {ACM} (2008),
  \url{https://doi.org/10.1145/1453101.1453131}

\bibitem{DiSE}
Person, S., Yang, G., Rungta, N., Khurshid, S.: Directed incremental symbolic
  execution. In: Proc.\ {PLDI}. pp. 504--515. {ACM} (2011),
  \url{https://doi.org/10.1145/1993498.1993558}

\bibitem{ROSE}
Quinlan, D., Liao, C.: The {ROSE} source-to-source compiler infrastructure. In:
  Cetus users and compiler infrastructure workshop, in conjunction with PACT.
  vol.~2011, pp.~1--3. Citeseer (2011)

\bibitem{UCKLEE}
Ramos, D.A., Engler, D.R.: Practical, low-effort equivalence verification of
  real code. In: Proc.\ CAV. pp. 669--685. LNCS~6806, Springer (2011),
  \url{https://doi.org/10.1007/978-3-642-22110-1\_55}

\bibitem{EvolCheck}
Sery, O., Fedyukovich, G., Sharygina, N.: Incremental upgrade checking by means
  of interpolation-based function summaries. In: Proc.\ {FMCAD}. pp. 114--121.
  {IEEE} (2012), \url{http://ieeexplore.ieee.org/document/6462563/}

\bibitem{DBLP:conf/oopsla/0001SCA13}
Sharma, R., Schkufza, E., Churchill, B.R., Aiken, A.: Data-driven equivalence
  checking. In: Proc.\ {OOPSLA}. pp. 391--406. {ACM} (2013),
  \url{https://doi.org/10.1145/2509136.2509509}

\bibitem{DBLP:conf/date/ShashidharBCJ05}
Shashidhar, K.C., Bruynooghe, M., Catthoor, F., Janssens, G.: Functional
  equivalence checking for verification of algebraic transformations on
  array-intensive source code. In: Proc.\ {DATE}. pp. 1310--1315. {IEEE}
  (2005), \url{https://doi.org/10.1109/DATE.2005.163}

\bibitem{PromelaFEQ}
Siegel, S.F., Mironova, A., Avrunin, G.S., Clarke, L.A.: Using model checking
  with symbolic execution to verify parallel numerical programs. In: Proc.\
  {ISSTA}. pp. 157--168. {ACM} (2006),
  \url{https://doi.org/10.1145/1146238.1146256}

\bibitem{CIVL}
Siegel, S.F., Zheng, M., Luo, Z., Zirkel, T.K., Marianiello, A.V., Edenhofner,
  J.G., Dwyer, M.B., Rogers, M.S.: {CIVL:} the concurrency intermediate
  verification language. In: Proc.\ {SC}. pp. 61:1--61:12. {ACM} (2015),
  \url{https://doi.org/10.1145/2807591.2807635}

\bibitem{FEVS}
Siegel, S.F., Zirkel, T.K.: {FEVS:} {A} functional equivalence verification
  suite for high-performance scientific computing. Mathematics in Computer
  Science  \textbf{5}(4),  427--435 (2011),
  \url{https://doi.org/10.1007/s11786-011-0101-6}

\bibitem{TASS}
Siegel, S.F., Zirkel, T.K.: {TASS:} the toolkit for accurate scientific
  software. Mathematics in Computer Science  \textbf{5}(4),  395--426 (2011),
  \url{https://doi.org/10.1007/s11786-011-0100-7}

\bibitem{DBLP:conf/pepm/SultanaT08}
Sultana, N., Thompson, S.J.: Mechanical verification of refactorings. In:
  Proc.\ {PEPM}. pp. 51--60. {ACM} (2008),
  \url{https://doi.org/10.1145/1328408.1328417}

\bibitem{DBLP:conf/cav/VerdoolaegeJB09}
Verdoolaege, S., Janssens, G., Bruynooghe, M.: Equivalence checking of static
  affine programs using widening to handle recurrences. In: Proc.\ CAV. pp.
  599--613. LNCS~5643, Springer (2009),
  \url{https://doi.org/10.1007/978-3-642-02658-4\_44}

\bibitem{RIE}
Wood, T., Drossopoulou, S., Lahiri, S.K., Eisenbach, S.: Modular verification
  of procedure equivalence in the presence of memory allocation. In: Proc.\
  {ESOP}. pp. 937--963. LNCS~10201, Springer (2017),
  \url{https://doi.org/10.1007/978-3-662-54434-1\_35}

\bibitem{RegresionMC}
Yang, G., Dwyer, M.B., Rothermel, G.: Regression model checking. In: Proc.\
  {ICSM}. pp. 115--124. {IEEE} (2009),
  \url{https://doi.org/10.1109/ICSM.2009.5306334}

\bibitem{DBLP:journals/tcs/Yang07}
Yang, H.: Relational separation logic. TCS  \textbf{375}(1-3),  308--334
  (2007), \url{https://doi.org/10.1016/j.tcs.2006.12.036}

\bibitem{Test}
Yoo, S., Harman, M.: Regression testing minimization, selection and
  prioritization: a survey. Softw. Test. Verification Reliab.  \textbf{22}(2),
  67--120 (2012), \url{https://doi.org/10.1002/stv.430}

\bibitem{OpenMPConsistency}
Yu, F., Yang, S., Wang, F., Chen, G., Chan, C.: Symbolic consistency checking
  of {OpenMP} parallel programs. In: Proc.\ {LCTES}. pp. 139--148. {ACM}
  (2012), \url{https://doi.org/10.1145/2248418.2248438}

\bibitem{CoVaC}
Zaks, A., Pnueli, A.: {CoVaC}: Compiler validation by program analysis of the
  cross-product. In: Proc.\ {FM}. pp. 35--51. LNCS~5014, Springer (2008),
  \url{https://doi.org/10.1007/978-3-540-68237-0\_5}

\end{thebibliography}

\appendix
\section{Proofs}
This appendix contains the proofs of our theorems.

\subsection{Auxiliary Lemmas on Program Executions}

\begin{lemma}\label{lem:ubBehave}
Let \(S\) be a program.
\[\begin{array}{l}\forall p=(S_0,\sigma_0)\stackrel{op_1}{\rightarrow}\dots\stackrel{op_n}{\rightarrow}(S_n, \sigma_n)\in ex(S):\forall \sigma'_0\in\Sigma: \sigma_0=_{|_{\mathcal{UB}_{ex}(p)}}\sigma'_0:\\\exists(S_0,\sigma'_0)\stackrel{op_1}{\rightarrow}\dots\stackrel{op_n}{\rightarrow}(S_n,\sigma'_n)\in ex(S): \sigma_n=_{|_{\{v\in\mathcal{V}\mid \sigma_0(v)=\sigma'_0(v)\vee \exists 1\leq i\leq n: op_i\equiv v:=expr;\}}}\sigma'_n\end{array}\]
\end{lemma}

\begin{proof}
Prove by induction on the length of the executions.
Show for all programs~\(S\) that for all executions of length $n$ the following holds:
\[\begin{array}{l}\forall p=(S_0,\sigma_0)\stackrel{op_1}{\rightarrow}\dots\stackrel{op_n}{\rightarrow}(S_n, \sigma_n)\in ex(S):\forall \sigma'_0\in\Sigma: \sigma_0=_{|_{\mathcal{UB}_{ex}(p)}}\sigma'_0:\\\exists(S_0,\sigma'_0)\stackrel{op_1}{\rightarrow}\dots\stackrel{op_n}{\rightarrow}(S_n,\sigma'_n)\in ex(S): \sigma_n=_{|_{\{v\in\mathcal{V}\mid \sigma_0(v)=\sigma'_0(v)\vee \exists 1\leq i\leq n: op_i\equiv v:=expr;\}}}\sigma'_n\end{array}\]

\textbf{Base case (n=0):} 
Let \(S\) and \(\sigma'_0\) be arbitrary. 
By definition, \((S_0,\sigma'_0)\in ex(S)\). Since \(n=0\), we have \(\{v\in\mathcal{V}\mid \sigma_0(v)=\sigma'_0(v)\vee \exists 1\leq i\leq n: op_i\equiv v:=expr;\}=\{v\in\mathcal{V}\mid \sigma_0(v)=\sigma'_0(v)\}\) and \(\sigma_0=\sigma_n\) and \(\sigma'_0=\sigma'_n\). 
The hypothesis follows.

\textbf{Step case (n-1$\rightarrow$n):}
Let \(S\), \(p=(S_0,\sigma_0)\stackrel{op_1}{\rightarrow}\dots\stackrel{op_n}{\rightarrow}(S_n, \sigma_n)\in ex(S)\) and \(\sigma'_0\in\Sigma\) with \(\sigma_0=_{|_{\mathcal{UB}_{ex}(p)}}\sigma'_0\) be arbitrary.
By definition, \(p'=(S_1,\sigma_1)\dots\stackrel{op_n}{\rightarrow}(S_n, \sigma_n)\in ex(S_1)\).
Consider three cases:

\textit{Case \(op_1\equiv \textbf{nop}\)}
Due to semantics, \(S_1\) starts with \(E\) or \([E \| \dots \| E]\), \(\sigma_0=\sigma_1\), and \((S_0,\sigma'_0)\stackrel{\textbf{nop}}{\rightarrow}(S_1,\sigma'_0)\).
Furthermore, \(\{v\in\mathcal{V}\mid \sigma_0(v)=\sigma'_0(v)\vee \exists 1\leq i\leq n: op_i\equiv v:=expr;\}=\{v\in\mathcal{V}\mid \sigma_0(v)=\sigma'_0(v)\vee \exists 2\leq i\leq n: op_i\equiv v:=expr;\}\) and by definition  \(\mathcal{UB}_{ex}(p')\subseteq\mathcal{UB}_{ex}(p)\).
By induction hypothesis, \(\exists(S_1,\sigma'_1)\stackrel{op_2}{\rightarrow}\dots\stackrel{op_n}{\rightarrow}(S_n,\sigma'_n)\in ex(S_1): \sigma'_1=\sigma'_0\wedge \sigma_n=_{|_{\{v\in\mathcal{V}\mid \sigma_0(v)=\sigma'_0(v)\vee \exists 2\leq i\leq n: op_i\equiv v:=expr;\}}}\sigma'_n\). The induction hypothesis follows.

\textit{Case \(op_1\equiv bexpr\)\footnote{Since \(\neg bexpr\) is boolean expression, too, it is covered by this case.}}
By definition, \(\mathcal{V}(bexpr)\subseteq \mathcal{UB}_{ex}(p)\). Thus, \(\sigma_0(bexpr)=\sigma'_0(bexpr)\).
Due to semantics, \((S_0,\sigma'_0)\stackrel{bexpr}{\rightarrow}(S_1,\sigma'_0)\) and \((S_0,\sigma_0)\stackrel{bexpr}{\rightarrow}(S_1,\sigma_1)\) implies \(\sigma_0=\sigma_1\).
Furthermore, \(\{v\in\mathcal{V}\mid \sigma_0(v)=\sigma'_0(v)\vee \exists 1\leq i\leq n: op_i\equiv v:=expr;\}=\{v\in\mathcal{V}\mid \sigma_0(v)=\sigma'_0(v)\vee \exists 2\leq i\leq n: op_i\equiv v:=expr;\}\) and by definition  \(\mathcal{UB}_{ex}(p')\subseteq\mathcal{UB}_{ex}(p)\).
By induction hypothesis, \(\exists(S_1,\sigma'_1)\stackrel{op_2}{\rightarrow}\dots\stackrel{op_n}{\rightarrow}(S_n,\sigma'_n)\in ex(S_1): \sigma'_1=\sigma'_0\wedge \sigma_n=_{|_{\{v\in\mathcal{V}\mid \sigma_0(v)=\sigma'_0(v)\vee \exists 2\leq i\leq n: op_i\equiv v:=expr;\}}}\sigma'_n\). The induction hypothesis follows.

\textit{Case \(op_1\equiv v:=aexpr;\)} 
By definition, we get \(\mathcal{V}(aexpr)\subseteq \mathcal{UB}_{ex}(p)\). Thus, \(\sigma_0(aexpr)=\sigma'_0(aexpr)\).
Due to semantics, \((S_0,\sigma_0)\stackrel{v:=aexpr}{\rightarrow}(S_1,\sigma_1)\) implies that \(\sigma_1=\sigma_0[v:=\sigma_0(aexpr)]\) and furthermore, \((S_0,\sigma'_0)\stackrel{v:=aexpr}{\rightarrow}(S_1,\sigma'_1)\) with \(\sigma'_1=\sigma'_0[v:=\sigma'_0(aexpr)]=\sigma'_0[v:=\sigma_0(aexpr)]\).
Thus, we infer the following: \(\sigma_1=_{|_{\{v\in\mathcal{V}\mid \sigma_0(v)=\sigma'_0(v)\vee \exists 1\leq i\leq 1: op_i\equiv v:=expr;\}}}\sigma'_1\).
Moreover, \(\{v\in\mathcal{V}\mid \sigma_0(v)=\sigma'_0(v)\) \(\vee \exists 1\leq i\leq n: op_i\equiv v:=expr;\}=\{v\in\mathcal{V}\mid \sigma_1(v)=\sigma'_1(v)\vee \exists 2\leq i\leq n: op_i\equiv v:=expr;\}\) and by definition  \(\mathcal{UB}_{ex}(p')\subseteq(\mathcal{UB}_{ex}(p)\cup\{v\})\).
By induction hypothesis, \(\exists(S_1,\sigma'_1)\stackrel{op_2}{\rightarrow}\dots\stackrel{op_n}{\rightarrow}(S_n,\sigma'_n)\in ex(S_1): \sigma_n=_{|_{\{v\in\mathcal{V}\mid \sigma_1(v)=\sigma'_1(v)\vee \exists 2\leq i\leq n: op_i\equiv v:=expr;\}}}\sigma'_n\). The hypothesis follows.

\end{proof}

\begin{lemma}\label{lem:change}
Let \(S\) be a program. \[\forall\sigma\in\Sigma: (S,\sigma)\rightarrow^*(E,\sigma')\in ex(S)\implies \sigma=_{|_{V\setminus\mathcal{M}(S)}}\sigma'\]
\end{lemma}

\begin{proof}
Prove by induction on the length of the executions.
Show for all programs~\(S\) that \(\forall\sigma\in\Sigma: (S,\sigma)\rightarrow^n(E,\sigma')\in ex(S)\implies \sigma=_{|_{V\setminus\mathcal{M}(S)}}\sigma'\).

\textbf{Base case (n=0)}
Then, \(S=E\) and \(\sigma=\sigma_0\).
The hypothesis follows.

\textbf{Step case (n-1$\rightarrow$n):}
Let \(S\) be arbitrary and \(p=(S_0,\sigma_0)\stackrel{op_1}{\rightarrow}\dots\stackrel{op_n}{\rightarrow}(E, \sigma_n)\in ex(S)\).
 By definition \(p'=(S_1,\sigma_1)\stackrel{op_2}{\rightarrow}\dots\stackrel{op_n}{\rightarrow}(E, \sigma_n)\in ex(S_1)\) and \(\mathcal{M}(S_1)\subseteq\mathcal{M}(S)\).
Consider two cases:

\textit{Case \(op_1\equiv bexpr\) or \(op_1=\textbf{nop}\)}
Due to semantics, \(\sigma_0=\sigma_1\).
By induction, \(\sigma_1=_{|_{V\setminus\mathcal{M}(S_1)}}\sigma_n\).
Hence, \(\sigma_0=_{|_{V\setminus\mathcal{M}(S)}}\sigma_n\).

\textit{Case \(op_1\equiv v:=aexpr;\)} 
Due to semantics, \(\sigma_1=\sigma_0[v:=\sigma_0(aexpr)]\) and, hence, \(\sigma_0=_{|_{\mathcal{V}\setminus\{v\}}}\sigma_1\).
By induction, \(\sigma_1=_{|_{V\setminus\mathcal{M}(S_1)}}\sigma_n\).
Since  \(\mathcal{M}(S_1)\cup\{v\}\subseteq\mathcal{M}(S)\), we conclude \(\sigma_0=_{|_{V\setminus\mathcal{M}(S)}}\sigma_n\).

\end{proof}

\begin{lemma}\label{lem:binSing}
Let \(S\) be a program and \(\rho\) be a renaming function.
\[\begin{array}{l}\forall (S_0,\sigma_0)\stackrel{op_1}{\rightarrow}(S_1,\sigma_1)\in ex(S):\\\exists (\mathcal{R}(S_0,\rho),\rho(\sigma_0))\stackrel{\mathcal{R}(op_1,\rho)}{\rightarrow}(S_1,\rho),\rho(\sigma_1))\in ex(\mathcal{R}(S,\rho))\end{array}\]
\end{lemma}

\begin{proof}
Prove by induction over the length \(n\) of the derivation of \((S_0,\sigma_0)\stackrel{op_1}{\rightarrow}(S_1,\sigma_1)\) that there exists \((\mathcal{R}(S_0,\rho),\rho(\sigma_0))\stackrel{\mathcal{R}(op_1,\rho)}{\rightarrow}(S_1,\rho),\rho(\sigma_1))\).

\textbf{Base case (n=1)}
Due to semantics, \(S_0\) is an assignment, assert statement, if- or while-statement, empty parallel statement, or a sequence starting with an empty program.
Consider eight cases.

\emph{Case 1 (\(S_0\equiv v:=_\ell aexpr;\))}
Then, \(\mathcal{R}(S_0,\rho)=\rho(v):=_\ell \mathcal{R}(aexpr,\rho);\).
Due to the semantics, we conclude that \((S_0,\sigma_0)\stackrel{v:=aexpr;}{\rightarrow}(E,\sigma_0[v:=\sigma_0(expr)])\) and  \((\mathcal{R}(S,\rho), \rho(\sigma_0))\stackrel{\rho(v):=\mathcal{R}(aexpr,\rho);}{\rightarrow}(E,\rho(\sigma_0)[\rho(v):=\rho(\sigma_0)(\mathcal{R}(aexpr,\rho))])\in ex(\mathcal{R}(S_0,\rho))\)
By definition of \(\rho(\sigma_0)\), \(\rho(\sigma_0)(\mathcal{R}(aexpr,\rho))=\sigma_0(aexpr)\).
We conclude that \(\rho(\sigma_0)[\rho(v):=\rho(\sigma_0)(\mathcal{R}(aexpr,\rho))]=\rho(\sigma_0)[\rho(v):=\sigma_0(aexpr)]=\rho(\sigma_0[v:=\sigma_0(aexpr)])\).
Since \(\mathcal{R}(E,\rho)=E\) and \(\mathcal{R}(S_0,\rho)=\rho(v):=_\ell \mathcal{R}(aexpr,\rho);\), the hypothesis follows.

\emph{Case 2 (\(S_0\equiv \mathbf{assert}_\ell~bexpr;\))}
Then, \(\mathcal{R}(S_0,\rho)=\mathbf{assert_\ell}~\mathcal{R}(bexpr,\rho);\).
Due to the semantics, \((S_0,\sigma_0)\stackrel{bexpr}{\rightarrow}(S_1, \sigma_1)\) implies \(\sigma_1=\sigma_0\), \(\sigma_0(bexpr)=true\), and \(S_1=E\).
Since \(true=\sigma_0(bexpr)=\rho(\sigma_0)(\mathcal{R}(bexpr,\rho))\), we conclude from the semantics that \((\mathcal{R}(S,\rho),\rho(\sigma_0))\stackrel{\mathcal{R}(bexpr,\rho)}{\rightarrow}(E, \sigma'_1)\in ex(\mathcal{R}(S_0,\rho))\) and \(\sigma'_1=\rho(\sigma_0)\).
Since \(\mathcal{R}(E,\rho)=E\), the hypothesis follows.

\emph{Case 3 (\(S_0\equiv \mathbf{if}_\ell~bexpr~\mathbf{then}~ S'~\mathbf{else}~S'' \wedge\sigma_0(bexpr)\))}
Then, \(\mathcal{R}(S,\rho)=\mathbf{if}_\ell~\mathcal{R}(bexpr,\rho)~\mathbf{then}~ \mathcal{R}(S',\rho)~\mathbf{else}~\mathcal{R}(S'',\rho)\).
Due to the semantics, \((S_0,\sigma_0)\stackrel{bexpr}{\rightarrow}(S_1, \sigma_1)\) with \(\sigma_1=\sigma_0\), and \(S_1=S'\).
Since \(true=\sigma_0(bexpr)=\rho(\sigma_0)(\mathcal{R}(bexpr,\rho))\), we conclude from the semantics that  \((\mathcal{R}(S_0,\rho),\rho(\sigma_0))\stackrel{\mathcal{R}(bexpr,\rho)}{\rightarrow}(\mathcal{R}(S',\rho), \rho(\sigma_0))\).
Since \(\sigma_0=\sigma_1\), the induction hypothesis follows. 

\emph{Case 4 (\(S_0\equiv  \mathbf{if}_\ell~bexpr~\mathbf{then}~ S'~\mathbf{else}~S'' \wedge\neg\sigma_0(bexpr)\))}
%Then, \(\mathcal{R}(S,\rho)=\mathbf{if}_\ell~\mathcal{R}(bexpr,\rho)~\mathbf{then}~ \mathcal{R}(S',\rho)~\mathbf{else}~\mathcal{R}(S'',\rho)\).
Analogously to case~3.

\emph{Case 5 (\(S_0\equiv \mathbf{while}_\ell~bexpr~\mathbf{do}~ S' \wedge\sigma_0(bexpr)\))}
We infer that\(\mathcal{R}(S,\rho)=\mathbf{while}_\ell~\rho(bexpr)~\mathbf{do}~ \mathcal{R}(S',\rho)\).
Due to the semantics, \((S_0,\sigma_0)\stackrel{bexpr}{\rightarrow}(S_1, \sigma_1)\) with \(\sigma_1=\sigma_0\), and \(S_1=S';S_0\).
By definition of \(\rho(\sigma_0)\), \(\rho(\sigma_0)(\mathcal{R}(bexpr,\rho))=\sigma_0(bexpr)=true\).
Thus, \((\mathcal{R}(S_0,\rho),\rho(\sigma_0))\stackrel{\mathcal{R}(bexpr,\rho)}{\rightarrow}(\mathcal{R}(S',\rho);\mathcal{R}(S_0,\rho), \rho(\sigma_0))\).
Since \(\mathcal{R}(S_1,\rho)=\mathcal{R}(S';S_0,\rho)=\mathcal{R}(S',\rho);\mathcal{R}(S_0,\rho)\), the induction hypothesis follows. 

\emph{Case 6 (\(S_0\equiv \mathbf{while}_\ell~bexpr~\mathbf{do}~ S' \wedge\neg\sigma_0(bexpr)\))}
We infer that \(\mathcal{R}(S,\rho)=\mathbf{while}_\ell~\mathcal{R}(bexpr,\rho)~\mathbf{do}~ \mathcal{R}(S',\rho)\).
Due to the semantics, \((S_0,\sigma_0)\stackrel{\neg bexpr}{\rightarrow}(E, \sigma_0)\).
By definition of \(\rho(\sigma_0)\), \(\rho(\sigma_0)(\mathcal{R}(bexpr,\rho))=\sigma_0(bexpr)=false\).
Due to the semantics, \((\mathcal{R}(S,\rho),\rho(\sigma_0))\stackrel{\neg \mathcal{R}(bexpr,\rho)}{\rightarrow}(E, \rho(\sigma_0))\).
Since \(\mathcal{R}(E,\rho)=E\) and \(\neg\mathcal{R}(bexpr,\rho)=\mathcal{R}(\neg bexpr,\rho)\), the hypothesis follows.

\emph{Case 7 (\(S_0\equiv [E \| \dots \| E]\))}
Then, \(\mathcal{R}(S_0,\rho)=[E \| \dots \| E]\).
Due to the semantics, \((S_0,\sigma_0)\stackrel{\textbf{nop}}{\rightarrow}(S_1, \sigma_1)\) implies \(\sigma_1=\sigma_0\) and \(S_1=E\).
Due to \(\mathcal{R}(\textbf{nop},\rho)=\textbf{nop}\), also \((\mathcal{R}(S_0,\rho),\rho(\sigma_0))\stackrel{\mathcal{R}(\textbf{nop},\rho)}{\rightarrow}(E, \rho(\sigma_0))\in ex(\mathcal{R}(S_0,\rho))\).
Since \(\mathcal{R}(E,\rho)=E\), the hypothesis follows.

\emph{Case 8 (\(S_0\equiv E;S\))}
Then, \(\mathcal{R}(S_0,\rho)=E;\mathcal{R}(S,\rho)\).
Due to the semantics, \((S_0,\sigma_0)\stackrel{\textbf{nop}}{\rightarrow}(S_1, \sigma_1)\) implies \(\sigma_1=\sigma_0\) and \(S_1=S\).
Due to \(\mathcal{R}(\textbf{nop},\rho)=\textbf{nop}\), also \((\mathcal{R}(S_0,\rho),\rho(\sigma_0))\stackrel{\mathcal{R}(\textbf{nop},\rho)}{\rightarrow}(\mathcal{R}(S,\rho), \rho(\sigma_0))\in ex(\mathcal{R}(S_0,\rho))\).
The hypothesis follows.

\textbf{Step case (n$\rightarrow$n+1):}
Consider a ordered sequence of the derivation steps, which are derived from a derivation tree for \((S_0,\sigma_0)\stackrel{op_1}{\rightarrow}(S_1,\sigma_1)\) such that a step required by another step in the tree occurs earlier in the sequence.
Since \(n+1>1\), the last step in the sequence is a computational sequential composition steps or a parallel composition step.
Consider two cases.

\emph{Case 1 (\(S_0\equiv S;S' \wedge S\neq E\))}
Then, \(\mathcal{R}(S_0,\rho)=\mathcal{R}(S,\rho);\mathcal{R}(S',\rho)\).
Due to semantics, there exists \(S,\sigma_0\stackrel{op_1}(S'',\sigma_1)\) that can be derived in less than \(n+1\) steps and \(S_1=S'';S'\). 
By induction, there exists \((\mathcal{R}(S,\rho),\rho(\sigma_0))\stackrel{\mathcal{R}(op_1,\rho)}{\rightarrow}(S'',\rho),\rho(\sigma_1))\).
We conclude  \((\mathcal{R}(S;S',\rho),\rho(\sigma_0))\stackrel{\mathcal{R}(op_1,\rho)}{\rightarrow}(S'',\rho);\mathcal{R}(S',\rho),\rho(\sigma_1))\).
Since \(\mathcal{R}(S',\rho);\mathcal{R}(S'',\rho)=\mathcal{R}(S';S'',\rho)\), the induction hypothesis follows.

\emph{Case 2 (\(S_0\equiv [S'_1 \| \dots \| S'_i\|\dots \| S'_n] \wedge\exists j\in[1,n]: S'_i\neq E\))}
Then, \(\mathcal{R}(S'_0,\rho)=[\mathcal{R}(S'_1,\rho) \| \dots \| \mathcal{R}(S'_i,\rho)\|\dots \| \mathcal{R}(S_n,\rho)]\).
Due to semantics, there exists \(j\in[1,n]\) such that \(S'_j,\sigma_0\stackrel{op_1}(S''_j,\sigma_1)\), which can be derived in less than \(n+1\) steps, and \(S_1=[S'_1 \| \dots \| S''_j\|\dots \| S'_n]'\). 
By induction, there exists \((\mathcal{R}(S'_j,\rho),\rho(\sigma_0))\stackrel{\mathcal{R}(op_1,\rho)}{\rightarrow}(S''_j,\rho),\rho(\sigma_1))\).
Hence, \((\mathcal{R}([\mathcal{R}(S'_1,\rho) \| \dots \| \mathcal{R}(S'_j,\rho)\|\dots \| \mathcal{R}(S_n,\rho)],\rho),\rho(\sigma_0))\stackrel{\mathcal{R}(op_1,\rho)}{\rightarrow}([\mathcal{R}(S'_1,\rho) \| \dots \| \mathcal{R}(S''_j,\rho)\|\dots \| \mathcal{R}(S_n,\rho)],\rho(\sigma_1))\).
Finally, taking into account that \([\mathcal{R}(S'_1,\rho) \| \dots \| \mathcal{R}(S''_j,\rho)\|\dots \| \mathcal{R}(S_n,\rho)]=\mathcal{R}([S'_1 \| \dots \| S''_j\|\dots \| S'_n],\rho)\), the induction hypothesis follows.
\end{proof}

\begin{lemma}\label{lem:biRen}
Let \(S\) be a program and \(\rho\) be a renaming function.
\[\begin{array}{l}\forall (S_0,\sigma_0)\stackrel{op_1}{\rightarrow}\dots\stackrel{op_n}{\rightarrow}(S_n,\sigma_n)\in ex(S):\\\exists (\mathcal{R}(S_0,\rho),\rho(\sigma_0))\stackrel{\mathcal{R}(op_1,\rho)}{\rightarrow}\dots\stackrel{\mathcal{R}(op_n,\rho)}{\rightarrow}(\mathcal{R}(S_n,\rho),\rho(\sigma_n))\in ex(\mathcal{R}(S,\rho))\end{array}\]
\end{lemma}

\begin{proof}
Prove by induction on the length of the executions.
Show for all programs \(S\) that 
\(\forall (S_0,\sigma_0)\stackrel{op_1}{\rightarrow}\dots\stackrel{op_n}{\rightarrow}(S_n,\sigma_n)\in ex(S):\exists (\mathcal{R}(S_0,\rho),\rho(\sigma_0))\stackrel{\mathcal{R}(op_1,\rho)}{\rightarrow}\dots\stackrel{\mathcal{R}(op_n,\rho)}{\rightarrow}(\mathcal{R}(S_n,\rho),\rho(\sigma_n))\in ex(\mathcal{R}(S,\rho))\).

\textbf{Base case (n=0)}
Let \(S\) and \(\sigma\in\Sigma\) be arbitrary.
By definition, there exists \((\mathcal{R}(S,\rho),\rho(\sigma))\in ex(\mathcal{R}(S,\rho))\).
The induction hypothesis follows.

\textbf{Step case (n-1$\rightarrow$n):}
By definition of executions, \((S_0,\sigma_0)\stackrel{op_1}{\rightarrow}\dots\stackrel{op_n}{\rightarrow}(S_n,\sigma_n)\in ex(S)\) implies \((S_0,\sigma_0)\stackrel{op_1}{\rightarrow}\dots\stackrel{op_{n-1}}{\rightarrow}(S_{n-1},\sigma_{n-1})\in ex(S)\) and \((S_{n-1},\sigma_{n-1})\stackrel{op_n}{\rightarrow}(S_n,\sigma_n)\).
Furthermore, \((S_{n-1},\sigma_{n-1})\stackrel{op_n}{\rightarrow}(S_n,\sigma_n)\in ex(S_{n-1})\). 
By induction, exists \((\mathcal{R}(S_0,\rho),\rho(\sigma_0))\stackrel{\mathcal{R}(op_1,\rho)}{\rightarrow}\dots\stackrel{\mathcal{R}(op_{n-1},\rho)}{\rightarrow}(\mathcal{R}(S_{n-1},\rho),\rho(\sigma_{n-1}))\in ex(\mathcal{R}(S_0,\rho))\).
Due to Lemma~\ref{lem:binSing}, there exists \((\mathcal{R}(S_{n-1},\rho),\rho(\sigma_{n-1}))\stackrel{\mathcal{R}(op_n,\rho)}{\rightarrow}(\mathcal{R}(S_n,\rho),\rho(\sigma_n))\in ex(\mathcal{R}(S_{n-1},\rho))\).
Hence, \((\mathcal{R}(S_0,\rho),\rho(\sigma_0))\stackrel{\mathcal{R}(op_1,\rho)}{\rightarrow}(S_1,\rho),\rho(\sigma_1))\in ex(\mathcal{R}(S_0,\rho))\).
\end{proof}

\begin{corollary}\label{cor:biRen}
Let \(S\) be a program and \(\rho\) be a renaming function.
\[\begin{array}{l}\forall (\mathcal{R}(S_0,\rho),\rho(\sigma_0))\stackrel{\mathcal{R}(op_1,\rho)}{\rightarrow}\dots\stackrel{\mathcal{R}(op_n,\rho)}{\rightarrow}(\mathcal{R}(S_n,\rho),\rho(\sigma_n))\in ex(\mathcal{R}(S,\rho)):\\
\exists (S_0,\sigma_0)\stackrel{op_1}{\rightarrow}\dots\stackrel{op_n}{\rightarrow}(S_n,\sigma_n)\in ex(S):\end{array}\]
\end{corollary}

\begin{proof}
By construction, \(\mathcal{R}(S,\rho))\) is a program and \(\rho^{-1}\) is a bijective function.
Due to Lemma~\ref{lem:biRen}, \(\forall (\mathcal{R}(S_0,\rho),\rho(\sigma_0))\stackrel{\mathcal{R}(op_1,\rho)}{\rightarrow}\dots\stackrel{\mathcal{R}(op_n,\rho)}{\rightarrow}(\mathcal{R}(S_n,\rho),\rho(\sigma_n))\in ex(\mathcal{R}(S,\rho))\) \(\exists (\mathcal{R}(\mathcal{R}(S_0,\rho), \rho^{-1}),\rho^{-1}(\rho(\sigma_0)))\stackrel{\mathcal{R}(\mathcal{R}(op_n,\rho),\rho^{-1})}{\rightarrow}\dots\stackrel{\mathcal{R}((\mathcal{R}(op_n,\rho),\rho^{-1})}{\rightarrow}(\mathcal{R}((\mathcal{R}(S_n,\rho),\rho^{-1}),\rho^{-1}(\rho(\sigma_n)))\in ex(\mathcal{R}(S,\rho))\).
Since \(\rho^{-1}\circ\rho=id\), there exists \((S_0,\sigma_0)\stackrel{op_1}{\rightarrow}\dots\stackrel{op_n}{\rightarrow}(S_n,\sigma_n)\in ex(S)\).
\end{proof}

\begin{lemma}\label{lem:noInf}
Let \(S_1\) and \(S_2\) be two programs and \(\rho\) a renaming function that is appropriate for renaming.
\(\forall\sigma\in\Sigma: (\mathcal{R}(S_1,\rho),\sigma)\rightarrow^*(E,\sigma_1)\in ex(\mathcal{R}(S_1,\rho))\wedge (S_2,\sigma)\rightarrow^*(E,\sigma_2)\in ex(S_2)\Rightarrow \exists (\mathcal{R}(S_1,\rho);S_2,\sigma)\rightarrow^*(E,\sigma')\in ex(\mathcal{R}(S_1,\rho);S_2): \sigma'=_{|_{\mathcal{V}(\mathcal{R}(S_1,\rho))\cup\bigcup_{v\in\mathcal{M}(S_2)}\rho(v)}}\sigma_1 \wedge \sigma'=_{|_{\mathcal{V}(S_2)\cup\mathcal{M}(S_1)}}\sigma_2\)
\end{lemma}

\begin{proof}
First, show \(\sigma=_{|_{\mathcal{V}(S_2)\cup\mathcal{M}(S_1)}}\sigma_1\).
Due to Lemma~\ref{lem:change}, \(\sigma=_{|_{\mathcal{V}\setminus\mathcal{M}(\mathcal{R}(S_1,\rho))}}\sigma_1\).
Due to Lemma~\ref{lem:biRen}, Corollary~\ref{cor:biRen}, and the definition of modified variables, \(\mathcal{M}(\mathcal{R}(S_1,\rho))=\bigcup_{v\in\mathcal{M}(S_1)} \rho(v)\).
Since \(\rho\) is appropriate for renaming, we conclude \((\mathcal{V}(S_2)\cup\mathcal{M}(S_1))\cap\mathcal{M}(\mathcal{R}(S_1,\rho))=\emptyset\).
Hence, \(\sigma=_{|_{\mathcal{V}(S_2)\cup\mathcal{M}(S_1)}}\sigma_1\).

Due to Lemma~\ref{lem:ubBehave}, \(\exists(S_2,\sigma_1)\rightarrow^*(E,\sigma')\in ex(S_2)\) with \(\sigma'=_{|_{\mathcal{V}(S_2)\cup\mathcal{M}(S_1)}}\sigma_2\).
Due to semantics, \((\mathcal{R}(S_1,\rho);S_2,\sigma)\rightarrow^*(S_2,\sigma_1)\rightarrow^*(E,\sigma')\in ex(\mathcal{R}(S_1,\rho);S_2)\).
From Lemma~\ref{lem:change}, we conclude \(\sigma_1=_{|_{\mathcal{V}\setminus\mathcal{M}(S_2)}}\sigma'\).
Since \(\rho\) is appropriate for renaming, we conclude \((\mathcal{V}(\mathcal{R}(S_1,\rho))\cup\bigcup_{v\in\mathcal{M}(S_2)}\rho(v))\cap\mathcal{M}(S_2)=\emptyset\).
The claim follows.
\end{proof}

\begin{lemma}\label{lem:stepsNotRep}
Let \(S\) and \(S'\) be two programs, \(\gamma\) be a replacement function such that \(S'=\Gamma(S,\gamma)\), and  \(V\subseteq\mathcal{V}\) be a set of outputs.
If \(\neg\exists S^s_1, S^s_2: S^s_1\in dom(\gamma)\wedge (S=S^s_1;S^s_2 \vee S=S^s_1)\), then \(\forall \sigma, \sigma'\in\Sigma: (S,\sigma)\stackrel{op_1}{\rightarrow}(S_1,\sigma_1)\in ex(S)\wedge \sigma=_{|_{\mathcal{L}(S, V)\cup\mathcal{L}(S', V)}}\sigma' \implies \exists (S',\sigma')\stackrel{op_1}{\rightarrow}(S'_1,\sigma'_1): S'_1=\Gamma(S_1,\gamma) \wedge  \sigma_1=_{|_{\mathcal{L}(S_1, V)\cup\mathcal{L}(S'_1,V)}}\sigma'_1\).
\end{lemma}

\begin{proof}
Let \(S\) and \(S'\) be two programs, \(\gamma\) be a replacement function such that \(S'=\Gamma(S,\gamma)\) and \(\neg\exists S^s_1, S^s_2: S^s_1\in dom(\gamma)\wedge  (S=S^s_1;S^s_2 \vee S=S^s_1)\),  \(V\subseteq\mathcal{V}\) be a set of outputs and \(\sigma, \sigma'\in\Sigma\) with \(\sigma=_{|_{\mathcal{L}(S, V)\cup\mathcal{L}(S',V)}}\sigma'\).

Assume \(p=(S,\sigma)\stackrel{op_1}{\rightarrow}(S_1,\sigma_1)\in ex(S)\).
By definition, \(\mathcal{UB}(S)\cup\mathcal{UB}(S')\subseteq \mathcal{L}(S, V)\cup\mathcal{L}(S',V)\).
Consider two cases.

Case 1 (\(S_1=E\)):
Due to the semantics, we conclude that either \(S=E;E\), or \(S\) is  not a sequential composition, but a statement.
If \(S\) is an assignment, an assertion or a parallel statement, we conclude from  \(S\notin dom(\gamma)\), \(S'=\Gamma(S,\gamma)\), and \(\gamma\) does not replace statements in parallel statements that \(S'=S\). 
Similarly, if \(S=E;E\), also \(S=S'\). 
Due to Lemma~\ref{lem:ubBehave}, \(\exists (S',\sigma')\stackrel{op_1}{\rightarrow}(E,\sigma'_1)\) and \(\sigma_1=_{|_{\{v\in\mathcal{V}\mid \sigma(v)=\sigma'(v)\vee \exists 1\leq i\leq n: op_i\equiv v:=expr;\}}} \sigma'_1\).
Due to the definition of live variables, we conclude that \(\sigma_1=_{|_{\mathcal{L}(E, V)\cup\mathcal{L}(E, V)}}\sigma'_1\).
If \(S\) is an if- or while-statement, we conclude from  \(S\notin dom(\gamma)\) and \(S'=\Gamma(S,\gamma)\) that \(S'\) is an if-/while-statement and the condition is the same.
Due to semantics, definition of live variable analysis, and \(\sigma=_{|_{\mathcal{L}(S, V)\cup\mathcal{L}(S', V)}}\sigma'\), we then conclude that \((S',\sigma')\stackrel{op_1}{\rightarrow}(E,\sigma')\in ex(S')\) and \(\sigma_1=\sigma\).
Due to the definition of live variables, we conclude that \(\sigma_1=_{|_{\mathcal{L}(E,  V)\cup\mathcal{L}(E, V)}}\sigma'\).
By definition, \(E=\Gamma(E,\gamma)\).

Case 2 (\(S_1\neq E\)):
Since replacements do not occur in parallel statements and \(\neg\exists S^s_1, S^s_2: S^s_1\in dom(\gamma)\wedge  (S=S^s_1;S^s_2 \vee S=S^s_1)\), we conclude that \(\exists S^s_1, S^s_2, S^s_3, S^s_4:  S=S^s_1;S^s_2 \wedge S'=S^s_3;S^s_4\wedge S^s_3=\Gamma(S^s_1, \gamma) \wedge S^s_4=\Gamma(S^s_2, \gamma)\vee S=S^s_1\wedge S'=S^s_3\wedge S^s_3=\Gamma(S^s_1, \gamma)\) and either \(S^s_1=S^s_3\) or \(S^s_1\) and \(S^s_3\) are either both if- or both-while statements with the same condition and the if/else-body, the loop body of $S^s_3$ is a replacement of the body of $S^s_1$. 
First, consider the first case (\(S^s_1=S^s_3\)).
Due to semantics, either (1)~\(S_1=E;S^s_2\) and \((S^s_1, \sigma)\stackrel{op_1}{\rightarrow}(E,\sigma_1)\), (2)~\(S_1=S^s_2\wedge S^s_1=E\), \(\sigma=\sigma_1\), and \((S, \sigma)\stackrel{\textbf{nop}}{\rightarrow}(S_1,\sigma_1)\), or (3)~\(S_1=S^s_5;S^s_2\) and \((S^s_1, \sigma)\stackrel{op_1}{\rightarrow}(S^s_5,\sigma_1)\).
Due to Lemma~\ref{lem:ubBehave}, in case~(1) \(\exists (S^s_1,\sigma')\stackrel{op_1}{\rightarrow}(E,\sigma'_1)\) and \(\sigma'=\sigma'_1\), and in case~(3)  \(\exists (S^s_1,\sigma')\stackrel{op_1}{\rightarrow}(S^s_5,\sigma'_1)\).
Furthermore, \(\sigma_1=_{|_{\{v\in\mathcal{V}\mid \sigma(v)=\sigma'(v)\vee op_1\equiv v:=expr;\}}} \sigma'_1\).
Due to semantics, in case~(1) \(\exists (S',\sigma')\stackrel{op_1}{\rightarrow}(E;S^s_4,\sigma'_1)\), in case~(2) \((S', \sigma')\stackrel{\textbf{nop}}{\rightarrow}(S^s_4,\sigma')\), and in case~(3)
  \(\exists (S',\sigma')\stackrel{op_1}{\rightarrow}(S^s_5;S^s_4,\sigma'_1)\).
Since \(S^s_1=S^s_3=\Gamma(S^s_3,\gamma)\), \(\gamma\) is only defined for subprograms of \(S\) and statements (thus, subprograms) can be uniquely identified via labels, we get \(\Gamma(S^s_5,\gamma)=S^s_5\).
Hence, \(\Gamma(S^s_5;S^s_2,\gamma)=S^s_5;\Gamma(S^s_2,\gamma)=S^s_5;S^s_4\).
Similarly, \(\Gamma(E;S^s_2,\gamma)=E;\Gamma(S^s_2,\gamma)=S^s_5;S^s_4\).
Moreover, \(\sigma_1=_{|_{\{v\in\mathcal{V}\mid \sigma(v)=\sigma'(v)\vee op_1\equiv v:=expr;\}}} \sigma'_1\)., \(\sigma=_{|_{\mathcal{L}(S, V)\cup\mathcal{L}(S',V)}}\sigma'\), and the definition of live variable analyses let us conclude that \(\sigma_1=_{|_{\mathcal{L}(S_1, V)\cup\mathcal{L}(S'_1, V)}}\sigma'_1\). % with \(S'_1=S^s_4\) and \(S'_1=S^s_6;S^s_4\), respectively.

Second, consider that (\(S^s_1\neq S^s_3\)).
We know that \(S^s_1\) and \(S^s_3\) are either both if- or both-while statements with the same condition and the if/else-body, the loop body of $S^s_3$ is a replacement of the body of $S^s_1$.
Due to semantics, definition of live variable analysis, \(\sigma=_{|_{\mathcal{L}(S, V)\cup\mathcal{L}(S',V)}}\sigma'\), and the replacement function, we then conclude that \(\sigma=\sigma_1\) and \(\sigma'=\sigma'_1\) and either \(S=S^s_1\) and exists \((S',\sigma')\stackrel{op_1}{\rightarrow}(S'_1,\sigma_1)\in ex(S')\) with \(S'_1=\Gamma(S_1,\gamma)\) (due to \(S; \textbf{while}~expr~\textbf{do}~S\) is no subprogram of S) or \(S=S^s_1;S^s_2\) and \((S',\sigma')\stackrel{op_1}{\rightarrow}(S'_1,\sigma_1)\in ex(S')\) with \(S'_1=\Gamma(S_1,\gamma)\).
Due to the definition of live variables, we conclude that \(\sigma_1=_{|_{\mathcal{L}(S_1, V)\cup\mathcal{L}(S'_1, V)}}\sigma'_1\).
\end{proof}

\begin{corollary}\label{cor:stepsNotRep}
Let \(S\) and \(S'\) be two programs, \(\gamma\) be a replacement function such that \(S'=\Gamma(S,\gamma)\), and  \(V\subseteq\mathcal{V}\) be a set of outputs.
For all \((S_0,\sigma_0)\stackrel{op_1}{\rightarrow}\dots\stackrel{op_n}{\rightarrow}(S_n,\sigma_n)\in ex(S)\), if for all \(0\leq i<n\)  not exists \(S^{s1}_i, S^{s2}_i \) such that \((S_i=S^{s1}_i;S^{s2}_i \vee S_i=S^{s1}_i)\) and  \(S^{s1}_i \in dom(\gamma)\), then \(\forall \sigma'_0\in\Sigma: \sigma_0=_{|_{\mathcal{L}(S, V)\cup\mathcal{L}(S', V)}}\sigma'_0 \implies \exists (\Gamma(S_0,\gamma),\sigma'_0)\stackrel{op_1}{\rightarrow}\dots\stackrel{op_n}{\rightarrow}(\Gamma(S_n,\gamma),\sigma'_n)\in ex(S'): \forall 0\leq i\leq n: \sigma_i=_{|_{\mathcal{L}(S_i, V)\cup\mathcal{L}(\Gamma(S_i,\gamma),V)}}\sigma'_i\).
\end{corollary}

\begin{proof}
Proof by induction.

\textbf{Base case} (i=0):
By definition \((S',\sigma)=(\Gamma(S,\gamma),\sigma)\in ex(S')\) for arbitrary \(\sigma\in\Sigma\) (including all \(\sigma'_0\) with \(\sigma_0=_{|_{\mathcal{L}(S,V)\cup\mathcal{L}(S',V)}}\sigma'_0\)).

\textbf{Step case} (\(n-1\rightarrow n\)):
Due to Lemma~\ref{lem:stepsNotRep}, there exists \((\Gamma(S,\gamma),\sigma')\stackrel{op_1}{\rightarrow}(\Gamma(S_1,\gamma),\sigma'_1)\) with \(\sigma_1=_{|_{\mathcal{L}(S_1, V)\cup\mathcal{L}(S'_1,V)}}\sigma'_1\).
By induction,  \((\Gamma(S_1,\gamma),\sigma'_1)\stackrel{op_2}{\rightarrow}\dots\stackrel{op_n}{\rightarrow}(\Gamma(S_n,\gamma),\sigma'_n)\in ex(S') \wedge \forall 1\leq i\leq n: \sigma_i=_{|_{\mathcal{L}(S_i, V)\cup\mathcal{L}(\Gamma(S_i,\gamma),V)}}\sigma'_i\).
By definition, the induction hypothesis follows.
\end{proof}

\subsection{Auxiliary Lemmas for Soundness of Initialization and Equalization Part}

\begin{lemma}\label{lem:init}
Let \(\rho\) be a renaming function and \(V\subseteq\mathcal{V}\) a subset of variables such that \(\forall v\in V: \rho(v)=v\vee \rho(v)\notin V\).
Then, \(\forall \sigma\in\Sigma: (init(\rho,\mathrm{toSeq(V)}), \sigma)\rightarrow^* (E,\sigma')\implies \forall v\in V:\sigma'(v)=\sigma'(\rho(v))=\sigma(\rho(v))\).
\end{lemma}

\begin{proof}
Proof by induction on the cardinality of V.

\textbf{Base case (\(|V|=0\))}
\(|V|=0\) implies \(V=\emptyset\), the hypothesis trivially holds.

\textbf{Step case (\(|V|=n, n>0\))}
Let \(\mathrm{toSeq}(V)=v_1,\dots, v_n\) and \(\sigma\in\Sigma\) be arbitrary.
Then, \(init(\rho,\mathrm{toSeq}(V)=init(\rho,v_1,\dots, v_n))=v:=\rho(v);init(\rho,v_2,\dots, v_n)=v:=\rho(v);init(\rho,\mathrm{toSeq}(V\setminus\{v_1\}))\).
Due to semantics, \((init(\rho,\mathrm{toSeq(V)}), \sigma)\rightarrow^* (E,\sigma')\) implies that \((init(\rho,\mathrm{toSeq(V)}), \sigma)\stackrel{v_1:=\rho(v_1);}{\rightarrow}(init(\rho,\mathrm{toSeq}(V\setminus\{v_1\})),\) \(\sigma[v_1:=\sigma(\rho(v_1))])\rightarrow^* (E,\sigma')\).
Thus, we get \(\sigma[v_1:=\sigma(\rho(v_1))](v_1)=\sigma(\rho(v_1))=\) \mbox{\(\sigma[v_1:=\sigma(\rho(v_1))](\rho(v_1))\)}.
By induction, \(\forall v\in V\setminus\{v_1\}:\sigma'(v)=\sigma'(\rho(v))=\sigma(\rho(v))\).
Since \(\forall v\in V: \rho(v)=v\vee \rho(v)\notin V\) and  \(\mathcal{M}(\rho,\mathrm{toSeq}(V\setminus\{v_1\}))\subseteq\{v_2,\dots, v_n\}\), we get \(\mathcal{M}(\rho,\mathrm{toSeq}(V\setminus\{v_1\}))\cap\{v_1, \rho(v_1)\}=\emptyset\).
Due to Lemma~\ref{lem:change}, \(\sigma[v:=\sigma(\rho(v))](v_1)=\sigma'(v_1)\) and \(\sigma[v:=\sigma(\rho(v))](\rho(v_1))=\sigma'(\rho(v_1))\).
Hence, \(\sigma'(v_1)=\sigma'(\rho(v))=\sigma(\rho(v))\).
The induction hypothesis follows.
\end{proof}

\begin{lemma}\label{lem:equiv}
Let \(\rho\) be a renaming function and \(V\subseteq\mathcal{V}\) a subset of variables.
Then, \(\forall \sigma\in\Sigma: (equal(\rho,\mathrm{toSeq(V)}, \sigma)\rightarrow^* (E,\sigma')\implies \forall v\in V:\sigma(v)=\sigma(\rho(v))\).
\end{lemma}

\begin{proof}
Proof by induction on the cardinality of V.

\textbf{Base case (\(|V|=0\))}
\(|V|=0\) implies \(V=\emptyset\), the hypothesis trivially holds.

\textbf{Step case (\(|V|=n, n>0\))}
Let \(\mathrm{toSeq}(V)=v_1,\dots, v_n\) and \(\sigma\in\Sigma\) be arbitrary.
From definition,  we conclude that \(equal(\rho,\mathrm{toSeq}(V)=equal(\rho,v_1,\dots, v_n)=\textbf{assert}~\rho(v_1)==v_1;equal(\rho,v_2,\dots, v_n)=\textbf{assert}~\rho(v_1)==v_1;equal(\rho,\mathrm{toSeq}(V\setminus\{v_1\}))\).
Due to semantics and \((equal(\rho,\mathrm{toSeq(V)}, \sigma)\rightarrow^* (E,\sigma')\), we infer that \((equal(\rho,\mathrm{toSeq(V)}, \sigma)\stackrel{\rho(v_1)==v_1}{\rightarrow}(equal(\rho,\mathrm{toSeq}(V\setminus\{v_1\})),\sigma)\rightarrow^* (E,\sigma')\) and \(\sigma(\rho(v_1)==v_1)=true\).
Thus, \(\sigma(\rho(v_1))=\sigma(v_1)\).
By induction, \(\forall v\in V\setminus\{v_1\}:\sigma(v)=\sigma(\rho(v))\).
The induction hypothesis follows.
\end{proof}

\subsection{Proof of Theorem~\ref{theo:EquivTotal}}\label{ssec:proof:theo:EquivTotal}

\begin{theorem1}%[\ref{theo:EquivTotal}]
Let \(S_1\) and \(S_2\) be two (sub)programs. 
Given overapproximation 
 \(\mathcal{UB}(S_1)\subseteq U_1\subseteq \mathcal{V}(S_1)\) and \(\mathcal{UB}(S_2)\subseteq U_2\subseteq \mathcal{V}(S_2)\) of the variables used before definition and  overapproximations \(\mathcal{M}(S_1)\subseteq M_1\subseteq \mathcal{V}(S_1)\) and \(\mathcal{M}(S_2)\subseteq M_2\subseteq \mathcal{V}(S_2)\) of the modified variables, a renaming function \(\rho_\mathrm{switch}\), and \(C\subseteq M_1\cup M_2\).

If all \((eq\_task(S_1, S_2, \rho_\mathrm{switch},(U_1\cap U_2)\cap(M_1\cup M_2) ,C),\sigma)\rightarrow^*(S',\sigma')\in ex(eq\_task(S_1, S_2, \rho_\mathrm{switch}, (U_1\cap U_2)\cap(M_1\cup M_2) ,C))\) do not violate an assertion, then \(S_1\equiv_{\mathcal{V}\setminus((\mathcal{M}(S_1)\cup\mathcal{M}(S_2))\setminus C)} S_2\).

\end{theorem1}

\begin{proof}
Let \(\sigma\in\Sigma\), \((S_1,\sigma)\rightarrow^*(E,\sigma')\in ex(S_1)\), \((S_2,\sigma)\rightarrow^*(E,\sigma'')\in ex(S_2)\), and \(u\in \mathcal{V}\setminus((\mathcal{M}(S_1)\cup\mathcal{M}(S_2))\setminus C)\) be arbitrary.

Consider two cases.
First, consider \(u\in \mathcal{V}\setminus(\mathcal{M}(S_1)\cup\mathcal{M}(S_2))\).
Due to Lemma~\ref{lem:change}, \(\sigma(u)=\sigma'(u)\) and \(\sigma(u)=\sigma''(u)\).
Hence, \(\sigma'(u)=\sigma''(u)\).

Second, consider \(u\in (\mathcal{M}(S_1)\cup\mathcal{M}(S_2))\cap C\).
Let \(\sigma_r\in\Sigma\) with \(\sigma_r=_{|_{\mathcal{V}(S_2)\cup\mathcal{M}(S_1)}}\sigma\) and \(\sigma_r=_{|_{\mathcal{V}(\mathcal{R}(S_1,\rho_\mathrm{switch}))\cup\bigcup_{v\in\mathcal{M}(S_2)}\rho_\mathrm{switch}(v)}}\rho_\mathrm{switch}(\sigma)\).
Due to definition of \(\rho_\mathrm{switch}\) such a data state exists. 

Due to semantics and Lemma~\ref{lem:init}, \(\exists(init(\rho_\mathrm{switch},(U_1\cap U_2)\cap(M_1\cup M_2)), \sigma_r)\rightarrow^* (E,\sigma_\mathrm{init})\) with \(\forall v\in (U_1\cap U_2)\cap(M_1\cup M_2):\sigma_\mathrm{init}(v)=\sigma_\mathrm{init}(\rho_\mathrm{switch}(v))=\sigma_r(\rho_\mathrm{switch}(v))\).
By construction of \(\sigma_r\) and \(\rho_\mathrm{switch}\), we further get \(\forall v\in (U_1\cap U_2)\cap(M_1\cup M_2): \sigma_r(\rho_\mathrm{switch}(v))=\rho_\mathrm{switch}(\sigma)(\rho_\mathrm{switch}(v))=\sigma(\rho_\mathrm{switch}^{-1}(\rho_\mathrm{switch}(v)))=\sigma(v)\).
Due to Lemma~\ref{lem:change} and \(\mathcal{M}((init(\rho_\mathrm{switch},(U_1\cap U_2)\cap(M_1\cup M_2)))\subseteq (U_1\cap U_2)\cap(M_1\cup M_2)\), we infer that \(\forall v\in \mathcal{V}\setminus((U_1\cap U_2)\cap(M_1\cup M_2)): \sigma_\mathrm{init}(v)=\sigma_r(v)\).
Hence, \(\sigma_r=\sigma_\mathrm{init}\).

Due to Lemma~\ref{lem:biRen}, there exists \((\mathcal{R}(S_1,\rho_\mathrm{switch}),\rho_\mathrm{switch}(\sigma))\rightarrow^*(E,\rho_\mathrm{switch}(\sigma'))\in ex(\mathcal{R}(S_1,\rho_\mathrm{switch}))\).
By definition of \(\sigma_r\),  and Lemma~\ref{lem:ubBehave}, there exists \((S_2,\sigma_r)\rightarrow^*(E,\sigma''_r)\in ex(S_2)\) with \(\sigma''=_{|_{\mathcal{V}(S_2)\cup\mathcal{M}(S_1)}} \sigma''_r\) and \(((\mathcal{R}(S_1,\rho_\mathrm{switch}),\sigma_r)\rightarrow^*(E,\sigma'_r)\in ex((\mathcal{R}(S_1,\rho_\mathrm{switch}))\) with \(\rho_\mathrm{switch}(\sigma')=_{|_{\mathcal{V}(\mathcal{R}(S_1,\rho_\mathrm{switch}))\cup\bigcup_{v\in\mathcal{M}(S_2)}\rho_\mathrm{switch}(v)}} \sigma'_r\).
Due to Lemma~\ref{lem:noInf}, there exists \((\mathcal{R}(S_1,\rho_\mathrm{switch});S_2,\sigma_r)\rightarrow^*(E,\sigma_c)\in ex(\mathcal{R}(S_1,\rho_\mathrm{switch});S_2)\) with  \(\sigma_c=_{|_{\mathcal{V}(\mathcal{R}(S_1,\rho_\mathrm{switch}))\cup\bigcup_{v\in\mathcal{M}(S_2)}\rho_\mathrm{switch}(v)}}\sigma'_r\) and \(\sigma_c=_{|_{\mathcal{V}(S_2)\cup\mathcal{M}(S_1)}}\sigma''_r\).

Due to semantics ,\(\sigma_r=\sigma_\mathrm{init}\), and all \((eq\_task(S_1, S_2, \rho_\mathrm{switch}, (U_1\cap U_2)\cap(M_1\cup M_2) ,C),\sigma)\rightarrow^*(S',\sigma')\in ex(eq\_task(S_1, S_2, \rho_\mathrm{switch}, (U_1\cap U_2)\cap(M_1\cup M_2) ,C))\) do not violate assertions, there exists \((equal(\rho_\mathrm{switch}, C), \sigma_c)\rightarrow^*(E,\sigma'_c)\).
Due to Lemma~\ref{lem:equiv}, we infer for all \(v\in C\) that \(\sigma_c(v)=\sigma_c(\rho_\mathrm{switch}(v))\).
Since \(u\in (\mathcal{M}(S_1)\cup\mathcal{M}(S_2))\cap C\), we conclude that \(\sigma_c(\rho_\mathrm{switch}(u))=\sigma_c(u)=\sigma''_r(u)=\sigma''(u)\) and \(\sigma_c(\rho_\mathrm{switch}(u))=\sigma'_r(\rho_\mathrm{switch}(u))=\rho_\mathrm{switch}(\sigma')(\rho_\mathrm{switch}(u))=\sigma'(\rho_\mathrm{switch}^{-1}(\rho_\mathrm{switch}(u)))=\sigma'(u)\).
\end{proof}

\subsection{Proof of Theorem~\ref{theo:EquivNoApprox}}\label{ssec:proof:theo:EquivNoApprox}

\begin{lemma}\label{lem:equivCheck}
Let \(S_1\) and \(S_2\) be two (sub)programs, \(\mathcal{M}(S_1)\subseteq M_1\subseteq \mathcal{V}(S_1)\) and \(\mathcal{M}(S_2)\subseteq M_2\subseteq \mathcal{V}(S_2)\) overapproximations of the modified variables, \(\rho_\mathrm{switch}\) a renaming function , \(I=(\mathcal{UB}(S_1)\cap \mathcal{UB}(S_2))\cap(M_1\cup M_2)\), and \(C\subseteq M_1\cup M_2\).
If all \((eq\_task(S_1, S_2, \rho_\mathrm{switch},I ,C),\sigma)\rightarrow^*(S',\sigma')\in ex(eq\_task(S_1, S_2, \rho_\mathrm{switch}, I ,C))\) do not violate an assertion, then \(\forall \sigma_1, \sigma_2\in\Sigma, V\subseteq\mathcal{V}, v\in V\setminus((\mathcal{M}(S_1)\cup\mathcal{M}(S_2))\setminus C): p_1=(S_1,\sigma_1)\rightarrow^*(E,\sigma'_1)\in ex(S_1)\wedge p_2=(S_2,\sigma_2)\rightarrow^*(E,\sigma_2')\in ex(S_2)\wedge \mathcal{UB}(p_1)\cup\mathcal{UB}(p_2)\subseteq V\wedge \sigma_1=_{|_{V}}\sigma_2\implies \sigma'_1(v)=\sigma'_2(v)\).
\end{lemma}

\begin{proof}
Let  \(\mathcal{UB}(S_1)\cup\mathcal{UB}(S_2)\subseteq V\subseteq\mathcal{V}\) and \(u\in V\setminus((\mathcal{M}(S_1)\cup\mathcal{M}(S_2))\setminus C)\) be arbitrary.
Furthermore, consider arbitrary \(p_1=(S_1,\sigma_1)\rightarrow^*(E,\sigma'_1)\in ex(S_1)\) and \(p_2=(S_2,\sigma_2)\rightarrow^*(E,\sigma_2')\in ex(S_2)\) with \(\sigma_1=_{|_{V}}\sigma_2\).

Consider two cases.
First, consider \(u\in V\setminus(\mathcal{M}(S_1)\cup\mathcal{M}(S_2))\).
Due to Lemma~\ref{lem:change}, \(\sigma_1(u)=\sigma'_1(u)\) and \(\sigma_2(u)=\sigma'_2(u)\).
Hence, \(\sigma'_1(u)=\sigma'_2(u)\).

Second, consider \(u\in (\mathcal{M}(S_1)\cup\mathcal{M}(S_2))\cap C\).
Let us consider \(\sigma_r\in\Sigma\) with \(\sigma_r=_{|_{(\mathcal{UB}(S_1)\cup\mathcal{UB}(S_2)\cup\mathcal{M}(S_1)\cup\mathcal{M}(S_2))}}\sigma_2\) and \(\sigma_r=_{|_{\bigcup_{v\in(\mathcal{UB}(S_1)\cup\mathcal{UB}(S_2)\cup\mathcal{M}(S_1)\cup\mathcal{M}(S_2))}\rho_\mathrm{switch}(v)}}\rho_\mathrm{switch}(\sigma_1)\).
Due to definition of \(\rho_\mathrm{switch}\) such a data state exists. 

Due to semantics and Lemma~\ref{lem:init}, there exists \((init(\rho_\mathrm{switch},(\mathcal{UB}(S_1)\cap\mathcal{UB}(S_2))\cap(M_1\cup M_2)), \sigma_r)\rightarrow^* (E,\sigma_\mathrm{init})\) with \(\forall v\in (\mathcal{UB}(S_1)\cap\mathcal{UB}(S_2))\cap(M_1\cup M_2):\sigma_\mathrm{init}(v)=\sigma_\mathrm{init}(\rho_\mathrm{switch}(v))=\sigma_r(\rho_\mathrm{switch}(v))\).
By construction of \(\sigma_r\) and \(\rho_\mathrm{switch}\), \(\mathcal{UB}(S_1)\cap\mathcal{UB}(S_2)\subseteq V\subseteq\mathcal{V}\), and \(\sigma_1=_{|_{V}}\sigma_2\), we get \(\forall v\in (\mathcal{UB}(S_1)\cap\mathcal{UB}(S_2))\cap(M_1\cup M_2):\sigma_r(\rho_\mathrm{switch}(v))=\rho_\mathrm{switch}(\sigma_1)(\rho_\mathrm{switch}(v))=\sigma_1(\rho_\mathrm{switch}^{-1}(\rho_\mathrm{switch}(v)))=\sigma_1(v)=\sigma_2(v)\).
Due to Lemma~\ref{lem:change} and \(\mathcal{M}((init(\rho_\mathrm{switch},(\mathcal{UB}(S_1)\cap\mathcal{UB}(S_2))\cap(M_1\cup M_2)))\subseteq (\mathcal{UB}(S_1)\cap\mathcal{UB}(S_2))\cap(M_1\cup M_2)\), we infer that \(\forall v\in \mathcal{V}\setminus((\mathcal{UB}(S_1)\cap\mathcal{UB}(S_2))\cap(M_1\cup M_2)): \sigma_\mathrm{init}(v)=\sigma_r(v)\).
Hence, \(\sigma_r=\sigma_\mathrm{init}\).

Due to Lemma~\ref{lem:biRen}, there exists \((\mathcal{R}(S_1,\rho_\mathrm{switch}),\rho_\mathrm{switch}(\sigma_1))\rightarrow^*(E,\rho_\mathrm{switch}(\sigma'_1))\in ex(\mathcal{R}(S_1,\rho_\mathrm{switch}))\).
By definition of \(\sigma_r\) and Lemma~\ref{lem:ubBehave}, there exists execution \((S_2,\sigma_r)\rightarrow^*(E,\sigma''_r)\in ex(S_2)\) with \(\sigma'_2=_{|_{(\mathcal{UB}(S_1)\cup\mathcal{UB}(S_2))\cup\mathcal{M}(S_1)\cup\mathcal{M}(S_2)}} \sigma''_r\) as well as execution \(((\mathcal{R}(S_1,\rho_\mathrm{switch}),\sigma_r)\rightarrow^*(E,\sigma'_r)\in ex((\mathcal{R}(S_1,\rho_\mathrm{switch}))\) such that \(\rho_\mathrm{switch}(\sigma'_1)=_{|_{\bigcup_{v\in(\mathcal{UB}(S_1)\cup\mathcal{UB}(S_2)\cup\mathcal{M}(S_1)\cup\mathcal{M}(S_2))}\rho_\mathrm{switch}(v)}} \sigma'_r\).
Due to Lemma~\ref{lem:noInf}, there exists execution \((\mathcal{R}(S_1,\rho_\mathrm{switch});S_2,\sigma_r)\rightarrow^*(E,\sigma_c)\in ex(\mathcal{R}(S_1,\rho_\mathrm{switch});S_2)\) with  \(\sigma_c=_{|_{\mathcal{V}(\mathcal{R}(S_1,\rho_\mathrm{switch}))\cup\bigcup_{v\in\mathcal{M}(S_2)}\rho_\mathrm{switch}(v)}}\sigma'_r\) and \(\sigma_c=_{|_{\mathcal{V}(S_2)\cup\mathcal{M}(S_1)}}\sigma''_r\).

Due to semantics, \(\sigma_r=\sigma_\mathrm{init}\), and all \((eq\_task(S_1, S_2, \rho_\mathrm{switch}, \mathcal{UB}(S_1)\cap\mathcal{UB}(S_2))\cap(M_1\cup M_2) ,C),\sigma)\rightarrow^*(S',\sigma')\in ex(eq\_task(S_1, S_2, \rho_\mathrm{switch}, (\mathcal{UB}(S_1)\cap\mathcal{UB}(S_2))\cap(M_1\cup M_2) ,C))\) do not violate assertions, there exists \((equal(\rho_\mathrm{switch}, C), \sigma_c)\rightarrow^*(E,\sigma'_c)\).
Due to Lemma~\ref{lem:equiv}, we infer for all \(v\in C\) that \(\sigma_c(v)=\sigma_c(\rho_\mathrm{switch}(v))\).
Since \(u\in (\mathcal{M}(S_1)\cup\mathcal{M}(S_2))\cap C\), we conclude that \(\sigma_c(\rho_\mathrm{switch}(u))=\sigma_c(u)=\sigma''_r(u)=\sigma'_2(u)\) and \(\sigma_c(\rho_\mathrm{switch}(u))=\sigma'_r(\rho_\mathrm{switch}(u))=\rho_\mathrm{switch}(\sigma'_1)(\rho_\mathrm{switch}(u))=\sigma'_1(\rho_\mathrm{switch}^{-1}(\rho_\mathrm{switch}(u)))=\sigma'_1(u)\).

\end{proof}

\begin{theorem2}%[\ref{theo:EquivNoApprox}]
Let \(S\) and \(S'\) be two programs, \(\gamma\) be a replacement function such that \(S'=\Gamma(S,\gamma)\), and  \(V\subseteq\mathcal{V}\) be a set of outputs.
If for all \((S_1, S_2)\in\gamma\) there exists \(\mathcal{M}(S_1)\subseteq M_1\subseteq \mathcal{V}(S_1)\), \(\mathcal{M}(S_2)\subseteq M_2\subseteq \mathcal{V}(S_2)\), \(\mathcal{L}(S_1, S, V)\subseteq L_1\subseteq \mathcal{V}\), \(\mathcal{L}(S_2, S', V)\subseteq L_2\subseteq \mathcal{V}\), and renaming function \(\rho_\mathrm{switch}\) such that the equivalence task \(eq\_task(S_1, S_2, \rho_\mathrm{switch}, (\mathcal{UB}(S_1)\cap\mathcal{UB}(S_2))\cap(M_1\cup M_2), (M_1\cup M_2)\cap (L_1\cup L_2))\) does not violate an assertion, then \(S\equiv_V S'\).
\end{theorem2}

\begin{proof}
Consider \((S_p,\sigma)\rightarrow^*(E,\sigma')=(S_0,\sigma_0)\stackrel{op_1}{\rightarrow}\dots\stackrel{op_n}{\rightarrow}(S_n,\sigma_n)\) be a path for an arbitrary program \(S_p\).
We define the splitting of the path into \(m\geq0\) segments such that each segment represents either a sequence in which each program of the sequence's states except for the last one does not start with a replaced subprogram or the execution of the subprogram that will be replaced.
In case that there exists multiple replacements (nesting of replaced subprograms), we use the largest replacement.
Show by induction over the number of segments that for all programs \(S_p\) such that \(\exists\sigma,\sigma'\in\Sigma: (S,\sigma)\rightarrow^*(S_p,\sigma')\) if \((S_p,\sigma)\rightarrow^*(E,\sigma')\in ex(S_p)\), \(S'_p=\Gamma(S_p,\gamma)\), \(\exists\sigma,\sigma'\in\Sigma: (S',\sigma)\rightarrow^*(S'_p,\sigma')\), \(\sigma''\in\Sigma\) with \(\sigma=_{|_{\mathcal{L}(S_p,V)\cup\mathcal{L}(S'_p,V)}}\sigma''\), and \((S'_p,\sigma'')\rightarrow^*(E,\sigma''')\in ex(S'_p)\), then \(\sigma'=_{|_{\mathcal{L}(E,V)\cup\mathcal{L}(E,V)}}\sigma'''\). 

\textbf{Base case} (m=0):
Since \(m=0\), we conclude that \(S_p=E\). 
Since \(S'_p=\Gamma(S_p,\gamma)=\Gamma(E,\gamma)\), we conclude that \(S'_p=E\).
Hence, \(\sigma'=\sigma\wedge \sigma''=\sigma'''\).
By assumption \(\sigma=_{|_{\mathcal{L}(S_p,V)\cup\mathcal{L}(S_p',V)}}\sigma''\).
Thus, the induction hypothesis follows.

\textbf{Step case} (\(m>0\)):
Let \((S_0,\sigma_0)\stackrel{op_1}{\rightarrow}\dots\stackrel{op_i}{\rightarrow}(S_i,\sigma_i)\) be the first segment and  \(\sigma'_0\in\Sigma\) be arbitrary such that \(\sigma_0=_{|_{\mathcal{L}(S_p,V)\cup\mathcal{L}(S'_p,V)}}\sigma'_0\) and assume \((S'_p,\sigma'')\rightarrow^*(E,\sigma''')\in ex(S'_p)\).
We know that \(S_p=S_0\) and \(\sigma_0=\sigma\).
Consider two cases.

First, assume that the first segment represents a sequence in which each program of the first i-1 states does not start with a replaced subprogram.
Due to Corollary~\ref{cor:stepsNotRep}, there exists execution \((\Gamma(S_0,\gamma),\sigma'_0)\stackrel{op_1}{\rightarrow}\dots\stackrel{op_i}{\rightarrow}(\Gamma(S_i,\gamma),\sigma'_i)\) with \(\sigma_i=_{|_{\mathcal{L}(S_i, V)\cup\mathcal{L}(\Gamma(S_i,\gamma),V)}}\sigma'_i\).
By assumption \(S_p'=\Gamma(S_0,\gamma)\).
By definition, \((S_i,\sigma_i)\rightarrow^*(E,\sigma')\in ex(S_i)\), which consists of \(m-1\) segments and is reachable from \(S\).
Due to semantics, semantics being deterministic, and \((S'_p,\sigma'_0)\rightarrow^*(E,\sigma''')\in ex(S'_p)\), there exists \(((\Gamma(S_i,\gamma),\sigma'_i)\rightarrow^*(E,\sigma''')\in ex(\Gamma(S_i,\gamma))\).
Furthermore, since \(S'_p\) reachable from \(S'\), \(S_p'=\Gamma(S_0,\gamma)\), and \((\Gamma(S_0,\gamma),\sigma'_0)\rightarrow^*(\Gamma(S_i,\gamma),\sigma'_i)\), also \(\Gamma(S_i,\gamma)\) reachable from \(S'\).
By induction, \(\sigma'=_{|_{\mathcal{L}(E,V)\cup\mathcal{L}(E,V)}}\sigma'''\).
The induction hypothesis follows. 

Second, assume that the first segment is the execution of a subprogram \(S_p^s\) that will be replaced, i.e., \(S_p^s\in dom(\gamma)\) and \(S_p=S_0=S_p^s\wedge S_i=E\vee S_p=S_0=S_p^s;S_i\).
Furthermore, from \(S'_p=\Gamma(S_p,\gamma)\), we conclude that \(S'_p=\gamma(S_p^s)=\Gamma(S_p,\gamma)\) if \(S_p=S_0=S_p^s\) and  \(S'_p=\Gamma(S_p,\gamma)=\Gamma(S_p^s,\gamma);\Gamma(S_i,\gamma)=\gamma(S_p^s);\Gamma(S_i,\gamma)\)  otherwise.
Due to semantics, semantics being deterministic, and \((S'_p,\sigma'_0)\rightarrow^*(E,\sigma''')\in ex(S'_p)\), there exists \((S'_p,\sigma'_0)\rightarrow^*(S'_i,\sigma'_i)\in ex(S'_p)\) with \(S'_i=\Gamma(S_i,\gamma)\).
Furthermore, there exists \(S'_i\rightarrow^*(E,\sigma''')\in ex(S'_i)\).
We conclude that \((\Gamma(S_p^s,\gamma),\sigma'_0)\rightarrow^*(E,\sigma'_i)\) (semantics).
Due to the definitions of \(\mathcal{UB}\) and \(\mathcal{L}\), \(\mathcal{UB}(S_p^s)\cup\mathcal{UB}(\gamma(S_p^s))\subseteq\mathcal{UB}(S_p)\cup\mathcal{UB}(S'_p)\subseteq\mathcal{L}(S_p)\cup\mathcal{L}(S'_p)\).
Due to Lemma~\ref{lem:equivCheck}, \(\sigma_i=_{|_{((\mathcal{L}(S_p,V)\cup\mathcal{L}(S'_p,V))\setminus(\mathcal{M}(S_p^s)\cup\mathcal{M}(\gamma(S_p^s))))\cup((\mathcal{M}(S_p^s)\cup\mathcal{M}(\gamma(S_p^s))\cap(L_1\cup L_2)))}}\sigma'_i\).
By definition and \(S_p\) reachable from \(S\), \(\mathcal{L}(S_i,V)\subseteq\mathcal{L}(S, S_p^s, V)\subseteq L_1\).
Similarly,  \(\mathcal{L}(\Gamma(S_i,\gamma),V)\subseteq\mathcal{L}(S', \gamma(S_p^s), V)\subseteq L_2\).
We conclude from the definition of live variables that \(\sigma_i=_{|_{(\mathcal{L}(S_i,V)\cup\mathcal{L}(S'_i,V))}}\sigma'_i\).
By definition, \((S_i,\sigma_i)\rightarrow^*(E,\sigma')\in ex(S_i)\), which consists of \(m-1\) segments and is reachable from \(S\).
Furthermore, we can conclude from \(S'_p\) reachable from \(S'\), also \(\Gamma(S_i,\gamma)\) reachable from \(S'\).
By induction, \(\sigma'=_{|_{\mathcal{L}(E,V)\cup\mathcal{L}(E,V)}}\sigma'''\).
The induction hypothesis follows.

Let \(\sigma\in\Sigma\), \((S,\sigma)\rightarrow^*(E,\sigma')\in ex(S)\), \((S',\sigma)\rightarrow^*(E,\sigma'')\in ex(S')\), and \(u\in V\) be arbitrary.
Since \(\sigma=_{|_{\mathcal{L}(S,V)\cup\mathcal{L}(S',V)}}\sigma\), \(S'=\Gamma(S,\gamma)\), and \((S',\sigma)\rightarrow^*(E,\sigma'')\in ex(S')\), the induction hypothesis gives us \(\sigma'=_{|_{\mathcal{L}(E,V)\cup\mathcal{L}(E,V)}}\sigma''\).
By definition of live variable analysis, \(u\in\mathcal{L}(E,V)\).
Hence, \(\sigma'(u)=\sigma''(u)\).
\end{proof}

\subsection{Proof of Theorem~\ref{theo:EquivApprox}}\label{ssec:proof:theo:EquivApprox}

\begin{lemma}\label{lem:stepsNotRep2}
Let \(S\) and \(S'\) be two programs, \(\gamma\) be a replacement function such that \(S'=\Gamma(S,\gamma)\), and  \(V\subseteq\mathcal{V}\) be a set of outputs.
If \(\neg\exists S^s_1, S^s_2: S^s_1\in dom(\gamma)\wedge (S=S^s_1;S^s_2 \vee S=S^s_1)\), then \(\forall \sigma, \sigma'\in\Sigma: (S,\sigma)\stackrel{op_1}{\rightarrow}(S_1,\sigma_1)\in ex(S)\wedge \sigma=_{|_{\mathcal{L}(S', V)}}\sigma' \implies \exists (S',\sigma')\stackrel{op_1}{\rightarrow}(S'_1,\sigma'_1): S'_1=\Gamma(S_1,\gamma) \wedge  \sigma_1=_{|_{\mathcal{L}(S'_1,V)}}\sigma'_1\).
\end{lemma}

\begin{proof}
Let \(S\) and \(S'\) be two programs, \(\gamma\) be a replacement function such that \(S'=\Gamma(S,\gamma)\) and \(\neg\exists S^s_1, S^s_2: S^s_1\in dom(\gamma)\wedge  (S=S^s_1;S^s_2 \vee S=S^s_1)\),  \(V\subseteq\mathcal{V}\) be a set of outputs and \(\sigma, \sigma'\in\Sigma\) with \(\sigma=_{|_{\mathcal{L}(S',V)}}\sigma'\).

Assume \(p=(S,\sigma)\stackrel{op_1}{\rightarrow}(S_1,\sigma_1)\in ex(S)\).
By definition, \(\mathcal{UB}(S')\subseteq \mathcal{L}(S',V)\).
Consider two cases.

Case 1 (\(S_1=E\)):
Due to the semantics, we conclude that either \(S=E;E\), or \(S\) is  not a sequential composition, but a statement.
If \(S\) is an assignment, an assertion or a parallel statement, we conclude from  \(S\notin dom(\gamma)\), \(S'=\Gamma(S,\gamma)\), and \(\gamma\) does not replace statements in parallel statements that \(S'=S\). 
Similarly, if \(S=E;E\), also \(S=S'\). 
Due to Lemma~\ref{lem:ubBehave}, \(\exists (S',\sigma')\stackrel{op_1}{\rightarrow}(E,\sigma'_1)\) and \(\sigma_1=_{|_{\{v\in\mathcal{V}\mid \sigma(v)=\sigma'(v)\vee \exists 1\leq i\leq n: op_i\equiv v:=expr;\}}} \sigma'_1\).
Due to the definition of live variables, we conclude that \(\sigma_1=_{|_{\mathcal{L}(E, V)}}\sigma'_1\).
If \(S\) is an if- or while-statement, we conclude from  \(S\notin dom(\gamma)\) and \(S'=\Gamma(S,\gamma)\) that \(S'\) is an if-/while-statement and the condition is the same.
Due to semantics, definition of live variable analysis, and \(\sigma=_{|_{\mathcal{L}(S', V)}}\sigma'\), we then conclude that \((S',\sigma')\stackrel{op_1}{\rightarrow}(E,\sigma')\in ex(S')\) and \(\sigma_1=\sigma\).
Due to the definition of live variables, we conclude that \(\sigma_1=_{|_{\mathcal{L}(E, V)}}\sigma'\).
By definition, \(E=\Gamma(E,\gamma)\).

Case 2 (\(S_1\neq E\)):
Since replacements do not occur in parallel statements and \(\neg\exists S^s_1, S^s_2: S^s_1\in dom(\gamma)\wedge  (S=S^s_1;S^s_2 \vee S=S^s_1)\), we conclude that \(\exists S^s_1, S^s_2, S^s_3, S^s_4:  S=S^s_1;S^s_2 \wedge S'=S^s_3;S^s_4\wedge S^s_3=\Gamma(S^s_1, \gamma) \wedge S^s_4=\Gamma(S^s_2, \gamma)\vee S=S^s_1\wedge S'=S^s_3\wedge S^s_3=\Gamma(S^s_1, \gamma)\) and either \(S^s_1=S^s_3\) or \(S^s_1\) and \(S^s_3\) are either both if- or both-while statements with the same condition and the if/else-body, the loop body of $S^s_3$ is a replacement of the body of $S^s_1$. 
First, consider the first case (\(S^s_1=S^s_3\)).
Due to semantics, either (1)~\(S_1=E;S^s_2\) and \((S^s_1, \sigma)\stackrel{op_1}{\rightarrow}(E,\sigma_1)\), (2)~\(S_1=S^s_2\wedge S^s_1=E\), \(\sigma=\sigma_1\), and \((S, \sigma)\stackrel{\textbf{nop}}{\rightarrow}(S_1,\sigma_1)\), or (3)~\(S_1=S^s_5;S^s_2\) and \((S^s_1, \sigma)\stackrel{op_1}{\rightarrow}(S^s_5,\sigma_1)\).
Due to Lemma~\ref{lem:ubBehave}, in case~(1) \(\exists (S^s_1,\sigma')\stackrel{op_1}{\rightarrow}(E,\sigma'_1)\) and \(\sigma'=\sigma'_1\), and in case~(3)  \(\exists (S^s_1,\sigma')\stackrel{op_1}{\rightarrow}(S^s_5,\sigma'_1)\).
Furthermore, \(\sigma_1=_{|_{\{v\in\mathcal{V}\mid \sigma(v)=\sigma'(v)\vee op_1\equiv v:=expr;\}}} \sigma'_1\).
Due to semantics, in case~(1) \(\exists (S',\sigma')\stackrel{op_1}{\rightarrow}(E;S^s_4,\sigma'_1)\), in case~(2) \((S', \sigma')\stackrel{\textbf{nop}}{\rightarrow}(S^s_4,\sigma')\), and in case~(3)
  \(\exists (S',\sigma')\stackrel{op_1}{\rightarrow}(S^s_5;S^s_4,\sigma'_1)\).
Since \(S^s_1=S^s_3=\Gamma(S^s_3,\gamma)\), \(\gamma\) is only defined for subprograms of \(S\) and statements (thus, subprograms) can be uniquely identified via labels, we get \(\Gamma(S^s_5,\gamma)=S^s_5\).
Hence, \(\Gamma(S^s_5;S^s_2,\gamma)=S^s_5;\Gamma(S^s_2,\gamma)=S^s_5;S^s_4\).
Similarly, \(\Gamma(E;S^s_2,\gamma)=E;\Gamma(S^s_2,\gamma)=S^s_5;S^s_4\).
Moreover, \(\sigma_1=_{|_{\{v\in\mathcal{V}\mid \sigma(v)=\sigma'(v)\vee op_1\equiv v:=expr;\}}} \sigma'_1\), \(\sigma=_{|_{\mathcal{L}(S',V)}}\sigma'\), and the definition of live variable analyses let us conclude that \(\sigma_1=_{|_{\mathcal{L}(S'_1, V)}}\sigma'_1\). % with \(S'_1=S^s_4\) and \(S'_1=S^s_6;S^s_4\), respectively.

Second, consider that (\(S^s_1\neq S^s_3\)).
We know that \(S^s_1\) and \(S^s_3\) are either both if- or both-while statements with the same condition and the if/else-body, the loop body of $S^s_3$ is a replacement of the body of $S^s_1$.
Due to semantics, definition of live variable analysis, \(\sigma=_{|_{\mathcal{L}(S',V)}}\sigma'\), and the replacement function, we then conclude that \(\sigma=\sigma_1\) and \(\sigma'=\sigma'_1\) and either \(S=S^s_1\) and exists \((S',\sigma')\stackrel{op_1}{\rightarrow}(S'_1,\sigma_1)\in ex(S')\) with \(S'_1=\Gamma(S_1,\gamma)\) (due to \(S; \textbf{while}~expr~\textbf{do}~S\) is no subprogram of S) or \(S=S^s_1;S^s_2\) and \((S',\sigma')\stackrel{op_1}{\rightarrow}(S'_1,\sigma_1)\in ex(S')\) with \(S'_1=\Gamma(S_1,\gamma)\).
Due to the definition of live variables, we conclude that \(\sigma_1=_{|_{\mathcal{L}(S'_1, V)}}\sigma'_1\).
\end{proof}

\begin{corollary}\label{cor:stepsNotRep2}
Let \(S\) and \(S'\) be two programs, \(\gamma\) be a replacement function such that \(S'=\Gamma(S,\gamma)\), and  \(V\subseteq\mathcal{V}\) be a set of outputs.
For all \((S_0,\sigma_0)\stackrel{op_1}{\rightarrow}\dots\stackrel{op_n}{\rightarrow}(S_n,\sigma_n)\in ex(S)\), if for all \(0\leq i<n\)  not exists \(S^{s1}_i, S^{s2}_i \) such that \((S_i=S^{s1}_i;S^{s2}_i \vee S_i=S^{s1}_i)\) and  \(S^{s1}_i \in dom(\gamma)\), then \(\forall \sigma'_0\in\Sigma: \sigma_0=_{|_{\mathcal{L}(S', V)}}\sigma'_0 \implies \exists (\Gamma(S_0,\gamma),\sigma'_0)\stackrel{op_1}{\rightarrow}\dots\stackrel{op_n}{\rightarrow}(\Gamma(S_n,\gamma),\sigma'_n)\in ex(S'):\) \mbox{\(\forall 0\leq i\leq n: \sigma_i=_{|_{\mathcal{L}(\Gamma(S_i,\gamma),V)}}\sigma'_i\)}.
\end{corollary}

\begin{proof}
Proof by induction.

\textbf{Base case} (i=0):
By definition \((S',\sigma)=(\Gamma(S,\gamma),\sigma)\in ex(S')\) for arbitrary \(\sigma\in\Sigma\) (including all \(\sigma'_0\) with \(\sigma_0=_{|_{\mathcal{L}(S',V)}}\sigma'_0\)).

\textbf{Step case} (\(n-1\rightarrow n\)):
Due to Lemma~\ref{lem:stepsNotRep2}, there exists \((\Gamma(S,\gamma),\sigma')\stackrel{op_1}{\rightarrow}(\Gamma(S_1,\gamma),\sigma'_1)\) with \(\sigma_1=_{|_{\mathcal{L}(S'_1,V)}}\sigma'_1\).
By induction,  \((\Gamma(S_1,\gamma),\sigma'_1)\stackrel{op_2}{\rightarrow}\dots\stackrel{op_n}{\rightarrow}(\Gamma(S_n,\gamma),\sigma'_n)\in ex(S'): \forall 1\leq i\leq n: \sigma_i=_{|_{\mathcal{L}(\Gamma(S_i,\gamma),V)}}\sigma'_i\).
By definition, the induction hypothesis follows.
\end{proof}

\begin{lemma}\label{lem:equivCheck3}
Let \(S_1\) and \(S_2\) be two (sub)programs of programs \(S\) and \(S'\), respectively.
Consider arbitrary \((S,\sigma)\rightarrow^*(S_i,\sigma_i)\rightarrow^*(S_j,\sigma_j)\in ex(S)\) and  \((S',\sigma')\rightarrow^*(S'_i,\sigma'_i)\rightarrow^*(S'_j,\sigma'_j)\in ex(S')\) such that \(S_i=S_1\wedge S_j=E \vee S_i=S_1;S_j\), \(S'_i=S_2\wedge S'_j=E \vee S'_i=S_2;S'_j\), \((S_1,\sigma_i)\rightarrow^*(E,\sigma_j)\in ex(S_1)\), and \((S_2,\sigma'_i)\rightarrow^*(E,\sigma'_j)\in ex(S_2)\).
Given overapproximations \(\mathcal{UB}(S_1)\subseteq U_1\subseteq \mathcal{V}(S_1)\), \(\mathcal{UB}(S_2)\subseteq U_2\subseteq \mathcal{V}(S_2)\) and overapproximations \(\mathcal{M}(S_1)\subseteq M_1\subseteq \mathcal{V}(S_1)\) %\(\mathcal{M}(S_1)\cup\{v\in\mathcal{V}\mid\exists S_1\rightarrow^*S_k\stackrel{v:=aexpr}{\rightarrow}S_r\in syn_P(S_1)\}\subseteq M_1\subseteq \mathcal{V}(S_1)\) 
and \(\mathcal{M}(S_2)\cup\{v\in\mathcal{V}\mid\exists S_2\rightarrow^*S'_k\stackrel{v:=aexpr}{\rightarrow}S'_r\in syn_P(S_2)\}\subseteq M_2\subseteq \mathcal{V}(S_2)\) of the modified variables, a renaming function \(\rho_\mathrm{switch}\), and \((\mathcal{L}(S_2, S', V))\cap(M_1\cup M_2)\subseteq C\subseteq M_1\cup M_2\).
If all executions \((eq\_task(S_1, S_2, \rho_\mathrm{switch},(U_1\cap U_2)\cap(M_1\cup M_2) ,C),\sigma)\rightarrow^*(S',\sigma')\in ex(eq\_task(S_1, S_2, \rho_\mathrm{switch}, (U_1\cap U_2)\cap(M_1\cup M_2) ,C))\) do not violate an assertion, \(V\subseteq\mathcal{V}\), and \(\sigma_i=_{|_{\mathcal{L}(S'_i,V)}} \sigma'_i\), then \(\sigma_j=_{|_{\mathcal{L}(S'_j,V)}}\sigma'_j\).
\end{lemma}

\begin{proof}
Let  \(u\in \mathcal{L}(S'_j,V)\) be arbitrary.
In the following, we write \(v\) is assigned in \(p=(S_0,\sigma_0)\stackrel{op_1}{\rightarrow}\dots\stackrel{op_n}{\rightarrow}(S_n,\sigma_n)\in ex(S)\) if \(\exists i\in[1,n]: op_i=v:=aexpr;\).

First, consider \(u\notin(\mathcal{M}(S_1)\cup\mathcal{M}(S_2))\) and  \(\sigma_i(u)=\sigma_j(u)\).
Due to Lemma~\ref{lem:change}, \(\sigma_i(u)=\sigma_j(u)\) and \(\sigma'_i(u)=\sigma'_j(u)\).
Hence, \(\sigma_j(u)=\sigma'_j(u)\).

Second, consider \(u\in (\mathcal{M}(S_1)\cup\mathcal{M}(S_2))\) or \(\sigma_i(u)\neq\sigma_j(u)\).
Let \(\sigma_r\in\Sigma\) with \(\sigma_r=_{|_{(\mathcal{L}(S'_i,V))\cup\mathcal{M}(S_1)\cup\mathcal{M}(S_2)}}\sigma'_i\) and \(\sigma_r=_{|_{\bigcup_{v\in(\mathcal{V}(S_1)\cup\mathcal{M}(S_2))}\rho_\mathrm{switch}(v)}}\rho_\mathrm{switch}(\sigma_i)\).
Due to definition of \(\rho_\mathrm{switch}\) such a data state exists. 

Due to semantics and Lemma~\ref{lem:init}, \((init(\rho_\mathrm{switch},(U_1\cap U_2)\cap(M_1\cup M_2)), \sigma_r)\rightarrow^* (E,\sigma_\mathrm{init})\) with \(\forall v\in (U_1\cap U_2)\cap(M_1\cup M_2):\sigma_\mathrm{init}(v)=\sigma_\mathrm{init}(\rho_\mathrm{switch}(v))=\sigma_r(\rho_\mathrm{switch}(v))\).
By construction of \(\sigma_r\) and \(\rho_\mathrm{switch}\), and \(\sigma_i=_{|_{\mathcal{L}(S'_i,V)}}\sigma'_i\), we get \(\forall v\in (\mathcal{L}(S'_i,V)):\sigma_r(\rho_\mathrm{switch}(v))=\rho_\mathrm{switch}(\sigma_i)(\rho_\mathrm{switch}(v))=\sigma_i(\rho_\mathrm{switch}^{-1}(\rho_\mathrm{switch}(v)))=\sigma_i(v)=\sigma'_i(v)\).
Due to Lemma~\ref{lem:change}  and modifications \(\mathcal{M}((init(\rho_\mathrm{switch},(U_1\cup U_2)\cap(M_1\cup M_2))))\subseteq (U_1\cap U_2)\cap(M_1\cup M_2)\), we infer that \(\forall v\in \mathcal{V}\setminus((U_1\cap U_2)\cap(M_1\cup M_2)): \sigma_\mathrm{init}(v)=\sigma_r(v)\).
Hence, \(\sigma'_i=_{|_{\mathcal{L}(S'_i,V)}}\sigma_\mathrm{init}\) and \(\sigma_{init}=_{|_{\bigcup_{v\in(\mathcal{V}(S_1)\cup\mathcal{M}(S_2))}\rho_\mathrm{switch}(v)}}\rho_\mathrm{switch}(\sigma_i)\).

Due to Lemma~\ref{lem:biRen}, there exists \((\mathcal{R}(S_1,\rho_\mathrm{switch}),\rho_\mathrm{switch}(\sigma_i))\rightarrow^*(E,\rho_\mathrm{switch}(\sigma_j))\) \(\in ex(\mathcal{R}(S_1,\rho_\mathrm{switch}))\).
%By definition of \(\sigma_r\), \(\sigma_r=_{|_{\mathcal{V}\setminus((U_1\cap U_2)\cap(M_1\cup M_2))}}\sigma_\mathrm{init}\) as well as \(\sigma'_i=_{|_{\mathcal{L}(S'_i,V)}}\sigma_\mathrm{init}\),
Since \(\sigma'_i=_{|_{\mathcal{L}(S'_i,V)}}\sigma_\mathrm{init}\), by definition \(\mathcal{UB}(S_2)\subseteq\mathcal{L}(S'_i,V)\), and \(\sigma_{init}=_{|_{\bigcup_{v\in(\mathcal{V}(S_1)\cup\mathcal{M}(S_2))}\rho_\mathrm{switch}(v)}}\rho_\mathrm{switch}(\sigma_i)\), we infer from Lemma~\ref{lem:ubBehave} that \mbox{\(\exists p'=(S_2,\sigma_\mathrm{init})\rightarrow^*(E,\sigma''_r)\in ex(S_2)\)} with \(\sigma'_j=_{|_{\mathcal{L}(S'_i,V))\cup\{v\in\mathcal{V}\mid v~\textrm{~assigned in~}p'\}}} \sigma''_r\) and \(\exists p=(\mathcal{R}(S_1,\rho_\mathrm{switch}),\sigma_\mathrm{init})\rightarrow^*(E,\sigma'_r)\in ex(\mathcal{R}(S_1,\rho_\mathrm{switch}))\) such that \(\rho_\mathrm{switch}(\sigma_j)\) \(=_{|_{\bigcup_{v\in(\mathcal{V}(S_1)\cup\mathcal{M}(S_2))}\rho_\mathrm{switch}(v)}} \sigma'_r\).
Due to Lemma~\ref{lem:noInf}, \(\exists(\mathcal{R}(S_1,\rho_\mathrm{switch});S_2,\sigma_\mathrm{init})\rightarrow^*(E,\sigma_c)\in ex(\mathcal{R}(S_1,\rho_\mathrm{switch});S_2)\) with  \(\sigma_c=_{|_{\mathcal{V}(\mathcal{R}(S_1,\rho_\mathrm{switch}))\cup\bigcup_{v\in\mathcal{M}(S_2)}\rho_\mathrm{switch}(v)}}\sigma'_r\) and \(\sigma_c=_{|_{\mathcal{V}(S_2)\cup\mathcal{M}(S_1)}}\sigma''_r\).

Due to semantics and all \((eq\_task(S_1, S_2, \rho_\mathrm{switch}, U_1\cap U_2\cap(M_1\cup M_2) ,C),\sigma)\rightarrow^*(S',\sigma')\in ex(eq\_task(S_1, S_2, \rho_\mathrm{switch}, (U_1\cap U_2)\cap(M_1\cup M_2) ,C))\) do not violate assertions, there exists \((equal(\rho_\mathrm{switch}, C), \sigma_c)\rightarrow^*(E,\sigma'_c)\).
Due to Lemma~\ref{lem:equiv}, we infer for all \(v\in C\) that \(\sigma_c(v)=\sigma_c(\rho_\mathrm{switch}(v))\).

Distinguish two cases.
First, consider \(u\in (\mathcal{M}(S_1)\cup\mathcal{M}(S_2))\).
We conclude that \(u\in C\).
Hence, \(\sigma_c(u)=\sigma_c(\rho_\mathrm{switch}(u))\).
By definition of \(\mathcal{L}\), we conclude that \(u\in\mathcal{L}(S'_i,V)\) or \(u\) is assigned on \((S_2,\sigma'_i)\rightarrow^*(E,\sigma'_j)\in ex(S_2)\).
Due to Lemma~\ref{lem:ubBehave}, we infer that \(u\in\mathcal{L}(S'_i,V)\) or \(u\) is assigned on \(p'\).
We conclude that \(\sigma_c(u)=\sigma''_r(u)=\sigma'_j(u)\) and \(\sigma_c(\rho_\mathrm{switch}(u))=\sigma'_r(\rho_\mathrm{switch}(u)) = \rho_\mathrm{switch}(\sigma_j)(\rho_\mathrm{switch}(u))=\sigma_j(\rho_\mathrm{switch}^{-1}(\rho_\mathrm{switch}(u))\) \(=\sigma_j(u)\).
Since \(\sigma_c(u)=\sigma_c(\rho_\mathrm{switch}(u))\), we get \(\sigma'_j(u)=\sigma_j(u)\).

Second, consider \(\sigma_j(u)\neq\sigma'_j(u)\) and \(u\notin (\mathcal{M}(S_1)\cup\mathcal{M}(S_2))\).
Since \(u\in\mathcal{L}(S'_j,V)\) and \(\sigma_i(u)\neq\sigma_j(u)\), we conclude that \(u\notin\mathcal{L}(S'_i,V)\).
By definition of \(\mathcal{L}\), we conclude that \(u\) is assigned on \((S_2,\sigma'_i)\rightarrow^*(E,\sigma'_j)\in ex(S_2)\).
Hence, \(u\in M_2\) and, therefore, \(u\in C\).
Then, due to Lemma~\ref{lem:ubBehave}, \(u\) is assigned on \(p'\).
Furthermore, we conclude that \(\sigma_c(u)=\sigma''_r(u)=\sigma'_j(u)\) and \(\sigma_c(\rho_\mathrm{switch}(u))=\sigma'_r(\rho_\mathrm{switch}(u)) = \rho_\mathrm{switch}(\sigma_j)(\rho_\mathrm{switch}(u))=\sigma_j(\rho_\mathrm{switch}^{-1}(\rho_\mathrm{switch}(u))=\sigma_j(u)\).
Since \(u\in C\), we conclude that \(\sigma'_j(u)=\sigma_c(u)=\sigma_c(\rho_\mathrm{switch}(u))=\sigma_j(u)\).
\end{proof}

\begin{theorem3b}%[\ref{theo:EquivApprox}]
Let \(S\) and \(S'\) be two programs, \(\gamma\) be a replacement function such that \(S'=\Gamma(S,\gamma)\), and  \(V\subseteq\mathcal{V}\) be a set of outputs.
If for all \((S_1, S_2)\in\gamma\) there exists overapproximations \(\mathcal{UB}(S_1)\subseteq U_1\subseteq \mathcal{V}(S_1)\), \(\mathcal{UB}(S_2)\subseteq U_2\subseteq \mathcal{V}(S_2)\),
 \(\mathcal{M}(S_1)\subseteq M_1\subseteq \mathcal{V}(S_1)\), %\(\mathcal{M}(S_1)\cup\{v\in\mathcal{V}\mid\exists S_1\rightarrow^*S_k\stackrel{v:=aexpr}{\rightarrow}S_r\in syn_P(S_1)\}\subseteq M_1\subseteq \mathcal{V}(S_1)\), 
\(\mathcal{M}(S_2)\cup\{v\in\mathcal{V}\mid\exists S_2\rightarrow^*S_k\stackrel{v:=aexpr}{\rightarrow}S_r\in syn_P(S_2)\}\subseteq M_2\subseteq \mathcal{V}(S_2)\), 
 \(\mathcal{L}(S_1, S, V)\subseteq L_1\subseteq \mathcal{V}\), \(\mathcal{L}(S_2, S', V)\subseteq L_2\subseteq \mathcal{V}\), and renaming \(\rho_\mathrm{switch}\) s.t.\ \mbox{\(eq\_task(S_1, S_2, \rho_\mathrm{switch}, (U_1\cap U_2)\cap(M_1\cup M_2), (M_1\cup M_2)\cap (L_1\cup L_2))\)} does not violate an assertion, then \(S\equiv_V S'\).
\end{theorem3b}

\begin{proof}
Consider \((S_p,\sigma)\rightarrow^*(E,\sigma')=(S_0,\sigma_0)\stackrel{op_1}{\rightarrow}\dots\stackrel{op_n}{\rightarrow}(S_n,\sigma_n)\) be a path for an arbitrary program \(S_p\).
We define the splitting of the path into \(m\geq0\) segments such that each segment represents either a sequence in which each program of the sequence's states except for the last one does not start with a replaced subprogram or the execution of the subprogram that will be replaced. In case that there exists multiple replacements (nesting of replaced subprograms), we use the largest replacement.
Show by induction over the number of segments that for all programs \(S_p\) such that \(\exists\sigma,\sigma'\in\Sigma: (S,\sigma)\rightarrow^*(S_p,\sigma')\) if \((S_p,\sigma)\rightarrow^*(E,\sigma')\in ex(S_p)\), \(S'_p=\Gamma(S_p,\gamma)\), \(\exists\sigma,\sigma'\in\Sigma: (S',\sigma)\rightarrow^*(S'_p,\sigma')\), \(\sigma''\in\Sigma\) with \(\sigma=_{|_{\mathcal{L}(S'_p,V)}}\sigma''\), and \((S'_p,\sigma'')\rightarrow^*(E,\sigma''')\in ex(S'_p)\), then \(\sigma'=_{|_{\mathcal{L}(E,V)}}\sigma'''\). 

\textbf{Base case} (m=0):
Since \(m=0\), we conclude that \(S_p=E\). 
Since \(S'_p=\Gamma(S_p,\gamma)=\Gamma(E,\gamma)\), we conclude that \(S'_p=E\).
Hence, \(\sigma'=\sigma\wedge \sigma''=\sigma'''\).
By assumption \(\sigma=_{|_{\mathcal{L}(S_p',V)}}\sigma''\).
Thus, the induction hypothesis follows.

\textbf{Step case} (\(m>0\)):
Let \((S_0,\sigma_0)\stackrel{op_1}{\rightarrow}\dots\stackrel{op_i}{\rightarrow}(S_i,\sigma_i)\) be the first segment and  \(\sigma'_0\in\Sigma\) be arbitrary such that \(\sigma_0=_{|_{\mathcal{L}(S'_p,V)}}\sigma'_0\) and assume \((S'_p,\sigma'')\rightarrow^*(E,\sigma''')\in ex(S'_p)\).
We know that \(S_p=S_0\) and \(\sigma_0=\sigma\).
Consider two cases.

First, assume that the first segment represents a sequence in which each program of the first i-1 states does not start with a replaced subprogram.
Due to Corollary~\ref{cor:stepsNotRep2}, there exists execution \((\Gamma(S_0,\gamma),\sigma'_0)\stackrel{op_1}{\rightarrow}\dots\stackrel{op_i}{\rightarrow}(\Gamma(S_i,\gamma),\sigma'_i)\) with \(\sigma_i=_{|_{\mathcal{L}(\Gamma(S_i,\gamma),V)}}\sigma'_i\).
By assumption \(S_p'=\Gamma(S_0,\gamma)\).
By definition, \((S_i,\sigma_i)\rightarrow^*(E,\sigma')\in ex(S_i)\), which consists of \(m-1\) segments and is reachable from \(S\).
Due to semantics, semantics being deterministic, and \((S'_p,\sigma'_0)\rightarrow^*(E,\sigma''')\in ex(S'_p)\), there exists \(((\Gamma(S_i,\gamma),\sigma'_i)\rightarrow^*(E,\sigma''')\in ex(\Gamma(S_i,\gamma))\).
Furthermore, since \(S'_p\) reachable from \(S'\), \(S_p'=\Gamma(S_0,\gamma)\), and \((\Gamma(S_0,\gamma),\sigma'_0)\rightarrow^*(\Gamma(S_i,\gamma),\sigma'_i)\), also \(\Gamma(S_i,\gamma)\) reachable from \(S'\).
By induction, \(\sigma'=_{|_{\mathcal{L}(E,V)}}\sigma'''\).
The induction hypothesis follows. 

Second, assume that the first segment is the execution of a subprogram \(S_p^s\) that will be replaced, i.e., \(S_p^s\in dom(\gamma)\) and \(S_p=S_0=S_p^s\wedge S_i=E\vee S_p=S_0=S_p^s;S_i\wedge (S_p^s,\sigma_0)\rightarrow^*(E,\sigma_i)\in ex(S_p^s)\).
Furthermore, from \(S'_p=\Gamma(S_p,\gamma)\), we conclude that \(S'_p=\gamma(S_p^s)=\Gamma(S_p,\gamma)\) if \(S_p=S_0=S_p^s\) and  \(S'_p=\Gamma(S_p,\gamma)=\Gamma(S_p^s,\gamma);\Gamma(S_i,\gamma)=\gamma(S_p^s);\Gamma(S_i,\gamma)\)  otherwise.
Due to semantics, semantics being deterministic, and \((S'_p,\sigma'_0)\rightarrow^*(E,\sigma''')\in ex(S'_p)\), there exists \((S'_p,\sigma'_0)\rightarrow^*(S'_i,\sigma'_i)\in ex(S'_p)\) with \(S'_i=\Gamma(S_i,\gamma)\) and \((\gamma(S_p^s),\sigma'_0)\rightarrow^*(E,\sigma'_i)\).
Furthermore, there exists \(S'_i\rightarrow^*(E,\sigma''')\in ex(S'_i)\).
We conclude that \((\Gamma(S_p^s,\gamma),\sigma'_0)\rightarrow^*(E,\sigma'_i)\) (semantics).
%Due to the definitions of \(\mathcal{UB}\) and \(\mathcal{L}\), \(\mathcal{UB}(S_p^s)\cup\mathcal{UB}(\gamma(S_p^s))\subseteq\mathcal{UB}(S_p)\cup\mathcal{UB}(S'_p)\subseteq\mathcal{L}(S_p)\cup\mathcal{L}(S'_p)\).
Due to Lemma~\ref{lem:equivCheck3}, \(\sigma_i=_{|_{\mathcal{L}(S'_i,V)}}\sigma'_i\).
By definition, \((S_i,\sigma_i)\rightarrow^*(E,\sigma')\in ex(S_i)\), which consists of \(m-1\) segments and is reachable from \(S\).
Furthermore, we can conclude from \(S'_p\) reachable from \(S'\), also \(\Gamma(S_i,\gamma)\) reachable from \(S'\).
By induction, \(\sigma'=_{|_{\mathcal{L}(E,V)}}\sigma'''\).
The induction hypothesis follows. 

Let \(\sigma\in\Sigma\), \((S,\sigma)\rightarrow^*(E,\sigma')\in ex(S)\), \((S',\sigma)\rightarrow^*(E,\sigma'')\in ex(S')\), and \(u\in V\) be arbitrary.
Since \(\sigma=_{|_{\mathcal{L}(S',V)}}\sigma\), \(S'=\Gamma(S,\gamma)\), and \((S',\sigma)\rightarrow^*(E,\sigma'')\in ex(S')\), the induction hypothesis gives us \(\sigma'=_{|_{\mathcal{L}(E,V)}}\sigma''\).
By definition of live variable analysis, \(u\in\mathcal{L}(E,V)\).
Hence, \(\sigma'(v)=\sigma''(v)\).
\end{proof}

\subsection{Correctness of \(\rho_\mathrm{switch}\)}

\begin{lemma}\label{lem:corrRenaming}
Let \(S_1\) and \(S_2\) be two (sub)programs. 
Given overapproximation 
 \(\mathcal{UB}(S_1)\) \(\subseteq U_1\subseteq \mathcal{V}(S_1)\) and \(\mathcal{UB}(S_2)\subseteq U_2\subseteq \mathcal{V}(S_2)\) of the variables used before definition and  overapproximations \(\mathcal{M}(S_1)\subseteq M_1\subseteq \mathcal{V}(S_1)\) and \(\mathcal{M}(S_2)\subseteq M_2\subseteq \mathcal{V}(S_2)\) of the modified variables.
Any function \(\rho_\mathrm{switch}\) is  appropriate for renaming and ensures \(\forall v\in (U_1\cap U_2)\cap(M_1\cup M_2): \rho_\mathrm{switch}(v)=v\vee \rho_\mathrm{switch}(v)\notin  (U_1\cap U_2)\cap(M_1\cup M_2)\).
\end{lemma}

\begin{proof}
Let \(\mathrm{switch}: M_1\cup M_2\rightarrow \mathcal{V}\setminus(\mathcal{V}(S_1)\cup\mathcal{V}(S_2))\) be an arbitrary injective function.
Due to injectivity of function \(\mathrm{switch}\) and the construction of \(\rho_\mathrm{switch}\), function \(\rho_\mathrm{switch}\) is bijective.

Since \(\mathcal{M}(S_2)\subseteq M_2\subseteq\mathcal{V}(S_2)\) and by definition of \(\rho_\mathrm{switch}\) for all \(v\in\mathcal{V}(S_1)\) either \(\rho_\mathrm{switch}(v)=v\) and \(v\notin M_1\cup M_2\) or \(\rho_\mathrm{switch}(v)\in\mathcal{V}\setminus(\mathcal{V}(S_1)\cup\mathcal{V}(S_2))\), we infer \(\forall v\in\mathcal{V}(S_1)\cup\mathcal{M}(S_2):\rho(v)\notin\mathcal{M}(S_2)\)

Since \(\mathcal{M}(S_1)\subseteq M_1\) and for all \(v\in M_1\cup M_2\) renamed variable \(\rho_\mathrm{switch}(v)\in\mathcal{V}\setminus(\mathcal{V}(S_1)\cup\mathcal{V}(S_2))\), we infer \(\forall v\in\mathcal{M}(S_1): \rho(v)\notin\mathcal{V}(S_2)\cup\mathcal{M}(S_1)\).

We conclude that \(\rho_\mathrm{switch}\) is appropriate for renaming.

By definition, \((U_1\cap U_2)\cap(M_1\cup M_2)\subseteq U_1\cap U_2\subseteq \mathcal{V}(S_1)\cup\mathcal{V}(S_2)\).
By construction of \(\rho_\mathrm{switch}\), for all \(v\in \mathcal{V}(S_1)\cup\mathcal{V}(S_2)\) either \(\rho_\mathrm{switch}(v)=v\) or \(\rho_\mathrm{switch}(v)\notin \mathcal{V}(S_1)\cup\mathcal{V}(S_2)\).
Hence, \(\forall v\in(U_1\cap U_2)\cap(M_1\cup M_2): \rho_\mathrm{switch}(v)=v\vee \rho_\mathrm{switch}(v)\notin  (U_1\cap U_2)\cap(M_1\cup M_2)\).
\end{proof}

\end{document}